\documentclass[arma]{svjour}

\usepackage{amsfonts}
\usepackage{amsmath,amssymb,amsrefs}
\usepackage{graphicx}
\usepackage[mathscr]{euscript}
\usepackage[shortlabels]{enumitem}
\usepackage{bm}
\usepackage{hyperref}
\usepackage{color}

\numberwithin{equation}{section}

\spnewtheorem{assumption}{Assumption}{\bf}{\it}

\newcommand{\subscript}[2]{$#1 _ #2$}
\newcommand{\bb}[1]{\mathbf{#1}}
\newcommand{\bs}[1]{\boldsymbol{#1}}
\newcommand{\inr}{\in \mathbb{R}^2}

\newcommand{\R}{\mathbb{R}}
\newcommand{\C}{\mathbb{C}}
\newcommand{\Z}{\mathbb{Z}}

\newcommand{\B}{\mathcal{B}}
\newcommand{\pc}{\mathcal{PC}}

\newcommand{\bk}{{\bf k}}
\newcommand{\bfm}{{\bf m}}

\newcommand{\bK}{{\bf K}}
\newcommand{\bKp}{{\bf K'}}
\newcommand{\bv}{{\bf v}}

\newcommand{\bx}{{\bf x}}
\newcommand{\bX}{{\bf X}}
\newcommand{\by}{{\bf y}}

\newcommand{\LA}{{\mathcal{L}^A}}
\newcommand{\LB}{{\mathcal{L}^B}}

\newcommand{\balpha}{{\bm{\alpha}}}

\newcommand{\vtilde}{{\bm{\mathfrak{v}}}}
\newcommand{\ktilde}{{\bm{\mathfrak{K}}}}
\newcommand{\smallktilde}{{\mathfrak{K}}}
\newcommand{\kpar}{{k_{\parallel}}}

\newcommand{\kparv}{{\kpar=\bK\cdot\vtilde_1}}
\newcommand{\kparpv}{{\kpar=\bKp\cdot\vtilde_1}}
\newcommand{\kparvstar}{{\kpar=\bK_\star\cdot\vtilde_1}}

\newcommand{\abs}[1]{\left\lvert#1\right\rvert}
\newcommand{\norm}[1]{\left\lVert#1\right\rVert}
\newcommand{\inner}[1]{\left\langle#1\right\rangle}
\newcommand{\D}{\partial}

\newcommand{\eps}{\varepsilon}
\newcommand{\nit}{\noindent}
\newcommand{\nn}{\nonumber}

\newcommand{\lamsharp}{{\lambda_\sharp}}
\newcommand{\thetasharp}{{\vartheta_{\sharp}}}
\newcommand{\exponent}{\nu}

\newcommand{\bkappa}{{\bs \kappa}}

\newcommand{\e}{a}

\newcommand{\intomega}[1]{{\int_\Omega #1 \,d\mathbf{x}}}

\begin{document}
\title{Elliptic operators with honeycomb symmetry: Dirac points, Edge States and Applications to Photonic Graphene}

\author{J.P. Lee-Thorp \and M.I. Weinstein \and Y. Zhu}


\date{Received: date / Revised version: date}

\maketitle

\begin{abstract}
Consider  electromagnetic waves  in two-dimensional  {\it honeycomb structured media}, whose constitutive laws have the symmetries of a hexagonal tiling of the plane. The properties of transverse electric (TE) polarized waves are determined by the spectral properties of the elliptic operator $\LA=-\nabla_\bx\cdot A(\bx) \nabla_\bx$,  where $A(\bx)$ is $\Lambda_h-$ periodic
 ($\Lambda_h$ denotes the equilateral triangular lattice), and such that with respect to some origin of coordinates, $A(\bx)$ is $\mathcal{P}\mathcal{C}-$ invariant ($A(\bx)=\overline{A(-\bx)}$) and $120^\circ$ rotationally invariant ($A(R^*\bx)=R^*A(\bx)R$, where $R$ is a $120^\circ$ rotation in the plane).
A summary of our results is as follows:
a) For generic honeycomb structured media, the  band structure of $\LA$  has {\it Dirac points}, {\it i.e.} conical intersections between two adjacent Floquet-Bloch dispersion surfaces. b) Initial data of wave-packet type, which are spectrally concentrated about a Dirac point, give rise to solutions of the time-dependent Maxwell equations whose wave-envelope, on long time scales, is governed by an effective two-dimensional  time-dependent system of massless Dirac equations.
 c) Dirac points are unstable to arbitrary small perturbations which break either $\mathcal{C}$ (complex-conjugation) symmetry or $\mathcal{P}$ (inversion) symmetry. d) The introduction through small and slow variations of  a {\it domain wall} across a line-defect  gives rise to the bifurcation from Dirac points of highly robust (topologically protected) {\it edge states}. These are time-harmonic solutions of Maxwell's equations which are propagating parallel to the line-defect and spatially localized transverse to it.
 The transverse localization and strong robustness to perturbation of these edge states  is rooted in the protected zero mode of a one-dimensional  effective Dirac operator with spatially varying mass term.
e) These results imply the existence of {\it uni-directional} propagating edge states for two classes of time-reversal invariant media in which  $\mathcal{C}$ symmetry is broken: magneto-optic media and bi-anisotropic media.
\end{abstract}


\section{Introduction and Summary of Results}

\subsection{Introduction}\label{intro}

Motivated by the novel and subtle properties of electronic waves in graphene \cites{geim2007rise}, there has been very wide interest in the propagation of waves in two-dimensional structures having the symmetries of a hexagonal tiling of the plane and its applications to electromagnetic and other types of waves; see, for example, \cites{Singha_11,Rechtsman-etal:13,Shvets-PTI:13,MKW:15,yang2015topological,wu2015scheme}. Such physical systems have been dubbed artificial graphene. The present article is motivated by {\it photonic graphene}; the propagation of waves governed by the two-dimensional Maxwell equations in honeycomb media.

Among the remarkable properties of graphene (actual and artificial) being intensively explored in the fundamental and applied scientific communities are the existence of Dirac points \cites{RMP-Graphene:09,novoselov2005two,wallace1947band,FW:12,FLW-CPAM:17,berkolaiko-comech:15} (conical points at the intersections between dispersion surfaces), the implied wave-packet (quasi-particle) dynamics of states which are
 spectrally localized near Dirac points \cites{Segev07prl,ablowitz2009conical,FW:14}, and topologically protected edge states \cites{HR:07,RH:08,RMP-Graphene:09,Delplace-etal:11,poo2011experimental,plotnik2013observation,Shvets:14,chen2014experimental,ma2015guiding,FLW-2DM:15,cheng2016robust,FLW-AnnalsPDE:16}. Conical singularities
 in dispersion relations and their consequences for wave propagation have a long history. For example, they are well-known to occur  for spatially homogeneous anisotropic Maxwell's equations; see
 \cite{Berry-Jeffrey:07} and references cited therein.

The goal of this paper is to investigate the phenomena of Dirac points, effective dynamics of wavepackets and the bifurcation theory of topologically protected edge states in the context of a class of elliptic partial differential operators,
which incorporates important physical cases of electromagnetic propagation in honeycomb structures.
Specifically we  consider a class of periodic {\it scalar} divergence-form operators:
\begin{equation}
\LA=-\nabla\cdot A \nabla\ =\ -\sum_{i,j=1}^2 \frac{\partial}{\partial x_i} a_{ij}(\bx) \frac{\partial}{\partial x_j} \ ,\ \
A(\bx)= \left(a_{ij}(\bx)\right)_{i,j=1,2}\ ,
\label{L_def}
\end{equation}
with structural assumptions on $A(\bx)$, which we now discuss.

 Introduce the following operations:  $\mathcal{C}[f](\bx)=\overline{f(\bx)}, \ $
$\mathcal{P}[f](\bx)=f(-\bx)$,  and
$\mathcal{R}[f](\bx)\equiv f(R^*\bx)$, where $R$ is the $120^\circ$ clockwise rotation matrix of vectors in the plane.
  For the bulk (unperturbed) structure,  specified by a $2\times2$ matrix function $A(\bx)$, we assume the following properties:
\begin{enumerate}
\item [(i)] $A(\bx)$ is  periodic with respect to the equilateral triangular lattice, $\Lambda_h$; see Section \ref{triangle}.
\item [(ii)] $\LA$ is self-adjoint and uniformly elliptic on $\R^2$,\\
 and  with respect to some origin of coordinates:
\item [(iii)] $\LA$ is $\mathcal{PC}-$ invariant; $\left[\mathcal{PC},\LA\right]=0$, and
\item [(iv)] $\LA$ is $\mathcal{R}-$ invariant; $\left[\mathcal{R},\LA\right]=0$, where $[\mathcal{L}_1,\mathcal{L}_2]=\mathcal{L}_1\mathcal{L}_2-\mathcal{L}_2\mathcal{L}_1$.
\end{enumerate}
\medskip

\nit A characterization of matrix constitutive laws, $A(\bx)$,  which satisfy assumptions (i)-(iv) is presented in Section \ref{honeycomb-structures}.
\medskip

In this article we obtain results on:
\begin{enumerate}
\item Dirac points (conical singularities) in the Floquet-Bloch dispersion surfaces associated with the eigenvalue problem $\LA\psi=E\psi$; see \eqref{Eq_Eigen}.
\item The non-persistence of Dirac points (opening of a local spectral gap) for the perturbed eigenvalue problem for the operator $\mathcal{L}^{A+\delta B}$, where $\delta\ll1$ and  $B(\bx)$ breaks $\mathcal{P}-$ or $\mathcal{C}-$ invariance, and the persistence of Dirac points for small $\mathcal{P}\mathcal{C}-$ invariant perturbations.
\item The bifurcation of  topologically protected edge states from Dirac points upon introduction of a (non-compact) domain wall modulated periodic perturbation, $\mathcal{L}_{\rm dw}^{\delta}$, of $\LA$; see (d) in Section \ref{main_results}.
   \end{enumerate}

Two examples in 2D electromagnetics which motivate our study of this class of operators, and to which our results apply,  arise for Maxwell's equations considered
 in (1) {\it magneto-optic media} and  (2) in {\it bi-anisotropic meta-materials}.

\nit {\bf Example 1:}\  In \cites{HR:07,RH:08}, Haldane and Raghu proposed a photonic analogue of the quantum Hall effect. In particular, they demonstrated that uni-directional protected edge states could propagate along a domain-wall for systems governed by Maxwell's equations in, for example, magneto-optic materials which break $\mathcal{C}-$ (complex-conjugation) invariance.
Such photonic edge states were experimentally observed by Wang {\it et. al.} \cites{Soljacic-etal:08}.
Our domain-wall modulated operator $\mathcal{L}_{\rm dw}^\delta$, as well as the class of Schr\"odinger operators considered
in \cites{FLW-PNAS:14,FLW-2DM:15,FLW-AnnalsPDE:16,FLW-MAMS:17}, models the physical setting of \cites{HR:07,RH:08}.
In this setting, the eigenfunction $\psi$
%
%
corresponds to $H_3$ in the case of a TE-polarized electromagnetic field:
 ${\bf E}=(E_1,E_2,0)$ and $\vec{H}=(0,0,H_3)$.
\begin{remark}[TM Polarization]
For TM-Polarization: ${\bf E}=(0,0,E_3)$ and ${\bf H}=(H_1,H_2,0)$, and
 $\psi=E_3$ satisfies a scalar eigenvalue problem for a Helmholtz operator. The techniques of this paper
 can be adapted to this case as well, but we focus on the divergence form operator $\LA$.
\end{remark}

 \nit {\bf Example 2:}\ In Khanikaev {\it et. al.} \cites{Shvets-PTI:13}, bi-anisotropic media which respect time-reversal symmetry but break $\mathcal{C}-$ invariance,  were studied and the counter-propagation of uni-directional states of opposite ``spin'', $\psi_\pm=H_3\pm E_3$, associated with an underlying Kane-Mele model, was computationally demonstrated.
 \medskip

\nit  A detailed discussion of both examples is presented in Appendix \ref{maxwell}.

\subsection{Main results}\label{main_results}

We give a brief summary of our main results. Our approach makes use of and extends methods developed to treat the case of honeycomb Schr\"odinger operators \cites{FW:12,FW:14,FLW-PNAS:14,FLW-2DM:15,FLW-AnnalsPDE:16,FLW-MAMS:17,FLW-CPAM:17}.
\begin{enumerate}[(a)]
\item {\it Characterization of honeycomb structured media:}\ Theorem \ref{invariance} and Corollary \ref{A-expand} characterize 2D honeycomb structured media, defined by
 the Hermitian operator $\LA=-\nabla\cdot A(\bx)\nabla$ (see \eqref{L_def}), where $A(\bx)$ satisfies conditions (i)-(iv) above.
 \item {\it Dirac points:}\ Theorem \ref{prop_conical} gives conditions for the existence of Dirac points at the vertices of the Brillouin zone,
 $\mathcal{B}_h$.
 Theorem \ref{low-dp} states that all sufficiently low-contrast honeycomb structures, for which a distinguished Fourier coefficient of the structure is non-zero, have Dirac points. Depending on the sign of this Fourier coefficient, these Dirac points occur either at intersections of the $1^{st}$ and $2^{nd}$ or $2^{nd}$ and $3^{rd}$ dispersion surfaces.
  Extension of these results to generic honeycomb structured media of arbitrary contrast, using the strategy of
\cites{FW:12} (see also Appendix D of \cites{FLW-MAMS:17}) is discussed in Section \ref{generic_eps}.
 In Figure \ref{E_k_BZ} we plot two dispersion surfaces for a honeycomb structured medium, $A(\bx)$, which intersect at Dirac points for quasi-momenta located at the size vertices of the Brillouin zone.

 Remark \ref{wavepkts} discusses the effective dynamics of wave-packets for initial data which are spectrally localized near Dirac points. The envelope of such wave-packets evolves slowly and on long time scales according to an effective massless Dirac equation.

 \begin{figure}
\centering
\includegraphics[width=4in]{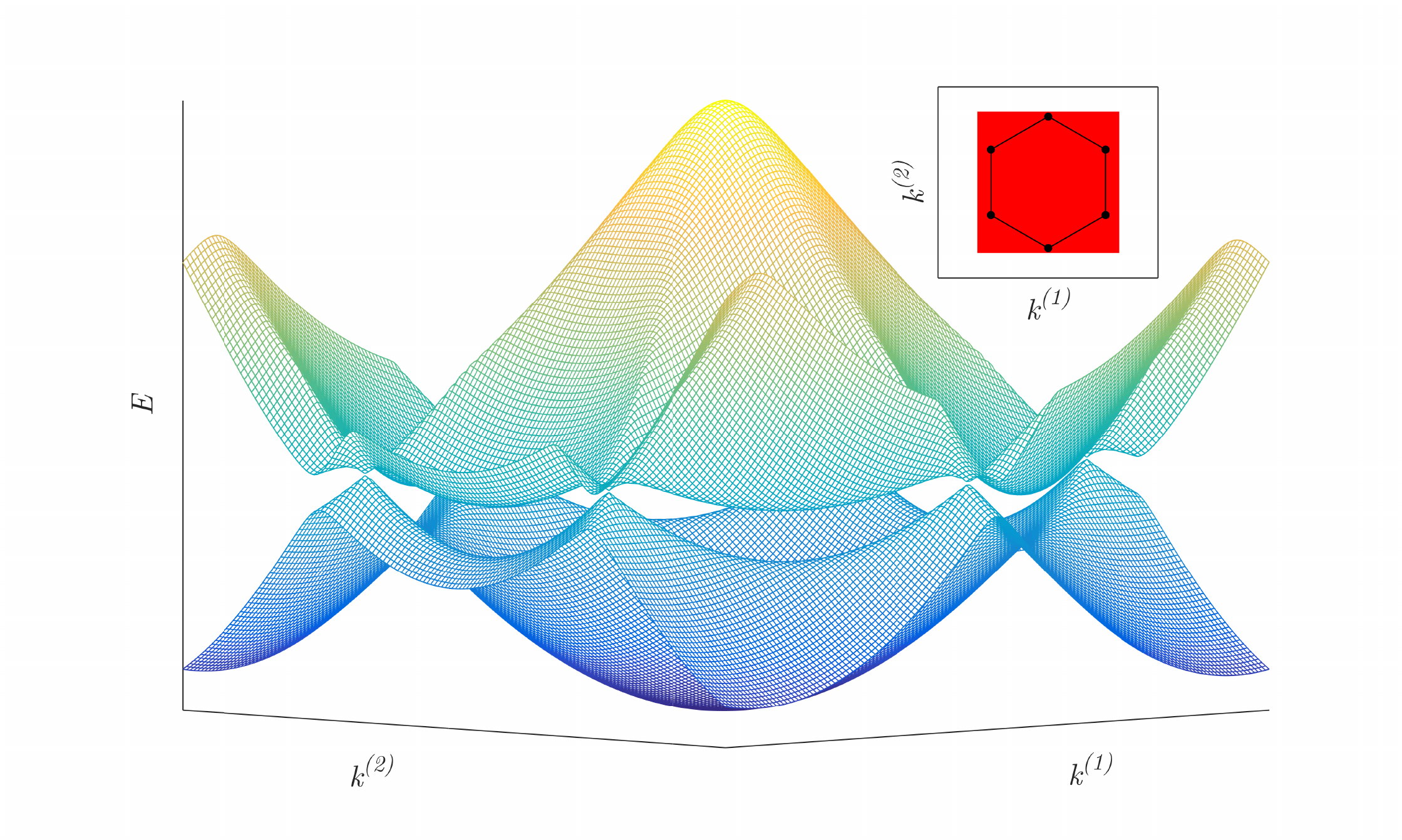}
 \caption{\footnotesize
 Lowest two dispersion surfaces $\bk\equiv(k^{(1)},k^{(2)})\in\mathcal{B}_h\mapsto E(\bk)$ of the band structure of $\LA$, with $A(\bx)=\e(x)I$ where $\e(\bx)$ is a particular honeycomb lattice function.
 Dirac points occur at the intersection of two dispersion surfaces, at the six vertices of the Brillouin zone, $\mathcal{B}_h$.}
 \label{E_k_BZ}
 \end{figure}

 \item {\it Non-persistence and persistence of Dirac points:}\ Theorem \ref{PT_thm} studies the instability of Dirac points under perturbations of honeycomb media, $-\nabla\cdot A\nabla\ \to\
-\nabla\cdot\ \left(\ A(\bx) +\delta B(\bx)\ \right) \nabla$    which break either
 $\mathcal{P}-$ or $\mathcal{C}-$ invariance; thus, $[\pc,\LB]\ne0$.  Here, $\mathcal{P}-$ refers to parity inversion, $f(\bx)\mapsto\mathcal{P}[f](\bx)=f(-\bx)$ and $\mathcal{C}-$ refers to complex conjugation, $f(\bx)\mapsto\mathcal{C}[f](\bx)=\overline{f(\bx)}$.
 If  $\pc-$ invariance is preserved, then the Dirac points persist, although their associated quasi-momenta may no longer be located on the vertices of the Brillouin Zone.

 \item {\it Topologically protected edge states (Theorems \ref{edgemode_thm} and \ref{rig-edge}):}
 We consider a honeycomb structure perturbed by an edge, across which the structure is adiabatically modulated by a domain wall.
 \footnote{We comment briefly on the terms {\it edge} and {\it edge state}.
   An edge is frequently understood to mean an abrupt termination of bulk structure.
The terms ``edge'' for a line-defect across which there is a change in a key characteristic of the structure, and ``edge state'' are also used in the physics literature; see, for example, \cites{HR:07, RH:08,Shvets-PTI:13}.
The edge states we discuss are of the latter type.}

 Starting from a honeycomb operator with Dirac points, $\LA$, we first consider operators
 $\mathcal{L}^\delta_{\pm\infty}=-\nabla\cdot\left[A(\bx) \pm\delta\eta_{\infty}B(\bx)\right]  \nabla$.
 Here, $\delta$ is small non-zero real parameter, $\eta_\infty>0$,  and $B(\bx)$ breaks $\pc-$ invariance, {\it i.e.} $[\pc,B]\ne0$. It follows from (c) that the operators $\mathcal{L}^\delta_{\pm\infty}$ have a local spectral gap, {\it i.e.} a gap for quasi-momenta varying near  Dirac points. Next, transverse to the line $\R\vtilde_1$ in the lattice direction $\vtilde_1\in\Lambda_h$,  we interpolate between operators $\mathcal{L}^\delta_{\pm\infty}$
with a {\it domain wall function} $\eta(\zeta)$:
 \begin{equation}
  \label{dw_function}
\eta(0)=0,\ \  \eta(\zeta) \rightarrow \pm \eta_{\infty}\ \text{with} \ \ \eta_{\infty}>0 \ \ \text{as} \ \ \zeta\rightarrow\pm\infty,
 \end{equation}
by introducing the operator
\begin{equation}
  \mathcal{L}_{\rm dw}^{(\delta)} \equiv -\nabla\cdot\left[A(\bx) +\delta\eta( \delta\ktilde_2 \cdot \bx)B(\bx)\right]  \nabla\ .
\label{L_dw} \end{equation}
Here, $\ktilde_2\in\Lambda_h^*$ is such that $\ktilde_2\cdot\vtilde_1=0$.
  The constitutive matrix: $A(\bx)+\delta\eta(\delta\ktilde_2\cdot\bx)B(\bx)$ defines a medium which is an interpolation (small and adiabatic), transverse to the  line-defect / ``edge'' $\R\vtilde_1$,  between perturbed (``gapped '') honeycomb structures.

  Associated with the  $\bx\mapsto\bx+\vtilde_1$ translation invariance of
   $\mathcal{L}_{\rm dw}^{(\delta)}$  is
   a parallel quasi-momentum, denoted $\kpar$. An edge state for parallel quasi-momentum, $\kpar$,  is a non-trivial solution of the eigenvalue problem $\mathcal{L}_{\rm dw}^{(\delta)}\Psi =E\ \Psi$ such that
   $ \Psi(\bx+\vtilde_1)=e^{i\kpar}\Psi(\bx)$ and $\Psi(\bx)\to0$ as $|\ktilde_2\cdot\bx|\to\infty$. The edge state eigenvalue problem is naturally posed on $L^2(\Sigma)$, where $\Sigma=\R^2/\Z\vtilde_1$ is a cylinder. Edge states are propagating (plane-wave like) parallel to the edge and are localized transverse to the edge.

 Theorem \ref{edgemode_thm} presents the formal multiple scale expansion of such states. These modes bifurcate from the continuous spectrum at the Dirac point energy of the underlying, unperturbed ($\delta=0$) operator as the perturbation parameter $\delta$ is varied away from zero. The bifurcation of these modes is ``topologically protected'' in the sense that these modes persist in the presence of spatially localized (even large) perturbations of the domain-wall function, $\eta(\zeta)$.

 The rigorous formulation of Theorem \ref{edgemode_thm} is given in Theorem \ref{rig-edge}.
 Section \ref{rigor} contains a detailed sketch of the proof of Theorem \ref{rig-edge}, based on the corresponding result for a class of Schr\"odinger operators (see \cites{FLW-AnnalsPDE:16}
 and \cites{FLW-MAMS:17}). A central role is played by the \emph{spectral no-fold condition} which is motivated and discussed. The validity of this condition for low contrast honeycomb structures for zigzag edges follows directly from the analysis for the corresponding Schr\"odinger operator case \cites{FLW-AnnalsPDE:16}.
 Figure \ref{E_delta_zz} illustrates numerical computations of bifurcations of edge states.

 \begin{figure}
 \centering
 \includegraphics[width=\textwidth]{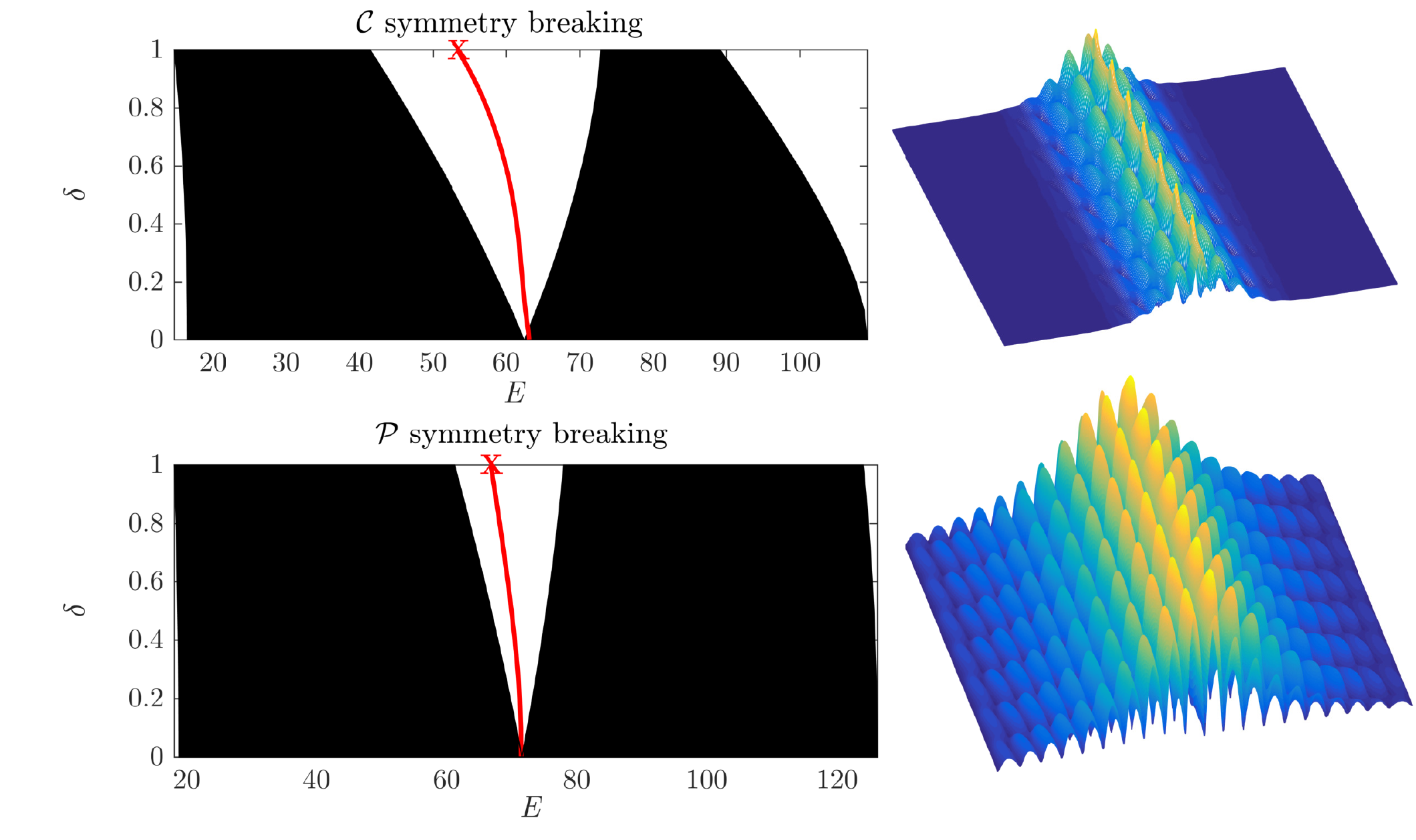}
 \caption{\footnotesize
 {\bf Left panels}: Bifurcation curves of topologically protected edge states (red curves) for the zigzag edge, illustrated by the $L^2_{\kpar}(\Sigma)-$ energy spectrum of $\mathcal{L}_{\rm dw}^{(\delta)}$ (see \eqref{L_dw}) vs. the perturbation parameter, $\delta$, for $\kpar=\bK\cdot\bv_1=2\pi/3$.
 The lowest two continuous spectral bands are shown in black.
 Here, $A(\bx)=\e(x)I_{_{2\times2}}$ defines a honeycomb structured medium, {\it i.e.} $\e(x)$ is real, even and $\e(R^*\bx)=\e(\bx)$. We take $\eta(\zeta)=\tanh(\zeta)$, a domain-wall function, and $B(\bx)$ is chosen to break $\mathcal{C}-$  or $\mathcal{P}-$  symmetries, respectively, in the top and bottom panels.
 {\bf Right panels}: Edge states corresponding to red $X$'s in the left panels.
 }
 \label{E_delta_zz}
 \end{figure}

 \item \emph{Group velocity properties of wavepackets concentrated along edges:}
 Theorem \ref{near-kpar1} shows the existence of edge states for all $\kpar$ near
  $\kpar=\bK\cdot\vtilde_1$ and $\kpar= -\bK\cdot\vtilde_1$.
   Taking a weighted and continuous superposition of such edge states in $\kpar$, yields fully localized wave-packets which are concentrated along the edge. The group velocity of these wave-packets is given by $\partial_\kpar E(\kpar)$ evaluated at $\kpar=\pm\bK\cdot\vtilde_1$.

\item \emph{Unidirectional edge propagation in time-reversal invariant and $\mathcal{C}-$ breaking media:} In Section \ref{HR-K} we apply our results to the propagation of edge states in the two examples discussed in Section \ref{intro}: (1) magneto-optic media, and (2) bi-anisotropic media; see also Appendix \ref{maxwell}. We consider the special cases where $\pc-$ invariance is broken by  perturbations which are (a) $\mathcal{P}-$ invariant, but $\mathcal{C}-$ anti-symmetric, and
 (b) $\mathcal{C}-$ invariant, but $\mathcal{P}-$ anti-symmetric. In case (a),  the families of edge states constructed are uni-directional, although Maxwell's equations in this setting is  time-reversal invariant (Section \ref{Maxwell_symmetry}); see the top panel of Figure \ref{E_kpar_zz}. In case (b), the resulting edge states propagate  bi-directionally; see the bottom panel of Figure \ref{E_kpar_zz}.

  We remark that in \cites{HR:07,RH:08}, which considers magneto-optic media, and \cites{Shvets-PTI:13}, which considers bi-anisotropic media,  it was assumed that the unperturbed bulk medium is isotropic ($A(\bx)=a(\bx)I_{_{2\times2}}$). In this article we allow for a more general class of  anisotropic media; see our definition of {\it honeycomb structured medium} in Section \ref{honey-media}.

 \begin{figure}
 \centering
 \includegraphics[width=\textwidth]{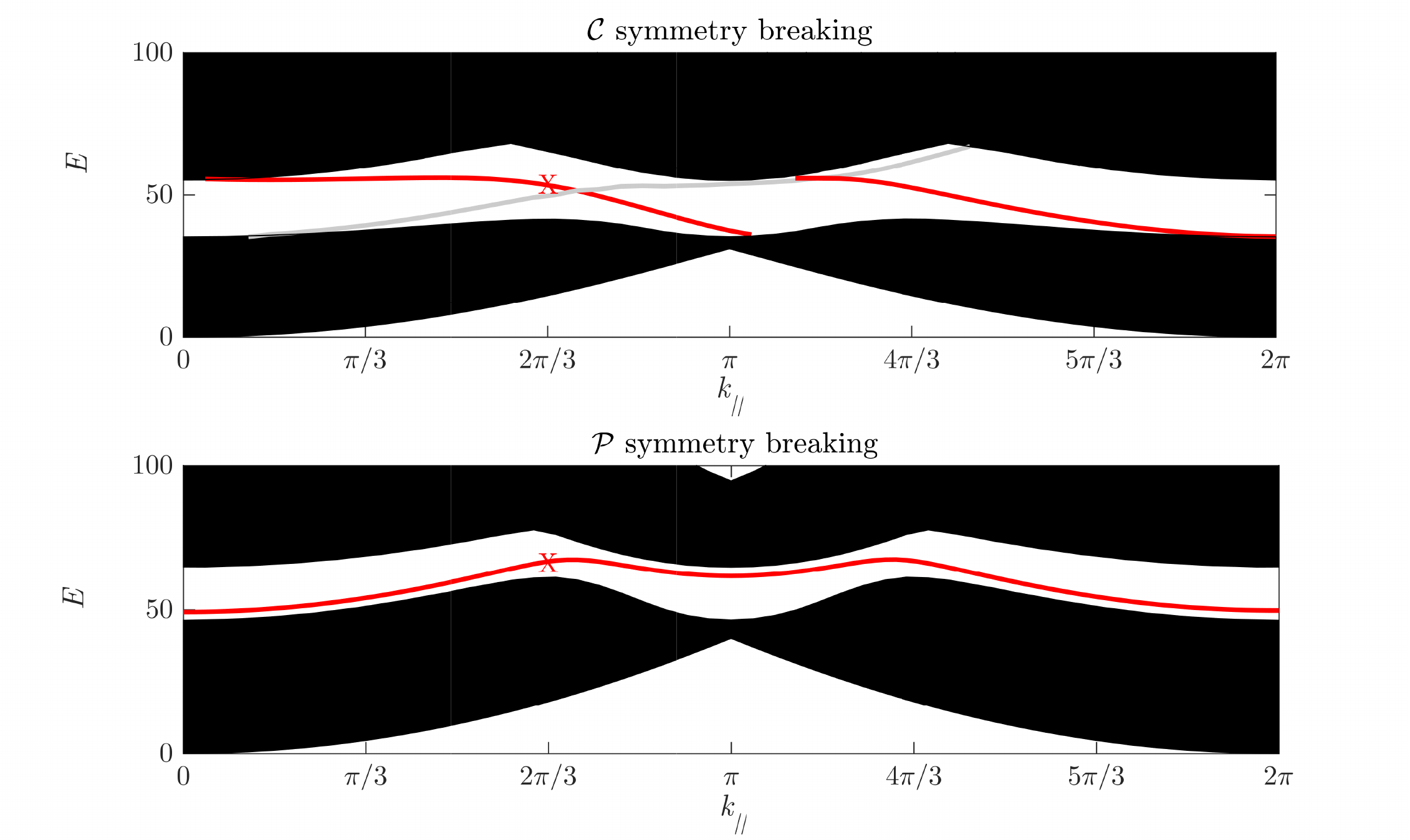}
 \caption{\footnotesize
 $L^2_{\kpar}(\Sigma)-$ energy spectrum of $\mathcal{L}_{\rm dw}^{(\delta=1)}$ vs. parallel quasi-momentum, $\kpar$. $\mathcal{L}_{\rm dw}^{(\delta=1)}$ is the same as given in each of the respective (top and bottom) left panels in Figure \ref{E_delta_zz}.
 {\bf Top panel:} \emph{$\mathcal{C}-$ symmetry breaking}.
 Topologically protected (domain-wall induced) edge modes are denoted by the red curves. These modes are uni-directional.
 Spurious (gray) ``hard edge'' modes are a result of the finiteness effects associated with the choice of numerical domain; see Appendix \ref{numerical_schemes} for a detailed discussion of the numerics and a discussion of the spurious modes.
 {\bf Bottom panel:} \emph{$\mathcal{P}-$ symmetry breaking}.
 The (red) edge modes are denoted by the red curves are bi-directional.}
 \label{E_kpar_zz}
 \end{figure}

 \end{enumerate}

\subsection{Connections to previous rigorous analytical work}
The existence of Dirac points for Schr\"odinger operators with generic honeycomb lattice potentials was proved in \cites{FW:12,FLW-MAMS:17}. Results for small amplitude potentials were obtained in  \cites{Colin-de-Verdiere:91,Grushin:09}.
A group representation perspective on the existence and persistence of Dirac points is developed in
  \cites{berkolaiko-comech:15}. The case where the potential is a superposition of delta function potentials centered on sites of the honeycomb structure is treated in \cite{Lee:16}. The low-lying dispersion surfaces of honeycomb Schr\"odinger operators in the strong binding regime, where the potential is the superposition of a general class of atomic potential wells, and its relation to the  tight-binding limit, was studied in
  \cites{FLW-CPAM:17}. Tight binding models have been studied extensively in the fundamental and applied physics and mathematics communities;
   see, for example, \cite{Ashcroft-Mermin:76,Dimassi-Sjoestrand:99,ACZ:12,Makwana-Craster:14}.

   A bifurcation theory of topologically protected bound states arising from domain wall perturbations of one-dimensional periodic Schr\"odinger operators with Dirac points (linear band crossings) is developed in
  \cites{FLW-PNAS:14,FLW-MAMS:17}. These results are  applied to a photonic setting in \cites{Thorp-etal:15} and \cites{poo2016observation,poo2017global}. The bifurcation of topologically protected edge states arising from domain wall perturbations of two-dimensional honeycomb Schr\"odinger operators with Dirac points (conical band crossings) was studied in \cites{FLW-AnnalsPDE:16}; see also
    \cites{FLW-2DM:15}.

    For an extensive discussion of Dirac points and edge states for nanotube structures in the context of  quantum graphs, see \cites{Kuchment-Post:07} and \cites{Do-Kuchment:13}.
For results on spectral gaps for elliptic problems modeling high contrast media,  see
 \cites{Figotin-Kuchment:96a,Figotin-Kuchment:96b,Hempel-Lienau:00,Hempel-Post:02,Lipton-Viator_1,Lipton-Viator_2,FLW-CPAM:17}.
  For a general review of mathematical problems arising in photonics see, for example, \cite{kuchment-01}.

\subsection{Outline}  The paper is structured as follows.
 In Section \ref{math-problem}, we review the relevant spectral (Floquet-Bloch) theory and introduce the notions of
 $\mathcal{C}$, $\mathcal{P}$ and $\mathcal{R}$ invariance, associated with complex conjugation, parity and rotational symmetries. Section \ref{HSM} introduces and characterizes honeycomb structured media via Fourier analysis.
 In Section \ref{dirac-pts} we discuss the notion of {\it Dirac point} in the band structure of the operator $\LA$ and prove
 results for the cases of low-contrast and generic honeycomb structures.  Section \ref{dirac_persistence} discusses the stability of Dirac points against small perturbations which preserve $\mathcal{P}\circ\mathcal{C}$ symmetry and their instability when such symmetry is violated.  In  Section \ref{edge_states} we discuss the existence of topologically protected edge states. We first present the detailed formal multi-scale construction of such states and then give a detailed guide to adapting the proof for the case of Schr\"odinger operators \cites{FLW-AnnalsPDE:16} to the case of divergence form operators  for honeycomb structured media. Note that results of this paper apply to large classes of complex-valued and anisotropic matrix-valued constitutive laws; see the examples arising in electromagnetics, discussed in Appendix \ref{maxwell}.
Edge states are constructed for a range of parallel quasi-momenta $\kpar$. In Section  \ref{edge_state_kpar} we discuss
the edge state dispersion curves $\kpar\mapsto E(\kpar)$, which determine the propagation properties of edge-localized wave-packets; see, in particular, the discussion of group velocity in Subsection \ref{direction}.
Appendix \ref{maxwell} discusses Maxwell's equations for bi-anisotropic media and the application of our general results to edge states and unidirectional edge-propagation in (1) magneto-optic and (2) bi-anisotropic media, both of which
are time-reversal invariant.
Appendix \ref{numerical_schemes} provides a brief discussion of the numerical methods used in our computer simulations. Appendix \ref{fig_potentials} provides the specific potentials used in these simulations.

\subsection{Notation and conventions}
\begin{enumerate}
\item $\Lambda_h\subset\R^2$ denotes the equilateral triangular lattice.\  $\Lambda_h^*\subset (\R^2)^*=\R^2_\bk$
denotes the dual  lattice. $\bv_j, ~j=1,2$ are the basis vectors of  $\Lambda_h$.
$\bk_l, ~l=1,2$ are the dual basis vectors of $\Lambda_h^*$, chosen to satisfy $\bk_l\cdot\bv_j=2\pi\delta_{lj}$.
\item $\vtilde_1=a_1\bv_1+a_2\bv_2\in\Lambda_h$,\ $a_1, a_2$ co-prime integers. The $\vtilde_1$- edge is $\R\vtilde_1$.
 $\vtilde_j,\ j=1,2$, is an alternate basis for $\Lambda_h$ with corresponding dual basis, $\ktilde_\ell, \ell=1,2$, satisfying $\ktilde_\ell\cdot\vtilde_j=2\pi\delta_{\ell j}$.
  \item For $\bfm=(m_1,m_2)\in\Z^2$, $\bfm\vec\bk=m_1\bk_1+m_2\bk_2$.
 \item $\B$ denotes the Brillouin Zone, associated with $\Lambda_h$, shown in the right panel of Figure \ref{lattice-and-itsdual}.
\item $z \in \mathbb C \Rightarrow \overline{z}$ denotes the complex conjugate of $z$.
 \item $J=\begin{pmatrix}0&-1\\1&0\end{pmatrix}$ is the counterclockwise $90^\circ$ rotation matrix.
 \item $\bk=(k^{(1)},k^{(2)})^T\in\mathbb R^2$ is represented in $\C$ by $\mathfrak{z}=k^{(1)}+ik^{(2)}$, and $\bk_\perp\equiv J\bk=(-k^{(2)},k^{(1)})^T$ is represented in $\mathbb C$ by $\mathfrak{z}_\perp=i\mathfrak{z}(\bk)=\mathfrak{z}(\bk_\perp)$.
  The corresponding points on the unit circle in $\C$ are denoted $\hat{\mathfrak{z}}(\bk)$ and $\hat{\mathfrak{z}}(\bk_\perp)$.
\item $\nabla=(\partial_{x_1}, \partial_{x_2})^T$.
\item $\inner{f,g}_{D}=\int_{D}\overline{f}g$  is the $L^2(D)$ inner product. If the region of integration is not specified, it is assumed to be a choice of fundamental period cell, $\Omega$.
\item Pauli matrices:
\[ \sigma_1=\begin{pmatrix}0&1\\1&0\end{pmatrix},\qquad \sigma_2=\begin{pmatrix}0&-i\\i&0\end{pmatrix}=iJ,\qquad \sigma_3=\begin{pmatrix}1&0\\0&-1\end{pmatrix}.\]
\item Unless otherwise specified, summation convention over repeated indices is assumed.
\end{enumerate}

\begin{acknowledgement}
The authors wish to thank C.L. Fefferman, L. Lu, M. Rechtsman, D. Ketcheson, V. Quenneville-B{\'e}lair and  N. Yu
 for stimulating discussions.
This research was supported in part by NSF grants: DMS-1412560, DMS-1620418, DGE-1069420 and Simons Foundation Math + X Investigator grant \#376319 (MIW); and the NSF grant DMR-1420073 (JPL-T).
YZ acknowledges the hospitality of the Department of Applied Physics and Applied Mathematics during academic visits to Columbia University,  supported by Tsinghua University Initiative Scientific Research Program \# 20151080424 and NSFC grants \#11471185 and \#11871299.
\end{acknowledgement}

\section{Preliminaries}\label{math-problem}

In this section, we outline the relevant spectral theory
 \cites{Eastham:74, RS4,kuchment2012floquet,kuchment2015overview} and introduce terminology and notation for discussing the symmetry properties of the (unperturbed) bulk operator $\LA$.

\subsection{Fourier analysis}\label{fourier-analysis}

Let $\{\bv_1, \bv_2\}$ be a linearly independent set in $\mathbb{R}^2$. The lattice generated by $\{\bv_1, \bv_2\}$ is the subset of $\R^2$:
\[\Lambda=\{\bfm\vec\bv=m_1\bv_1+m_2\bv_2: \bfm=(m_1, m_2)\in \mathbb{Z}^2\}=\mathbb{Z}\bv_1\oplus\mathbb{Z}\bv_2 .\]
A choice of fundamental cell is the parallelogram:
\[\Omega=\{\theta_1\bv_1+\theta_2\bv_2: 0\le\theta_j\le1,j=1,2\}.\]
The area of $\Omega$ is denoted $|\Omega|$.

The dual lattice, $\Lambda^*$, is defined to be
\[\Lambda^*=\{\bfm\vec\bk=m_1\bk_1+m_2\bk_2: \bfm=(m_1, m_2)\in \mathbb{Z}\}=\mathbb{Z}\bk_1\oplus\mathbb{Z}\bk_2,\]
where $\bk_1$ and $\bk_2$ are dual lattice vectors, satisfying the reciprocal relations
\[ \bk_i\cdot\bv_j=2\pi \delta_{ij}.\]
We also introduce the (affine) lattice $\Lambda^*_\bk$:
\begin{equation}
 \label{dual_lattice_k}
 \Lambda^*_{\bk}\equiv \bk+\Lambda^{*}=\left\{\bk+\bfm\vec\bk: \bfm\in \mathbb Z^2\right\}.
\end{equation}

Let $L^2(\mathbb{R}^2/\Lambda)$ denote the space of $L^2_{loc}$ functions which are $\Lambda-$periodic, i.e.
$f(\bx+\bv)=f(\bx)$, for almost all $ \bx\in \mathbb{R}^2$ and all $ \bv\in \Lambda$.
For $\bk\in\mathbb{R}^2$, we denote by $L^2_{\bk}$ the space of $L^2_{loc}$ functions such that
 $e^{-i\bk\cdot\bx}f(\bx)\in L^2(\mathbb{R}^2/\Lambda)$, {\it i.e.}
 $f(\bx+\bv)=e^{i\bk\cdot \bv}f(\bx), \quad\text{for} \quad \bx\in \mathbb{R}^2,\ \bv\in \Lambda$.
Note that if $\bk=0$, $L^2_{\bk}$ reduces to $L^2(\R^2/\Lambda)$.

Let $f$ and $g$ in $L^2_{\bk}$. Then, $\bar{f} g $ is locally integrable and $\Lambda-$periodic, and their inner product is naturally defined as
\[\inner{f,g}=\intomega{\overline{f(\bx)}g(\bx)}.\]
In a standard way, one can also introduce the Sobolev space $H^s_{\bk}$.

Any $f\in L^2_{\bk}$ has the Fourier expansion
\begin{align}
  f(\bx) &=\sum_{\bfm \in \mathbb{Z}^2}f_\bfm e^{i (\bk+\bfm\vec\bk)\cdot\, \bx},\ \ \textrm{where}\ \
f_\bfm = \frac{1}{|\Omega|}\int_\Omega e^{-i(\bk+\bfm\vec\bk)\cdot \,\by}f(\by) d\by. \label{fourier}
\end{align}

\subsection{Floquet-Bloch theory of $\LA=-\nabla\cdot A\nabla$}\label{fb_theory}

We summarize definitions and results from the theory of self-adjoint, elliptic and periodic divergence form
operators; see, for example,  \cite{Kuchment_Levendorskii_01,kuchment2012floquet,kuchment2015overview}.

\begin{assumption}\label{A_Assumption}
The $2\times2$ complex-valued matrix function $A(\bx)$ satisfies
\begin{enumerate}[label=(\subscript{A}{\arabic*})]
\item $A(\bx)$ is a smooth  and Hermitian, {\it i.e.}, $A(\bx)^{\dagger}\equiv\overline{A(\bx)^T}=A(\bx)$ for all $\bx$.
 \item The mapping $\bx\mapsto A(\bx)$ is $\Lambda-$ periodic.
\item $A(\bx)$ is uniformly elliptic, {\it i.e.} there exist constants $c_\pm>0$, such that for all $\bx\in\R^2$ and all $\xi\in\C^2$:
   $c_-|\xi|^2\ \le\ \left\langle \xi, A(\bx)\xi\right\rangle_{\C^2}\ \le c_+|\xi|^2$.
 \end{enumerate}
\end{assumption}

For each fixed $\bk\in \mathbb{R}^2$, consider the $L^2_\bk-$ Floquet-Bloch eigenvalue problem:
\begin{equation}\label{Eq_Eigen}
\begin{split}
\LA\Phi(\bx)&= E\Phi(\bx), \quad \bx \inr , \\
\Phi(\bx+\bv)&=e^{i\bk\cdot \bv}\Phi(\bx), \quad  \bx\in\R^2,\quad \bv\in \Lambda ,
\end{split}
\end{equation}
where $\LA=-\nabla\cdot A\nabla$ was defined in \eqref{L_def}.
A solution of \eqref{Eq_Eigen} is called a Floquet-Bloch state with quasi-momentum $\bk$.

Since the $\bk$-pseudo-periodic boundary condition in \eqref{Eq_Eigen} is invariant under translation $\bk\rightarrow \bk+{\bk}^\prime$ for any ${\bk}^\prime\in \Lambda^*$, it suffices to consider $\bk$ varying over a fundamental cell. A common choice is the {\it Brillouin Zone} $\B$, consisting of points $\bk\in \R^2$ which are closer to the origin than to any other  points in $\Lambda^*$; see Figure  \ref{lattice-and-itsdual}.

\begin{figure}
\centering
\includegraphics[width=\textwidth]{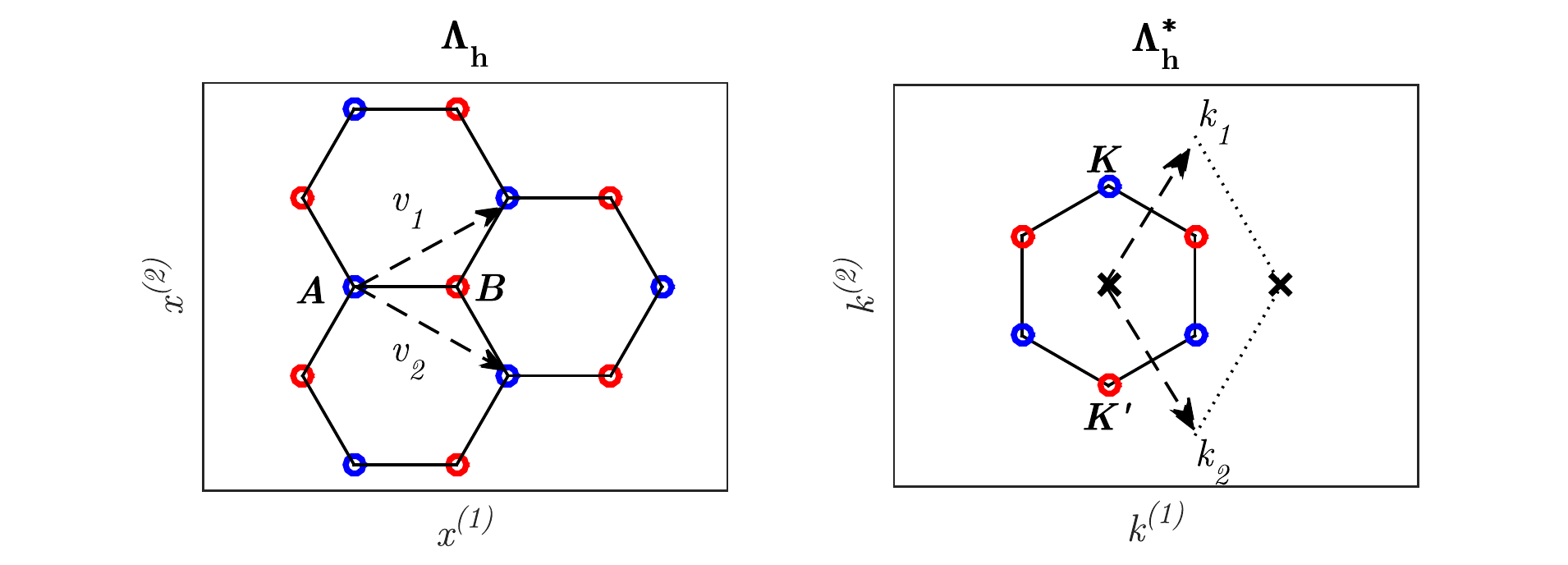}
\caption{\footnotesize
{\bf Left panel}: ${\bf A}=(0,0)$, ${\bf B}=(\frac{1}{\sqrt3},0)$.
Honeycomb structure, ${\bf H}$,  is the union of two sub-lattices  $\Lambda_{\bf A}={\bf A}+\Lambda_h$ (blue)
and $\Lambda_{\bf B}={\bf B}+\Lambda_h$ (red); several hexagons shown. The lattice vectors  $\{\bv_1,\bv_2\}$ generate $\Lambda_h$.
{\bf Right panel}:
Brillouin zone, $\B_h$, and dual basis $\{\bk_1,\bk_2\}$. $\bK$ and $\bK'$ are labeled.}
\label{lattice-and-itsdual}
\end{figure}

An alternative formulation of the eigenvalue problem \eqref{Eq_Eigen}  is obtained by setting $\Phi(\bx)=e^{i\bk\cdot \bx}\phi(\bx)$ for $\bk\in\B$. Let
\begin{equation}
\label{Eq_Eigen_k_op}
\LA(\bk)= e^{-i\bk\cdot\bx}\ \mathcal{L}^A\ e^{i\bk\cdot\bx}\ =\ -(\nabla+i\bk)\cdot A(\bx)(\nabla+i\bk).
\end{equation}
Then, \eqref{Eq_Eigen} is equivalent to  $\phi$ satisfying the  periodic elliptic boundary value problem:
\begin{equation}\label{Eq_Eigen_k}
\begin{split}
\LA(\bk)\phi(\bx) &= E(\bk)\phi(\bx), \quad \bx \inr ,\\
\phi(\bx+\bv) &=\phi(\bx), \quad \bx\in\R^2,\ \ \bv\in \Lambda .
\end{split}
\end{equation}

For each fixed $\bk$, the eigenvalue problem \eqref{Eq_Eigen_k} (equivalently \eqref{Eq_Eigen}) has a discrete spectrum: $ E_1(\bk)\le E_2(\bk)\le E_3(\bk)\le\cdots\le E_b(\bk)\le \cdots$
with eigenpairs $( \phi_b(\bx;\bk)$, $ E_b(\bk))$, $b\ge1$. The set $\left\{\phi_b(\bx;\bk):b\ge1\right\}$ can be taken to be a complete and orthonormal set in $L^2(\mathbb{R}^2/\Lambda)$.

The mappings $\bk \mapsto E_b(\bk)$ are called band dispersion functions. They are Lipschitz-continuous, see, for example, \cites{Avron-Simon:78,Kuchment_Levendorskii_01,kuchment2012floquet,kuchment2015overview} and Appendix A of \cites{FW:14}. As $\bk$ varies over $\B$, each function $ E_b(\bk)$ sweeps out a closed real interval. The union over $b\ge 1$ of these closed intervals is precisely the $L^2(\mathbb{R}^2)$-spectrum of  $\mathcal{L}^A$.
Furthermore,  the set $\left\{\Phi_b(\bx;\bk)\right\}_{b\ge 1, \bk \in \B}$ is complete in $L^2(\mathbb{R}^2)$:
\begin{equation*}
  f(\bx)=\displaystyle{\sum_{b\ge1}\int_\B\inner{\Phi_b(\cdot;\bk), f(\cdot)}_{L^2(\mathbb R^2)}\Phi_b(\bx;\bk)\,d\bk},
\end{equation*}
where the sum converges in the $L^2$ norm.

\subsection{$\mathcal{C}$, $\mathcal{P}$ and $\mathcal{R}$ invariance}\label{CPR-inv}

In this section we introduce terminology for discussing the symmetry properties of $A(\bx)$ and $\LA$. Let $g(\bx)$ denote a function defined on $\R^2$.

{\bf $\mathcal{C}-$ invariance:} The function $g(\bx)$ is $\mathcal{C}-$ invariant if
\begin{equation}
\left(\mathcal{C}g\right)(\bx)\ \equiv\ \overline{g(\bx)}\ =\ g(\bx)
\label{Cg_eq_g}
\end{equation}
for all $\bx\in\R^2$.

{\bf $\mathcal{P}-$ invariance:} The function $g(\bx)$ is $\mathcal{P}-$ (or parity-inversion) invariant with respect to $\bx_0\in\R^2$  if
\begin{equation*}
\left(\mathcal{P}g\right)(\bx)\ \equiv\ g(2\bx_0-\bx)\ =\ g(\bx).
\end{equation*}
for all $\bx\in\R^2$.

{\bf $\mathcal{R}-$ invariance:} The function $g(\bx)$ is $\mathcal{R}-$ (or $120^\circ$- rotationally) invariant with respect to $\bx_0\in\R^2$  if
\begin{equation*}
\left(\mathcal{R}g\right)(\bx)\ \equiv\ g(\bx_0+R^*(\bx-\bx_0))\ =\ g(\bx).
\end{equation*}
for all $\bx\in\R^2$. Here, $\bx\mapsto R\bx$ is the mapping on $\R^2$ which rotates a vector clockwise by
$120^\circ$ ($2\pi/3$) about  $\bx={\bf 0}$:
\begin{equation}\label{Def_R}
R=\begin{pmatrix}-\frac{1}{2} & \frac{\sqrt{3}}{2}\\[1.5 ex ]-\frac{\sqrt{3}}{2} & -\frac{1}{2}\end{pmatrix}.
\end{equation}
For later use, we record the eigenvalues and a choice of normalized eigenvectors of $R$:
\begin{equation}
\label{R_evals_evecs}
R\zeta\ =\ \tau\ \zeta, \quad  R\ \bar{\zeta}\ =\ \bar{\tau}\ \bar{\zeta}, \quad \text{where} \quad \tau=e^{i2\pi/3}\ \textrm{and}\
\zeta=\frac{1}{\sqrt{2}}\begin{pmatrix}1\\i\end{pmatrix} .
\end{equation}

{\it We say the operator  $L$ with domain $D(L)\subset D(\mathcal{S})$, is $\mathcal{S}-$ invariant if
the commutator } $ [\mathcal{S},L]\equiv \mathcal{S}L-L\mathcal{S}$ vanishes on $D(L)$.

\section{Honeycomb structured media}\label{HSM}

\subsection{The equilateral triangular lattice, $\Lambda_h$}\label{triangle}

In this section, we specialize the results of the previous section to the triangular lattice, $\Lambda_h$.
 Introduce the lattice
\[ \Lambda_h=\Z\bv_1 \oplus \Z\bv_2,\]
with lattice basis vectors
 \[
 \bv_1=\begin{pmatrix} \frac{\sqrt{3}}{2} \\  \\ \frac{1}{2}\end{pmatrix},\quad
 \bv_2=\begin{pmatrix}\frac{\sqrt{3}}{2} \\ \\ -\frac{1}{2} \end{pmatrix}.\]
The dual lattice
\[\Lambda_h^* = \Z\bk_1\oplus \Z\bk_2,\] is spanned by the dual basis vectors
\[ \bk_1=\frac{4\pi}{\sqrt{3}}\begin{pmatrix}\frac{1}{2} \\ \\ \frac{\sqrt{3}}{2}\end{pmatrix},\quad
 \bk_2=\frac{4\pi}{\sqrt{3}}\begin{pmatrix}\frac{1}{2}\\ \\ -\frac{\sqrt{3}}{2}\end{pmatrix} ;\]
see Figure  \ref{lattice-and-itsdual}.
Note  $\bk_l\cdot\bv_m=2\pi \delta_{lm}$, $l,m=1,2$.

The following proposition characterizes the high-symmetry points of the Brillouin Zone, $\B$.
\begin{proposition}\label{special_k_points}
The translated dual lattice $\Lambda^*_{h,\bk}= \bk + \Lambda_h^*$ is invariant under $120^\circ$ rotation,  {\it i.e.}
 $R:\Lambda^*_{h,\bk}\to \Lambda^*_{h,\bk}$ is one to one and onto, if and only if $\bk\in\{\bb\Gamma,\bK,\bK^\prime\}$, where
\begin{equation}
\label{defn_K_Kp}
\bb\Gamma \equiv 0,\quad \bK\equiv\frac{1}{3}(\bk_1-\bk_2),\quad \bK' \equiv-\bK .
\end{equation}
\end{proposition}

The set of six vertices of $\mathcal{B}_h$ is invariant under $R$ and decomposes into the two subsets:
the $\bK\,\, \text{type-points:} \ \bK, \ R\bK =\bK+\bk_2, \ R^2\bK = \bK - \bk_1$ and
the $\bK'\, \,\text{type-points:} \ \bK', \ R\bK' =\bK-\bk_2,\ R^2\bK' = \bK' + \bk_1$.
We note for future reference, the action of $R$ on $\bk_1$ and $\bk_2$:
\begin{equation}
 \label{R_bk_rel}
 R\bk_1=\bk_2, \ \ \text{and} \ \ R\bk_2=-\bk_1-\bk_2.
\end{equation}

\subsection{Fourier analysis in $L^2_{\bK_\star}$}

We focus on the vertices of $\B_h$. Consider such a vertex,  $\bK_\star$. If $f$ is $\bK_\star-$ pseudo-periodic, it is easy to check that $\mathcal{R}[f]$ is also $\bK_\star-$ pseudo-periodic and that $\mathcal{R}$ is an isomorphism of $L^2_{\bK_\star}$.
Furthermore,  $\mathcal R$ is unitary with eigenvalues $1$, $\tau$ and $\bar{\tau}$, where $\tau=\exp(2\pi i/3)$. It is therefore natural to split $L^2_{\bK_\star}$ into the direct (orthogonal) sum:
\begin{equation*}
L^2_{\bK_\star} = L^2_{\bK_\star,1}\oplus L^2_{\bK_\star,\tau}\oplus L^2_{\bK_\star,\overline\tau}\ ,
\end{equation*} where  the summand subspaces are given by:
\begin{equation}
\label{L2_K_sigma}
L^2_{\bK_\star,\sigma}~=~\left\{f\in L^2_{\bK_\star}:\mathcal{R}f=\sigma f\right\}, \quad \sigma=1,\tau,\bar\tau.
\end{equation}

\subsection{$\pc-$ invariance and $\mathcal{O}-$ invariance}\label{honeycomb-structures}

Let  $A(\bx)$ be a smooth $\Lambda_h-$ periodic and Hermitian matrix function, for which $\LA=-\nabla\cdot A\nabla$ is elliptic in the sense of Assumption  \ref{A_Assumption} of Section \ref{fb_theory}. In this section we derive further conditions on $A(\bx)$ which ensure that $\LA$ is $\pc-$ invariant and $\mathcal R-$ invariant.

We make use of the following general result. Let $O$ denote a real orthogonal matrix and define
\begin{equation}
\left(\mathcal{O}f\right)(\bx)=f({O}^*\bx)\ .
\end{equation}
\begin{theorem}\label{invariance}
Let $A(\bx)$ be an arbitrary smooth matrix function.
\begin{enumerate}
\item $\left[\mathcal{PC},\LA\right]\ =\ 0$ if and only if  $A(\bx)=\overline{A(-\bx)}$.
\item $\left[\mathcal{O},\LA\right]\ =\ 0$ if and only if  $A(O^*\bx)=O^*A(\bx)O$.
\item Assume $O^TO=I$ and $O\neq \pm I$. If the matrix $A(\bx)$ is Hermitian and $\mathcal{O}-$ invariant:
\begin{equation}
\mathcal{O}[A](\bx)\equiv A(O^*\bx)=A(\bx),
\label{A-Rinvariant}
\end{equation}
 then,  $[\LA,\mathcal{O}]=0$ and $[\LA,\mathcal{PC}]=0$ if and only if
 \begin{align}
   A(\bx)&=a(\bx)I + b(\bx) \sigma_2,\nn
   \end{align}
   where $a(\bx)$ and $b(\bx)$ are scalar functions satisfying
\begin{align}
   a(O^*\bx)&=a(\bx),\qquad a(-\bx)=a(\bx),\nn\\
b(O^*\bx)&=b(\bx),\qquad  b(-\bx)=-b(\bx).\nn
   \end{align}
\end{enumerate}
\end{theorem}

The proof of Theorem \ref{invariance} makes use of the following basic lemma:
\begin{lemma}\label{PCRnabla}
\[
(a)\ \mathcal{P}\nabla=-\nabla\mathcal{P},\qquad
(b)\ \mathcal{C}\nabla=\nabla\mathcal{C},\qquad
(c)\  \mathcal{O}\nabla=O^*\nabla\mathcal{O}.
\]
\end{lemma}
\begin{proof}[Proof of Lemma \ref{PCRnabla}]  We prove part (c). For any differentiable $f$, $\partial_{x_m}\mathcal{O}[f](\bx)=\partial_{x_m}f(O^*\bx)=\frac{\partial y_l}{\partial x_m}\partial_{y_l}f(\by)\Big|_{\by=O^*\bx}$. Now, $y_l=(O^*\bx)_l=O^*_{lr}x_r$ and therefore $\frac{\partial y_l}{\partial x_m}=O_{lr}^*\delta_{rm}=O^*_{lm}$.  Therefore,
$\partial_{x_m}\mathcal{O}[f](\bx)= O^*_{lm}\partial_{y_l}f(\by)\Big|_{\by=O^*\bx}=O_{ml}\partial_{y_l}f(\by)\Big|_{\by=O^*\bx} = O_{ml}\mathcal{O}\left[\partial_{y_l}f\right](\bx)$. This
implies $O^*_{lm}\partial_{x_m}\mathcal{O}[f](\bx)= \mathcal{O}\left[\partial_{y_l}f\right](\bx)$ for each $l=1,2$,
which completes the proof.
\end{proof}

We now turn to the proof of Theorem \ref{invariance}.
\begin{proof}[Proof of Theorem \ref{invariance}]
 Part 1 is straightforward. To prove part 2, note by part (c) of Lemma \ref{PCRnabla} that
 $\partial_{x_j}\mathcal{O}[f](\bx)=\mathcal{O}\left[O_{jl}\partial_{y_l}f\right](\bx)$. Multiplying this relation by $a_{ij}(\bx)$, using that
  $\mathcal{O}[a_{ij}(O\by)](\bx)=a_{ij}(\bx)$,  and summing over $j$ gives,
 \[ a_{ij}(\bx)\partial_{x_j}\mathcal{O}[f](\bx)=\mathcal{O}\left[a_{ij}(O\by)O_{jl}\partial_{y_l}f\right](\bx).\]
 Furthermore, again using Lemma \ref{PCRnabla}, we have
 \begin{align}
(\LA\circ\mathcal{O})[f](\bx) &=\ \partial_{x_i} a_{ij}(\bx)\partial_{x_j}\mathcal{O}[f](\bx)\nn\\
                 &=\partial_{x_i} \mathcal{O}\left[a_{ij}(O\by)O_{jl}\partial_{y_l}f\right](\bx)\nn\\
& = \mathcal{O}\left[O_{im}\partial_{y_m} a_{ij}(O\by)O_{jl}\partial_{y_l}f\right](\bx)\nn\\
& =  \mathcal{O}\left[\partial_{y_m} O_{mi}^*a_{ij}(O\by)O_{jl}\partial_{y_l}f\right](\bx)\nn\\
& = (\mathcal{O}\circ\mathcal{L}^{\mathcal{O}_\star A})[f](\bx)\ ,
\end{align}
where $[\mathcal{O}_\star A](\by)\equiv O^*A(O\by)O$.
Therefore, $[\mathcal{O},\LA]=0$ if and only if $\mathcal{O}_\star A(\by)=A(\by)$. Setting $\by=O^*\bx$, it follows that
 $O^*A(\bx)O=A(O^*\bx)$. This completes the proof of part 2.

Finally, we turn to the proof of part 3. By hypothesis
 $O=\begin{pmatrix}\cos \theta& \sin \theta \\ -\sin \theta &\cos \theta,
 \end{pmatrix}$ where $\theta \in (0,2\pi)$ and $\theta\ne\pi$.  Form the $4-$component vector of entries of $A$: $\vec{a}=(a_{11},a_{12},a_{21},a_{22})^T$. Since $A(O^*\bx)=A(\bx)$, by part 2 we have the relation $A(\bx)= O^* A(\bx) O$, which is equivalent to the system $M\vec{a}=0$,  where $M$ is the $4\times4$ matrix:
 \[
 M=\sin\theta\times\begin{pmatrix}-\sin\theta & -\cos\theta & -\cos\theta & \sin\theta \\
 \cos\theta &-\sin\theta &-\sin\theta & -\cos\theta\\
 \cos\theta&-\sin\theta &-\sin\theta &-\cos\theta \\ \sin\theta&\cos \theta &\cos\theta &-\sin\theta\\
  \end{pmatrix}.
 \]

Note that the second and third rows of $M$ are identical, and  the first and fourth rows are the same up to a factor of minus one. Therefore, since $\theta\neq 0,\ \pi$, the rank of $M$ is 2.  Hence, the dimension of the nullspace of $M$ is 2. In fact, the null space of $M$ is easily seen to be spanned by the vectors: $(1,0,0,1)^T$ and $(0,-1,1,0)^T$. Thus, any element of the nullspace of $M$, may be expressed in the form $a(\bx)(1,0,0,1)^T+ ib(\bx)(0,-1,1,0)^T$, for some choice of scalar functions $a$ and $b$. Using that $A^\dagger=A$, we have $A(\bx)=a(\bx)I + b(\bx) \sigma_2$, where $a(\bx), b(\bx)$ are real-valued. Finally, by part 1 we have $\overline{A(-\bx)}=A(\bx)$, and therefore $a(\bx)$ is even and $b(\bx)$ is odd.
 This completes the proof of Theorem \ref{invariance}.
 \end{proof}

\subsection{Characterization of honeycomb structured media}\label{honey-media}

Recall that $R$ denotes the $120^\circ$ rotation matrix \eqref{Def_R} and the mapping $\mathcal{R}:g(\bx)\mapsto g(R^*\bx)$. The bulk structures we consider in this article, modeled by the operator $\LA=-\nabla\cdot A\nabla$ are assumed to satisfy:
\begin{enumerate}
\item $\LA$ self-adjoint and elliptic. Hence, $A^\dagger=A$ and ellipticity condition $(A_3)$ of Assumption \ref{A_Assumption} holds.
\item $A(\bx+\bv)=A(\bx)$ for all $\bx\in\R^2$ and $\bv\in\Lambda_h$.
\item  $[\mathcal{PC}, \LA]=0$ and hence $\overline{A(-\bx)}=A(\bx)$, by Theorem \ref{invariance} part 1.
 \item $[\mathcal{R}, \LA]=0$ and hence $A(R^*\bx)=R^*A(\bx)R$,  by Theorem \ref{invariance} part 2. \\
 \end{enumerate}

 We shall also consider the more restrictive case where, in addition to (1)-(4), $\LA=-\nabla\cdot A\nabla$ satisfies:
 \begin{enumerate}[5.]
 \item $A(R^*\bx)=A(\bx)$ for all $\bx\in\R^2$. Hence, by Theorem \ref{invariance} part 3
there exist real-valued functions $a(\bx)$ (strictly positive) and $b(\bx)$, such that
 \begin{align}
 A(\bx)&=a(\bx)I_{_{2\times2}}+b(\bx)\sigma_2, \label{ARinv}
 \end{align} where
$a(R^*\bx)=a(\bx)$, $a(-\bx)=a(\bx)$, and
$b(R^*\bx)=b(\bx)$, $b(-\bx)=-b(\bx)$.
\end{enumerate}

We shall loosely refer to media which are modeled by $\LA$, satisfying the above properties (1)-(4) or (1)-(5)
 {\it honeycomb structured media}. When we impose the stronger constraint of $\mathcal{R}-$ invariance of the matrix $A(\bx)$ itself, {\it i.e.} $A(R^*\bx)=A(\bx)$, (5), we shall point this out explicitly.

We shall next derive the form of the  Fourier series of a general  matrix function giving rise to a honeycomb structured medium. The analysis follows that of the Schr\"odinger operator case, presented in \cites{FW:12}. We begin by noting that if $\bk=m_1\bk_1+m_2\bk_2\in\Lambda_h^*$, then $R\bk=-m_2\bk_1+(m_1-m_2)\bk_2$; see \eqref{R_bk_rel}.
Therefore, the mapping $R:\Lambda_h^*\to\Lambda_h^*$ induces a  mapping $\tilde{R}:\ \mathbb Z^2\rightarrow \mathbb Z^2$,  $\tilde{R}: (m_1, m_2) \mapsto (-m_2, m_1-m_2)$.  We note $\tilde{R}^2(m_1, m_2) = (m_2-m_1, m_1)$, $\tilde{R}^3(m_1, m_2) = (m_1, m_2)$, and that $\tilde{R}^{-1}=\tilde{R}^2$.

We say that $\mathbf{m}$ and $\mathbf{n}$ are in the same equivalence class if they lie on the same $\tilde{R}$ orbit.  Introduce the subset $\mathcal{S}\subset \Z^2\setminus\{{\bm 0}\}$,  consisting  of exactly one representative from each equivalence class. With the exception of  $\mathbf{m}=0$, left fixed by $\tilde{R}$, each $\bfm\in\Z^2$ lies on a $\tilde{R}-$ orbit of length
exactly three.

If $A(\bx)$ is a smooth $\Lambda_h-$periodic matrix function, then it can be represented as a convergent Fourier series  $A(\bx)=\sum_{\bfm\in\Z^2}A_\bfm e^{i\bfm\vec{\bk}\cdot\bx}$, with matrix-valued Fourier coefficients:
\begin{equation}
\label{fourier_coef}
A_{\mathbf{m}}\equiv\frac{1}{|\Omega_h|}\int_{\Omega_h} e^{-i\bfm\vec\bk\cdot \,\by}A(\by) d\by .
\end{equation}

\begin{proposition}\label{prop_fourier_coef}
\begin{enumerate}
\item Let $A(R^*\bx) = R^*A(\bx)R$. Then,
\begin{align}
A_{\tilde{R}^{-1}\mathbf{m}}=R A_\mathbf{m}R^*,\quad \mathbf{m}\in \mathbb Z^2,
\label{ARx-fr}
\end{align}
and therefore, $A_{\tilde{R}\bfm}=R^*A_\bfm R$ and $A_{\tilde{R}^2\bfm}=RA_\bfm R^*$, $\mathbf{m}\in \mathbb Z^2$.
\item Assume  $A(R^*\bx) = A(\bx)$, and therefore by part 3 of Theorem \ref{invariance},
 $A(R^*\bx) = a(\bx)I_{_{2\times2}}+b(\bx)\sigma_2$, where $a$ is real-valued and even, and $b$ is real-valued and odd.
 Then,
$
A_\bfm=a_\bfm I_{_{2\times2}}+b_\bfm\sigma_2,
$ where the Fourier coefficients of $a$ and $b$ satisfy:
\begin{equation}
a_\bfm=a_{\tilde{R}\bfm}=a_{\tilde{R}^2\bfm}, \quad b_\bfm=b_{\tilde{R}\bfm}=b_{\tilde{R}^2\bfm}\\
\end{equation}
Furthermore, since $a(-\bx)=a(\bx),\ b(-\bx)=-b(\bx)$, we have $\overline{a_\bfm}=a_\bfm=a_{-\bfm},\ \overline{b_\bfm}=-b_\bfm=b_{-\bfm}$.
\item Let $\mathcal{PC}[A](\bx) = A(\bx)$. Then,
\begin{align}
 A_\mathbf{m}=\overline{A_\mathbf{m}},\quad  \mathbf{m}\in \mathbb Z^2.
\label{PCA-fr}\end{align}
\item Let $A^\dagger=A$. Then,
\begin{align}
A_\mathbf{m}=(A_{-\mathbf{m}})^\dagger,\ \  \mathbf{m}\in \mathbb Z^2.\label{Asa-fr}
\end{align}
\end{enumerate}
\end{proposition}

\begin{proof}[Proof of Proposition \ref{prop_fourier_coef}]
We first prove \eqref{ARx-fr}.
\begin{equation}
\begin{split}
[\mathcal{R}A]_\bfm&=\frac{1}{|\Omega_h|}\int_{\Omega_h} e^{-i\bfm\vec\bk\cdot \,\by}A(R^*\by)d\by\ =\frac{1}{|\Omega_h|}\int_{\Omega_h} e^{-iR^*(\bfm\vec\bk)\cdot \,\bb z}A(\bb z) d\bb z\\
&=\frac{1}{|\Omega_h|}\int_{\Omega_h} e^{-i(\tilde{R}^{-1}\bb m)\bb k\cdot \,\bb z}A(\bb z) d\bb z\ =\ A_{\tilde{R}^{-1}\mathbf{m}}
\end{split}
\end{equation}
Therefore, $\mathcal{R}[A](\bx)\equiv A(R^*\bx)=R^*A(\bx)R$ $\iff$ $A_{\tilde{R}^{-1}\mathbf{m}}=R A_{\mathbf{m}}R^*$.

To prove \eqref{PCA-fr}, note that
\begin{equation}
\begin{split}
(\mathcal{PC})[A]_{\bfm}&=\frac{1}{|\Omega_h|}\int_{\Omega_h} e^{-i\bfm\cdot\vec\bk \,\by}\overline{A(-\by)}d\by =\frac{1}{|\Omega_h|}\int_{\Omega_h} e^{i\bfm\vec\bk\cdot \,\bb z}\overline{A(\bb z)} d\bb z=\overline{A_{\bfm}}\\
\end{split}
\end{equation}
Therefore,
$(\mathcal{PC})[A](\bx)\equiv \overline{A(-\bx)}=A(\bx)$ $\iff$ $A_\mathbf{m}$ is real.

To prove \eqref{Asa-fr} note that $(A_\bfm)^{\dagger}=(A^\dagger)_{-\bfm}$. Hence, if $A$ is Hermitian then $A_\bfm=(A_{-\bfm})^\dagger$.
\end{proof}

Because $A(\bx)$ is $\Lambda_h-$ periodic, we can express it in the Fourier series
\begin{equation}
\begin{split}
A(\bx)&= A_{\bf 0}+\sum_{\mathbf{m}\ne{\bf 0}}A_\mathbf{m}e^{i\mathbf{m}\vec{\bk}\cdot \bx}\\
&= A_{\bf 0}+\sum_{\mathbf{m}\in\tilde{\mathcal{S}}\setminus\{\bf 0\}}
\left(\ A_\mathbf{m}e^{i\mathbf{m}\vec{\bk}\cdot \bx}+A_{\tilde{R}\mathbf{m}}e^{i(\tilde{R}\mathbf{m})\vec{\bk}\cdot \bx}
+A_{\tilde{R}^2\mathbf{m}}e^{i(\tilde{R}^2\mathbf{m})\vec{\bk}\cdot \bx}\ \right).
\end{split}
\end{equation}

Symmetries such as $\pc-$ and $\mathcal{R}-$ invariance imply further constraints on the Fourier series:
\begin{corollary}[Characterization of honeycomb structured media]\label{A-expand}
Let $A(\bx)$ be a Hermitian and $\Lambda_h-$ periodic matrix function. Then, the following holds:
\begin{enumerate}
\item
Assume  $\mathcal{PC}[A]=A$ (hence $[\pc,\LA]=0$) and $A(R^*\bx)=R^*A(\bx)R$ (hence $[\mathcal{R},\LA]=0$).  Then,
\begin{equation}
A(\bx)=a_{_{\bb 0}}I_{_{2\times2}}+\sum_{\mathbf{m}\in\tilde{\mathcal{S}}\setminus\{\bb 0\}}A_\mathbf{m}\ e^{i\mathbf{m}\vec{\bk}\cdot \bx}+R^* A_\mathbf{m}R\ e^{i(\tilde{R}\mathbf{m})\vec{\bk}\cdot \bx}+R A_\mathbf{m}R^*\ e^{i(\tilde{R}^2\mathbf{m})\vec{\bk}\cdot \bx}.
\label{Afourier}
\end{equation}
%
%
where $a_{_{\bb 0}}$ is real and positive; $A_\mathbf{m}$ is real and satisfies $A_{-\mathbf{m}}=A_{\mathbf{m}}^T$.
\item Assume further that $A(R^*\bx)=A(\bx)$. Then,
\begin{align}
A(\bx)&=a_{_{\bb 0}}I_{_{2\times2}}+
\sum_{\mathbf{m}\in\tilde{\mathcal{S}}\setminus\{\bb 0\}}\left(a_\bfm I_{_{2\times2}}+b_\bfm\sigma_2\right)
\left(e^{i\mathbf{m}\vec{\bk}\cdot \bx}+e^{i(\tilde{R}\mathbf{m})\vec{\bk}\cdot \bx}+e^{i(\tilde{R}^2\mathbf{m})\vec{\bk}\cdot \bx}\right)\nn\\
&= a_{_{\bb 0}}I_{_{2\times2}}+\sum_{\mathbf{m}\in\tilde{\mathcal{S}}\setminus\{\bb 0\}}
\begin{pmatrix}
a_\bfm & -\beta_\bfm\\ \beta_\bfm & a_\bfm
\end{pmatrix}\
\left(e^{i\mathbf{m}\vec{\bk}\cdot \bx}+e^{i(\tilde{R}\mathbf{m})\vec{\bk}\cdot \bx}+e^{i(\tilde{R}^2\mathbf{m})\vec{\bk}\cdot \bx}\right),
\label{AAfourier}\end{align}
where $a_\bfm$ and $\beta_\bfm=ib_\bfm$ ($\bfm\in\mathcal{S}\setminus\{\bb 0\}$) are real.
\item Under the assumptions of (2), $A(\bx)$ may also be rewritten as
\begin{align}
A(\bx)&=a_{_{\bb 0}}I_{_{2\times2}}\ +\
I_{_{2\times2}}\cdot \sum_{\mathbf{m}\in\tilde{\mathcal{S}}\setminus\{\bb 0\}}\ a_\bfm\
\left(\ \cos(\bfm\vec\bk\cdot\bx)\ +\  \cos((\tilde{R}\bfm)\vec\bk\cdot\bx)\ +\  \cos((\tilde{R}^2\bfm)\vec\bk\cdot\bx)\ \right)\nn\\
&\ + \sigma_2\cdot \sum_{\mathbf{m}\in\tilde{\mathcal{S}}\setminus\{\bb 0\}}\ \beta_\bfm
\left(\ \sin(\bfm\vec\bk\cdot\bx)\ +\  \sin((\tilde{R}\bfm)\vec\bk\cdot\bx)\ +\  \sin((\tilde{R}^2\bfm)\vec\bk\cdot\bx)\ \right) ,
\end{align}
where $a_\bfm$ and $\beta_\bfm$ are real.
\end{enumerate}
\end{corollary}

\begin{proof}[Proof of Corollary \ref{A-expand}] From Proposition \ref{prop_fourier_coef}, $A_{\bb 0}$ should be a real matrix satisfying $A_{\bb 0}=R A_{\bb 0} R^*$. Thus $A_{\bb 0}=a_{_{\bb 0}}I_{_{2\times2}}$ following from Theorem \ref{invariance}.
The rest of part 1 of Corollary \ref{A-expand} follows from Proposition \ref{prop_fourier_coef}, using that $R^2=R^*$. Part 2 is a consequence of \eqref{ARinv} and that $R^*\sigma_2R=\sigma_2=R\sigma_2R^*$. Part 3 follows from part 2 using that $A=\frac12(A+A^\dagger)$, since $A$ is Hermitian.
\end{proof}

Define the vector operator
\begin{equation}\label{A_def}
\mathscr{A}=A\ \frac{1}{i}\nabla+\frac{1}{i}\nabla\cdot A,
\end{equation}
with component operators:
\begin{equation}\label{Aj_def}
\mathscr{A}_j=\sum_{l=1}^2a_{jl}\ \frac{1}{i}\partial_{x_l}+\sum_{l=1}^2\ \frac{1}{i}\partial_{x_l}a_{lj} , \quad j=1,2.
\end{equation}

Recall, by part 3 of Lemma \ref{PCRnabla}, the relation $\mathcal{R}\nabla f=R^*\nabla \mathcal{R}f$, for all smooth scalar function $f$. We next show that if $A(R^*\bx)=R^*A(\bx)R$ (which implies that $[\mathcal{R},\LA]=0$, by part 2 of Theorem \ref{invariance})
that this relation holds with $-i\nabla$ replaced by $\mathscr{A}$.

\begin{lemma}\label{RcomA}
Let $A(\bx)$ be a smooth matrix function and define $\LA=-\nabla\cdot A\nabla$.
 If $\left[\mathcal{R}, \LA\right]=0$, then  $\mathcal{R}\mathscr{A}=R^* \mathscr{A}\mathcal{R}$.
Equivalently, for $j=1,2$:
\[ \mathcal{R}\mathscr{A}_j=\sum_{l=1}^2R^*_{jl}\mathscr{A}_{l}\mathcal{R} .\]
\end{lemma}

\begin{proof}[Proof of Lemma \ref{RcomA}] For all differentiable $f$,
\begin{align}
\left(\mathcal{R}\circ\mathscr{A}_j\right)[f](\bx)
&=\ \mathcal{R}\left[a_{jl}(\by)\frac{1}{i}\partial_{y_l}f+\frac{1}{i}\partial_{y_l}a_{lj}(\by)f\right](\bx)\nn\\
&=\ a_{jl}(R^*\bx)\mathcal{R}\left[\frac{1}{i}\partial_{y_l}f\right]+\mathcal{R}\left[\frac{1}{i}\partial_{y_l}a_{lj}(\by)f(\by)\right](\bx) .
\label{RcomA-step1}
\end{align}
Recall that by part 2 of Theorem \ref{invariance} and part (c) of Lemma \ref{PCRnabla} we have $A(R^*\bx)=R^*A(\bx)R$ and $\mathcal{R}\nabla f=R^*\nabla\mathcal{R}f$. Therefore, continuing from \eqref{RcomA-step1}, we have
{\footnotesize{
\begin{align}
\left(\mathcal{R}\circ\mathscr{A}_j\right)[f](\bx) &= \
R^*_{js}a_{sm}(R^*\bx)R_{ml}\cdot R_{ln}^*\frac{1}{i}\partial_{x_n}\mathcal{R}\left[f\right]+R_{ln}^*\frac{1}{i}\partial_{x_n}a_{lj}(R^*\bx)\mathcal{R}\left[f\right](\bx)\nn\\
&=\ R^*_{js}a_{sm}(\bx) \delta_{mn}\  \frac{1}{i}\partial_{x_n}\mathcal{R}\left[f\right]+R_{ln}^*\frac{1}{i}\partial_{x_n} \ R^*_{lm}a_{ms}(\bx)R_{sj}\ \mathcal{R}\left[f\right](\bx)\nn\\
&=\ R^*_{js}a_{sm}(\bx) \delta_{mn}\  \frac{1}{i}\partial_{x_n}\mathcal{R}\left[f\right]+\delta_{nm}\ \frac{1}{i}\partial_{x_n} \ a_{ms}(\bx)R_{sj}\ \mathcal{R}\left[f\right](\bx)\nn\\
&=\ R^*_{js}\left[a_{sm}(\bx)\  \frac{1}{i}\partial_{x_m}\mathcal{R}\left[f\right]+ \frac{1}{i}\partial_{x_m} \ a_{ms}(\bx)\ \mathcal{R}\left[f\right] \right]\nn\\
&=\ R^*_{js}\left[a_{sm}(\bx)\  \frac{1}{i}\partial_{x_m}+ \frac{1}{i}\partial_{x_m} \ a_{ms}(\bx) \right]\ \mathcal{R}\left[f\right]
\ =\ R^*_{js}\ \mathscr{A}_s\mathcal{R}\left[f\right] .
\label{RcomA-step2}\end{align}
}}
This completes the proof of Lemma \ref{RcomA}.
\end{proof}

\section{Dirac points}\label{dirac-pts}

In this section we review the notion of {\it Dirac points}, conical points at the intersections of two dispersion surfaces.
We then provide sufficient conditions on the multiplicity of the $L^2_{\bK_\star}-$ eigenvalues and certain non-degeneracies for the existence of Dirac points for the operator  $\LA=-\nabla\cdot A\nabla$. We use this result in the Section \ref{low-dp} to prove the existence of Dirac points for $\LA$, under the conditions that $A(\bx)$ is smooth, Hermitian and $\Lambda_h-$ periodic, and that $[ \mathcal{PC},\LA] = [\mathcal{R},\LA]=0$.

\begin{definition}[Dirac points]\label{dirac_pt_defn}%
The quasi-momentum / energy pair $(\bK_\star,E_D)\in\B_h\times\R$ is called a {\it Dirac point} if there exists $b_\star\ge1$  and Floquet-Bloch eigenpairs mappings:
\[ \bk\mapsto (\Phi_{b_\star}(\bx;\bk),E_{b_\star}(\bk))\ \ {\rm and}\ \   \bk\mapsto (\Phi_{b_\star+1}(\bx;\bk),E_{b_\star+1}(\bk)),\]
such that:
\begin{enumerate}
\item $E_D=E_{b_\star}(\bK_\star)=E_{b_\star+1}(\bK_\star)$ is a two-fold degenerate $L^2_{\bK_\star}-$ eigenvalue of $\mathcal{L}$.
\item There exist functions
\[ \Phi_1(\bx) \in L^2_{\bK_\star,\tau},\qquad \Phi_2(\bx) = \left(\mathcal{P}\mathcal{C}\right)[\Phi_1](\bx)=\overline{\Phi_1(-\bx)}\in L^2_{\bK_\star,\bar\tau},\ \]
(see \eqref{L2_K_sigma}) and $\inner{\Phi_a, \Phi_b}_{L^2_{\bK_\star}(\Omega_h)} = \delta_{ab}$, $a,b=1,2$.
such that
\[ \textrm{Nullspace}(\mathcal{L}-E_D I) =
{\rm span}\{ \Phi_1(\bx) , \Phi_2(\bx)\}.\]
\item  Let $E_-(\bk)=E_{b_*}(\bk)$ and $E_+(\bk)=E_{b_*+1}(\bk)$.
There exist constants ${\upsilon_F} >0$ and $\zeta_0>0$, and Lipschitz continuous functions $e_\pm(\bk)$, defined for $|\bk-\bK_\star|<\zeta_0$, such that $\abs{e_\pm(\bk)|\le C|\bk-\bK_\star}$, and satisfying
\footnote{In  condensed matter physics  $v_{_F}$ is known as  the ``Fermi velocity''.}
\begin{equation}
\label{cones}
\begin{split}
E_+(\bk)-E_D\ &=\ + {\upsilon_F}\
\left| \bk-\bK_\star \right|\
\left( 1\ +\ e_+(\bk) \right),\\
E_-(\bk)-E_D\ &=\ - {\upsilon_F}\
\left| \bk-\bK_\star \right|\
\left( 1\ +\ e_-(\bk) \right) .
\end{split}
\end{equation}
\end{enumerate}
\end{definition}

\subsection{Multiplicity two $L^2_{\bK_\star}-$ eigenvalues of $\mathcal{L}^A$ and conical singularities}

The following result, the analogue of Theorem 4.1 of \cites{FW:12} for Schr\"odinger operators, provides sufficient conditions for the existence of Dirac points for honeycomb-structured media defined by $\mathcal{L}^A$.

\begin{theorem}\label{prop_conical}
Let $\mathcal{L}^A=-\nabla\cdot A\nabla$, where $A(\bx)$ defines a honeycomb structured medium.
Fix $\bK_\star$, a vertex of $\mathcal{B}_h$.
Assume  that
\begin{enumerate}[label=(\subscript{A}{\arabic*})]
\item $\mathcal{L}^A$ has an $L^2_{\bK_\star,\tau}-$ eigenvalue, $ E_D$, of multiplicity one, with corresponding eigenvector $\Phi_1(\bx)$, normalized to have $L^2(\Omega_h)$ norm equal to one.
\item $\mathcal{L}^A$ has an $L^2_{\bK_\star,\bar\tau}-$ eigenvalue, $ E_D$, of multiplicity one, with corresponding eigenvector $\Phi_2(\bx)=\overline{\Phi_1(-\bx)}$.
\item $ E_D$ is not an $L^2_{\bK_\star,1}-$ eigenvalue of $\mathcal{L}^A$.
\item the following non-degeneracy condition holds:
\begin{equation}\label{lamsharp_def}
\upsilon_{_F} =\frac12\left|
\overline{\inner{\Phi_1, \mathscr{A}\Phi_2}}\ \cdot\ \begin{pmatrix}1\\i\end{pmatrix}\right|>0,
\end{equation}
where $\mathscr{A}$ is defined in \eqref{A_def}.
\end{enumerate}

Then, we have the following:
\begin{enumerate}
\item $(\bK_\star, E_D)$ is a Dirac point in the sense of Definition \ref{dirac_pt_defn}.\medskip

\item
Assume further that $A(\bx)$ defines  a real-valued honeycomb structured media.
 Then,
 \begin{itemize}
 \item[(i)]\ $(E,\Phi(\bx))$ is an $L^2_\bK$-  eigenpair of $\mathcal{L}^A$ if and only if $\left(E,(\mathcal{C}\Phi)(\bx)\right)$ is an $L^2_{\bK^\prime}$- eigenpair  of $\mathcal{L}^A$; here $\bK^\prime=-\bK$.
\item[(ii)]\ $(\bK,E_D)$ is a Dirac point with associated normalized modes $\Phi_1^\bK\in L^2_{\bK,\tau}$ and
 $\Phi_2^\bK=\left(\mathcal{C}\circ\mathcal{P}\right) \Phi_1\in L^2_{\bK,\overline\tau}$, if and only if
 $(\bK^\prime,E_D)$ is a Dirac point with associated  normalized modes $\Phi_1^{\bK^\prime}= \mathcal{C}\Phi_2^\bK \in L^2_{\bK^\prime,\tau}$ and
 $\Phi_2^{\bK^\prime}=\mathcal{C} \Phi_1^\bK\in L^2_{\bK^\prime,\overline\tau}$\ .
%
\item[(iii)]\ Since $\mathcal{R}$ commutes with $\mathcal{L}^A$, we have that $(E,\Phi(\bx;R^j\bK))$,\  are  $L^2_\bK-$ eigenpairs and $(E,\Phi(\bx;R^j\bK^\prime))$ are $L^2_{\bK^\prime}-$ eigenpairs for $j=0,1,2$.
\item[(iv)] \ ${\upsilon_F}^{\bK}={\upsilon_F}^{\bKp}$.
\end{itemize}
\end{enumerate}
\end{theorem}

\begin{remark}\label{rem_upsilon_positive} The choice of normalized eigenfunctions $\{\Phi_1,\Phi_2=\pc\Phi_1\}$ in Theorem \ref{prop_conical} is unique once the phase of $\Phi_1$ is specified.  However, ${\upsilon_F}$ is independent of the choices of phase. We may use these phase degrees of freedom to select $\Phi_1\in L^2_{\bK,\tau}$ and $\Phi_2=\pc\Phi_1\in L^2_{\bK,\bar\tau}$ such that
\begin{equation}
{\upsilon_F}=\frac12
\overline{\inner{\Phi_1, \mathscr{A}{\Phi}_2}}\ \cdot\ \begin{pmatrix}1\\i\end{pmatrix}> 0.\label{lamsharp_pos}
\end{equation}
Indeed, given  $\tilde{\Phi}_1\in L^2_{\bK_\star,\tau}$
 and $\tilde{\Phi}_2=\pc\tilde{\Phi}_1\in L^2_{\bK_\star,\bar\tau}$ satisfying the properties of Theorem \ref{prop_conical}, define $\lamsharp=\frac12
\overline{\inner{\tilde{\Phi}_1, \mathscr{A}{\mathcal{PC}\tilde{\Phi}}_1}}\ \cdot\ \begin{pmatrix}1\\i\end{pmatrix}$ and $\lamsharp=|\lamsharp|e^{i\arg\lambda_\sharp}$.
Set $\Phi_1=e^{-\frac{1}{2}i\arg\lambda_\sharp}\tilde{\Phi}_1$ and $\Phi_2=\left(\mathcal{C}\mathcal{P}\right)\Phi_1$. Then,
$\frac12\
\overline{\inner{\Phi_1, \mathscr{A}\Phi_2}}\ \cdot\ \begin{pmatrix}1\\i\end{pmatrix}=e^{-i\arg\lambda_\sharp}\ \frac12\ \overline{\inner{\tilde{\Phi}_1, \mathscr{A}\tilde{\Phi}_2}}\cdot \begin{pmatrix}1\\i\end{pmatrix}\
=|\lamsharp|>0$.
Henceforth, we shall assume this choice of  $\Phi_1$ and $\Phi_2$.
\end{remark}

\begin{remark}\label{dp-dependence}
In the statement of Theorem \ref{prop_conical}, the mappings $\bk\mapsto E_\pm$, $\bk\mapsto\Phi_\pm(\bx;\bk)$
 as well as ${\upsilon_F}$, depend on $\bK_\star$, but this dependence has been suppressed.
   We shall at times make this dependence explicit by
 writing $E^{\bK_\star}_\pm$, $\Phi^{\bK_\star}_\pm(\bx;\bk)$, and ${\upsilon_F}^{\bK_\star}$.
\end{remark}

\begin{remark}\label{dp-nonreal}
Theorem \ref{prop_conical} states that if $A(\bx)$ is real-valued honeycomb structured media, then the $E_D^{\bK}=E_D^{\bK'}$ and the associated eigenfunctions are related by symmetry.
For general complex-valued honeycomb structured media, this relation between the eigenspaces
 of Dirac points for $\bK$ and $\bK'$ does not hold.  A numerical illustration is  given in Appendix \ref{nonreal}.
\end{remark}

\begin{remark}[Wave-packets dynamics for data spectrally localized near Dirac points]\label{wavepkts}
The conical behavior of dispersion surfaces in a neighborhood of Dirac points suggests that the dynamics of wave-packets, which are initially spectrally localized near a Dirac point, are governed by a system of Dirac equations.
Specifically, consider the time-dependent 2D wave equation to which Maxwell's equations reduces (see Appendix \ref{maxwell})
\begin{equation}
\partial_t^2\psi+\LA\psi\ =\ \left(\ \partial_t^2-\nabla\cdot A\nabla\ \right) \psi\ =\ 0.
\label{LA-wave}
\end{equation}
Using the approach of \cites{FW:14} we may construct solutions for initial conditions which are spectrally localized in a neighborhood of Dirac point $(\bK_\star,E_D)$. To leading order, such solutions have a two-scale structure
\begin{equation}
\psi(\bx,t)\approx e^{-i\sqrt{E_D}\, t}\ \delta\ \left(\ \alpha_1(\bb \delta\bx,\delta t)\Phi_1(\bx)+\alpha_2(\delta\bx,\delta t)\Phi_2(\bx)\ \right).
\label{2scale-sol}\end{equation}
The expression in \eqref{2scale-sol} is of order $\delta^0$ in $H^s(\R^2),\ s\ge0$ and the corrector can be shown to be $o(1)$ in $H^s(\R^2)$
on large but finite time scales: $t\lesssim \mathcal{O}(\delta^{-2+})$. Here,
 $\alpha(\bX,T)=\left(\alpha_1(\bX,T),\alpha_2(\bX,T)\right)^{\rm tr}$, $\bX=\delta\bx=(X_1,X_2)$ and $T=\delta t$,   satisfies 2-dimensional system of massless Dirac equations
\begin{equation}
i\ 2\sqrt{E_D}\ \partial_T\begin{pmatrix}\alpha_1\\ \alpha_2\end{pmatrix}
 =\ \upsilon_{_F}\begin{pmatrix}0&i\partial_{X_1}-\partial_{X_2}\\ i\partial_{X_1}+\partial_{X_2}&0\end{pmatrix}\ \begin{pmatrix}\alpha_1\\ \alpha_2\end{pmatrix}\ .
\end{equation}

This is in contrast to the dynamics of wave-packets, which are spectral concentrated near a spectral band edge which borders a spectral gap. Consider the case where the energy at the spectral band edge is a simple Floquet-Bloch eigenvalue. In this case, wave packet initial conditions which are spectral localized near the edge energy have envelope dynamics governed by an effective Schr\"odinger equation with non-zero effective mass; see, for example,
\cites{Ashcroft-Mermin:76,Allaire-Piatnitski:05,hoefer-weinstein:11,APR:11,KMOW:18} and \cite{BS:99,BS:01,BS:03,Suslina:04,BS:06,KP:01,KP:07}.
A different asymptotic limit is the  geometrical optics / semi-classical regime, considered for Maxwell's equations, for example, in
 \cite{De_Nittis-Lein:14b,De_Nittis-Lein:17b}.
 \end{remark}
\medskip

\begin{proof}[Proof of Theorem \ref{prop_conical}]
The proof of part 2 of the theorem is straightforward.
By part 2, in order to prove part 1,  it suffices to prove all assertions for the vertex $\bK_\star=\bK$.

The proof follows a Lyapunov-Schmidt reduction strategy; see also \cites{FW:12}.
Let $\bK+\bkappa$, $|\bkappa|$ small, denote a quasi-momentum in a neighborhood of $\bK$. We seek a non-trivial solution $(\phi,E)$ of the Floquet-Bloch eigenvalue problem:
\begin{equation}\label{k_evp}
\begin{split}
\mathcal{L}^A(\bK+\bkappa)\phi(\bx) &= E\phi(\bx), \quad \bx \inr, \\
\phi(\bx+\bv) &=\phi(\bx), \quad \bx\in\R^2,\quad \bv\in \Lambda_h,
\end{split}
\end{equation}
where
$\mathcal{L}^A(\bk)=-(\nabla+i\bk)\cdot A(\bx)(\nabla+i\bk)$;
 see \eqref{Eq_Eigen_k}-\eqref{Eq_Eigen_k_op}.

For $\bkappa$ small, we perturbatively construct solutions of the Floquet-Bloch eigenvalue problem \eqref{k_evp}.
 Expanding $\mathcal{L}^A(\bK+\bkappa)$, we have
\begin{equation*}
\mathcal{L}^A(\bK+\bkappa)=\mathcal{L}^A(\bK)+\bkappa\cdot \mathscr{A}(\bK)+\bkappa^T A(\bx)\bkappa ,
\end{equation*}
where $\mathscr{A}$ is defined in \eqref{A_def} and
\begin{align}
\mathscr{A}(\bk)&= e^{-i\bk\cdot\bx} \mathscr{A} e^{i\bk\cdot\bx} = A(\bx) \frac{1}{i}(\nabla_\bx+i\bk) + \frac{1}{i}(\nabla_\bx+i\bk)\cdot A(\bx) . \label{Ak_def}
\end{align}

Let $ E^{(0)} = E_D$ denote the hypothesized double eigenvalue, and
\begin{equation}
 \label{phi1_exp}
 \phi^{(0)}(\bx)=\alpha \phi_1(\bx)+\beta\phi_2(\bx) ,
\end{equation}
with $\alpha$ and $\beta$ to be determined.
Introduce the orthogonal projections: $P_\parallel$, onto $\textrm{span}\{\phi_1,\phi_2\}$:
\[P_\parallel f\ \equiv\  \ \inner{\phi_1,f}\phi_1(\bx)\ +\ \inner{\phi_2,f}\phi_2(\bx) ,
\] and $P_\perp=I-P_\parallel$.
We seek $ E(\bK+\bkappa)$ and $\phi(\bx;\bK+\bkappa)$ in the form:
\begin{equation}\label{Eq_Expan2}
 E(\bK+\bkappa)= E_D+ E^{(1)},\quad \phi(\bx;\bK+\bkappa)=\phi^{(0)}+\phi^{(1)},\ \ P_\parallel\phi^{(1)}=0.
\end{equation}
Substituting \eqref{Eq_Expan2} into the eigenvalue problem \eqref{k_evp} we obtain
\begin{equation}
\label{eqn_for_phi1}
(\mathcal{L}^A(\bK)- E_D)\phi^{(1)}=\left(-\bkappa\cdot \mathscr{A}(\bK)-\bkappa^T A \bkappa + E^{(1)}\right)(\phi^{(0)}+\phi^{(1)}), \quad \phi^{(1)}\in L^2(\R^2/\Lambda_h).
\end{equation}

Equation \eqref{eqn_for_phi1} may equivalently expressed as the system:
\begin{eqnarray}
(\mathcal{L}^A(\bK)- E_D)\phi^{(1)}=P_\perp(-\bkappa\cdot \mathscr{A}(\bK)-\bkappa^T A \bkappa+ E^{(1)})(\phi^{(1)}+\phi^{(0)}),\label{Eq_Sol1}\\
0=P_\parallel(-\bkappa\cdot \mathscr{A}(\bK)-\bkappa^T A \bkappa+ E^{(1)})(\phi^{(1)}+\phi^{(0)})\label{Eq_Sol2}.
\end{eqnarray}

We first solve \eqref{Eq_Sol1}, for small $\bkappa$ and bounded $E^{(1)}$, to obtain $\phi^{(1)}$, which is linear in $\alpha$, $\beta$, with smooth dependence on $\bkappa$ and $E^{(1)}$.  We then substitute $\phi^{(1)}[\alpha,\beta,\bkappa,E^{(1)}]$ into \eqref{Eq_Sol2} to obtain a homogeneous linear system for $\alpha, \beta$, depending on $E^{(1)}$ and $\bkappa$. The solvability condition of this system specifies the dependence of $E^{(1)}$ on $\bkappa$.

By the elliptic theory, $R_{_\bK}(E_D)=(\mathcal{L}^A(\bK)- E_D)^{-1}$ is defined as a bounded map from $P_\perp L^2(\R^2/\Lambda_h)$ to $P_\perp H^2(\R^2/\Lambda_h)$.
In addition, the mapping
\begin{align}
f\mapsto \Xi(\bkappa,E^{(1)}) f &\equiv R_{_\bK}(E_D)\ P_\perp\left(-\bkappa\cdot \mathscr{A}(\bK)-\bkappa^T A \bkappa+ E^{(1)}\right) f\nn
\end{align}
is a bounded operator on $H^s(\R^2/\Lambda_h)$, for any $s\geq0$.
Furthermore, for $|\bkappa|+| E^{(1)}|$ sufficiently small, the operator norm of $\Xi=\Xi(\bkappa,E^{(1)}) $ is less than one, and therefore $(I-\Xi)^{-1}$ is defined as a bounded operator on $H^s(\R^2/\Lambda_h)$. Therefore,  \eqref{Eq_Sol1} has a unique solution:
\begin{equation}
\phi^{(1)} =\hat{c}[\bkappa, E^{(1)}]\phi_1\ \alpha\ +\ \hat{c}[\bkappa, E^{(1)}]\phi_2\ \beta\ ,\ P_\perp\phi^{(1)}=\phi^{(1)},
\label{phi1-def}\end{equation}
where the operator $g\mapsto \hat{c}[\bkappa, E^{(1)}]\ g$,  defined by:
\begin{equation}
\hat{c}[\bkappa, E^{(1)}]\ g=\left[I-\Xi(\bkappa,E^{(1)})\right]^{-1}\ \Xi(\bkappa,E^{(1)})\ g ,
\label{c-def}\end{equation}
maps $H^s(\R^2/\Lambda_h)$ to $H^s(\R^2/\Lambda_h)$,\ $s\ge0$.
For $j=1,2$, $(\bkappa, E^{(1)})\mapsto \hat{c}[\bkappa, E^{(1)}]\ \phi_j(\bx)$ are smooth mappings from a neighborhood of $(0,0)$ in $\R^2\times \mathbb C$ into $H^2(\R^2/\Lambda_h)$ satisfying the bound:
\begin{equation}\label{Boundc}
\norm{\hat{c}[ E^{(1)},\bkappa]\ \phi_j}_{H^2(\R^2/\Lambda_h)}\le C\left(|\bkappa|+|E^{(1)}|\right), \quad j=1,2.
\end{equation}
and $P_\parallel \hat{c}[\bkappa, E^{(1)}]\phi_j=0$.

Substituting \eqref{phi1-def} into \eqref{Eq_Sol2}, we obtain a system of two homogeneous linear equations for $\alpha$ and $\beta$:
\begin{equation*}
\mathcal M( E^{(1)}, \bkappa)\begin{pmatrix}\alpha\\ \beta\end{pmatrix}=0 ,
\end{equation*}
where $\mathcal M( E^{(1)}, \bkappa)$ is a $2\times2$ matrix:
\begin{equation}\label{Mek}
\mathcal M( E^{(1)}, \bkappa)= E^{(1)}I_{_{2\times2}}-\mathcal{M}_{\mathscr{A}}(\bkappa)-\mathcal{M}_R[ E^{(1)},\bkappa] .
\end{equation}

Recall that $\phi_j=e^{-i\bK\cdot\bx}\Phi_j$, $j=1,2$ and define
\begin{equation}
\mathcal{\hat{C}}[\bkappa,E^{(1)}] g\ \equiv\ e^{i\bK\cdot\bx}\ \hat{c}[\bkappa,E^{(1)}]\ e^{-i\bK\cdot\bx}\ g .
\end{equation}
Hence, $\inner{\Phi_j,\hat{C}[\bkappa,E^{(1)}] g}=0,\ j=1,2$. Expressed in terms of $\Phi_j,\ j=1,2$, the matrices $\mathcal{M}_{\mathscr{A}}$
and $\mathcal{M}_{R}$ are given by:
\begin{equation}\label{M_A-def}
\mathcal{M}_{\mathscr{A}}(\bkappa)
=\begin{pmatrix}\inner{\Phi_1,\bkappa\cdot \mathscr{A} \Phi_1}&\inner{\Phi_1,\bkappa\cdot \mathscr{A} \Phi_2}\\
\inner{\Phi_2,\bkappa\cdot \mathscr{A} \Phi_1}&\inner{\Phi_2,\bkappa\cdot \mathscr{A} \Phi_2}\end{pmatrix} ,\ \textrm{and}
\end{equation}
\begin{align}\label{M_R-def}
&\mathcal M_R[ E^{(1)},\bkappa] \\
&=
\begin{pmatrix}
\inner{\Phi_1,\left[\left(\bkappa\cdot\mathscr{A}+\bkappa^T A\bkappa\right)\hat{C}+
\bkappa^T A\bkappa\right]\Phi_1} &
\inner{\Phi_1,\left[\left(\bkappa\cdot\mathscr{A}+\bkappa^T A\bkappa\right)\hat{C}+
\bkappa^T A\bkappa\right]\Phi_2} \\
\inner{\Phi_2,\left[\left(\bkappa\cdot\mathscr{A}+\bkappa^T A\bkappa\right)\hat{C}+
\bkappa^T A\bkappa\right]\Phi_1}
&
\inner{\Phi_2,\left[\left(\bkappa\cdot\mathscr{A}+\bkappa^T A\bkappa\right)\hat{C}+
\bkappa^T A\bkappa\right]\Phi_2}
 \end{pmatrix},
\nn\end{align}
where $\hat{C}=\hat{C}[\bkappa,E^{(1)}]$. Furthermore,
by \eqref{Boundc},
\begin{equation} \mathcal M_{R,lm}[ E^{(1)},\bkappa] = \mathcal{O}\left(|\bkappa|\left(|\bkappa|+E^{(1)}\right)\right),\ l,m=1,2.\label{M_R-est}\end{equation}

We have therefore obtained the following result which characterizes the dispersion surfaces in a neighborhood of Dirac points:
\begin{proposition}\label{det0} For $|\bkappa|$ sufficiently small,
$ E= E_D+ E^{(1)}$ is a $L^2_{_{\bK+\bkappa}}-$ eigenvalue problem \eqref{k_evp} if and only if
\begin{equation}\label{Det}
\det\mathcal M( E^{(1)},\bkappa)=0,
\end{equation}
where $M( E^{(1)},\bkappa)$ is given by \eqref{Mek}, \eqref{M_A-def} and \eqref{M_R-def}.
\end{proposition}

We shall next exploit  symmetry to obtain simplified expressions for the entries of $\mathcal{M}_{\mathcal{A}}(\bkappa)$.
\begin{proposition}\label{prop_g}
Let $\bK_\star$ denote any vertex of the Brillouin zone, $\mathcal{B}_h$. Let $A(\bx)$ define a honeycomb structured medium; see Section \ref{honey-media}.
 Let  $\Phi_1$ and $\Phi_2$ be as  in the Theorem \ref{prop_conical}, $\mathscr{A}$ be given by \eqref{A_def} and $\bkappa=(\kappa^{(1)},\kappa^{(2)})\in\R^2-\{\bs 0\}$. Then,
\begin{enumerate}
\item $\inner{\Phi_1,\bkappa\cdot \mathscr{A} \Phi_1}
=\inner{\Phi_2,\bkappa\cdot \mathscr{A} \Phi_2}=0$.
\item $\inner{\Phi_1,\bkappa\cdot \mathscr{A} \Phi_2}
=\overline{\inner{\Phi_2,\bkappa\cdot \mathscr{A} \Phi_1}}=\upsilon_{_F}(\kappa^{(1)}+i\kappa^{(2)}),$
where $\upsilon_{_F}$ is given by \eqref{lamsharp_def}.
\end{enumerate}
Thus,
\begin{equation}
\mathcal{M}_{\mathscr{A}}(\bkappa)
=\upsilon_{_F}\ \begin{pmatrix}0 & \kappa^{(1)}+i\kappa^{(2)}\\
\kappa^{(1)}-i\kappa^{(2)}&0\end{pmatrix}.
\label{MA-simp}\end{equation}
\end{proposition}

\begin{proof}[Proof of Proposition \ref{prop_g}]
Without loss of generality, let $\bK_\star=\bK$.
 Since $\Phi_1\in L^2_{\bK,\tau}$ and $\Phi_2\in L^2_{\bK,\overline{\tau}}$, we have $\mathcal{R}\Phi_1=\tau\Phi_1$ and $\mathcal{R}\Phi_2=\overline{\tau}\Phi_2$.
Let $\nu_1=\tau$ and $\nu_2=\overline{\tau}$.
 For $l,m=1,2$ consider the column vector $\inner{\Phi_l,\mathscr{A}\Phi_m}$.
Applying  Lemma \ref{RcomA} we have
\begin{equation*}
\inner{\Phi_l,\mathscr{A}\Phi_m}=\inner{\mathcal R \Phi_l, \mathcal R \mathscr{A} \Phi_m}=\inner{\mathcal R \Phi_l, R^*\mathscr{A} \mathcal R \Phi_m}=\overline{\nu_l}\nu_m R^*\inner{\Phi_l,\mathscr{A} \Phi_m},\ l,m=1,2\ .
\end{equation*}
Hence, $R \inner{\Phi_l,\mathscr{A}\Phi_m} = \overline{\nu_l}\nu_m\ \inner{\Phi_l,\mathscr{A}\Phi_m}$.

If $l=m$, then $\overline{\nu_l}\nu_m=|\nu_m|^2=1$. Therefore, $R\inner{\Phi_m,\mathscr{A}\Phi_m}_{L^2_\bK}=\inner{\Phi_m,\mathscr{A}\Phi_m}$. Since $1$ is not an eigenvalue of  $R$, it follows that
\[\inner{\Phi_1,\mathscr{A}\Phi_1}_{L^2_\bK}=\inner{\Phi_2,\mathscr{A}\Phi_2}_{L^2_\bK}=0.\]
 If $l=1$ and $m=2$, then $R\inner{\Phi_1,\mathscr{A}\Phi_2}_{L^2_\bK}=\tau \inner{\Phi_1,\mathscr{A}\Phi_2}_{L^2_\bK}$. By \eqref{R_evals_evecs},
 \begin{equation}
\label{lambda_inner_prod}
\inner{\Phi_1, \mathscr{A}\Phi_2}_{L^2_\bK}=\upsilon_{_F}\begin{pmatrix}1\\i\end{pmatrix}.
\end{equation}
Note, by hypothesis $(A4)$ and Remark \ref{rem_upsilon_positive}, we may take $\upsilon_{_F}>0$.
 Finally, if $l=2$ and $m=1$
\begin{equation*}
  \inner{\Phi_2,\mathscr{A}\Phi_1}=\inner{\mathscr{A}\Phi_2,\Phi_1}=\overline{\inner{\Phi_1,\mathscr{A}\Phi_2}}
  ={\upsilon_{_F}}\begin{pmatrix}1\\-i\end{pmatrix} .
\end{equation*}
This completes the proof of Proposition \ref{prop_g}.
\end{proof}

It follows that the eigenvalue condition  $\det\mathcal{M}( E^{(1)},\bkappa) =0$ (\eqref{Det}) is of the form:
$\left( E^{(1)}\right)^2={\upsilon_F}^2|\bkappa|^2+g_{21}(E^{(1)},\bkappa)+g_{12}(E^{(1)},\bkappa)+g_{03}(\bkappa),$
where the functions $g_{rs}$ are smooth and satisfy the bounds
$
|g_{rs}(E^{(1)},\bkappa)|\le C |E^{(1)}|^r\ |\bkappa|^s,
$
for $|E^{(1)}|\le1,$ $|\bkappa|\le1$.
 The  proof of Proposition \ref{prop_conical} is  completed along the lines of Proposition 4.2 of \cites{FW:12} and yields
 two locally conical solution branches: $E_+(\bK+\bkappa)=E_D+E^{(1)}_{+}(\bkappa)=E_D+{\upsilon_{_F}}\ |\bkappa|(1+e_+(\bK+\bkappa))$ and $E_-(\bK+\bkappa)=E_D+E^{(1)}_{-}(\bkappa)=E_D-{\upsilon_{_F}}\ |\bkappa|(1+e_-(\bK+\bkappa))$.
\end{proof}

\subsection{Dirac points in low-contrast media}\label{sec:dp-low}

Any low-contrast honeycomb structured medium in the sense of Section \ref{honey-media}
is of the form  $A^{(\eps)}(\bx)=a_0I+\eps A^{(1)}(\bx)$ where $a_0>0$ and where $\eps/a_0 >0$ is small.
Indeed, suppose $A^{(\eps)}(\bx)=A^{(0)}+\eps A^{(1)}(\bx)$ with $A^{(0)}$ constant. By Theorem \ref{invariance}, $A^{(0)}=a_0I+b_0\sigma_2$ where $a_0>0$ and $b_0$ is an odd function. But $b_0$ constant and odd implies that $b_0=0$.

In this section we apply Theorem \ref{prop_conical} to the study of Dirac points for low-contrast structures. We consider the Floquet-Bloch eigenvalue problem  for the operator
$\mathcal{L}^{(\eps)}=-\nabla\cdot A^{(\eps)}(\bx)\nabla$. By linearity, we may make the replacements: $A^{(\eps)}(\bx) \to  \frac{1}{a_0}A^{(\eps)}(\bx)$ and $E\to  E/a_0$. Hence without loss of generality, we assume
\begin{equation}
\label{normalized_honeycomb}
 A^{(\eps)}(\bx)=I+\eps A^{(1)}(\bx) .
\end{equation}
For small $\eps$, we study $\bK-$ pseudo-periodic eigenvalue problem
\begin{equation}
\label{eps_evp}
\mathcal{L}^{(\eps)} \Phi= E\Phi, \quad \Phi\in L^2_{\bK}(\R^2/\Lambda_h),
\end{equation}
where
\begin{align}\label{L_eps_def}
\mathcal{L}^{(\eps)}\equiv\mathcal{L}^{A^{(\eps)}} &=
 -\Delta-\eps\nabla\cdot A^{(1)}\nabla = -\Delta+\eps \mathcal{L}^{(1)}.
\end{align}
We shall solve the eigenvalue problem \eqref{eps_evp} for small $\eps$. We begin by summarizing the  relevant spectral properties of $\mathcal{L}^{(0)}=-\Delta$ in  $L^2_{\bK}(\R^2/\Lambda_h)$; see \cites{FW:12,berkolaiko-comech:15}:

\begin{proposition}
\label{eps_0_prop}
\begin{enumerate}
\item $E^{(0)}=|\bK|^2$ is an $L^2_\bK-$ eigenvalue of $-\Delta$ of multiplicity of three with corresponding three-dimensional eigenspace:
\begin{equation*}
  \emph{span}\left\{e^{i\bK\cdot \bx},e^{iR\bK\cdot\bx},e^{iR^2\bK\cdot\bx}\right\}.
\end{equation*}
\item Considered in each of the three orthogonal $\mathcal{R}$-invariant subspaces, $L^2_{\bK,\sigma}$, $\sigma=1,\tau,\bar{\tau}$, $ E^{(0)}=|\bK|^2$ is a simple eigenvalue of
$-\Delta$ with corresponding eigenspace spanned by the $L^2(\R^2/\Lambda_h)$ normalized eigenvector:
\begin{equation}
  \Phi^{(0)}_\sigma \equiv \frac{1}{\sqrt{3|\Omega_h|}}\left[e^{i\bK\cdot\bx}+\bar{\sigma} e^{iR\bK\cdot\bx}+\sigma e^{iR^2\bK\cdot\bx}\right] \in\ L^2_{\bK,\sigma} .
\label{Phi0sig}\end{equation}
\end{enumerate}
\end{proposition}

We next turn to the eigenvalue problem \eqref{eps_evp} for $\eps$ small.

\begin{theorem}\label{low-dp}
Let $A^{(\eps)}(\bx)=I_{_{2\times2}}+\eps A^{(1)}(\bx)$ define a honeycomb structured medium, with Fourier expansion characterized by \eqref{AAfourier} (Corollary \ref{A-expand}).
Assume that $A^{(\eps)}(\bx)$ satisfies the non-degeneracy condition:
\begin{equation}
\bK^T\ A_{0,-1}\ R\bK\neq0\label{non-deg}.
\end{equation}

Then, there exists $\eps^0>0$, and mappings
$\eps\mapsto E_D^\eps$ and $\eps\mapsto \widetilde{E}^\eps$,
 $\eps\mapsto \Phi_1^{\eps}\in L^2_{\bK,\tau}$,  $\eps\mapsto \Phi_2^{\eps}\in L^2_{\bK,\bar\tau}$
  and $\eps\mapsto \widetilde{\Phi}^{(\eps)}\in L^2_{\bK,1}$
 such that the following holds for all $\eps\in(-\eps^0,\eps^0)$:
 \begin{enumerate}
\item $ E^{\eps}_D$ is a simple $L^2_{\bK,\tau}-$ eigenvalue of $\mathcal{L}^{(\eps)}$ with
\begin{equation}
 E_D^\eps= \ |\bK|^2+\eps\left(\bK^T\ A^{(1)}_{0,0}\ \bK-\bK^T\ A^{(1)}_{0,-1}\ R\bK\ \right)+\mathcal{O}(\eps^2) ,\label{dp-eps}
\end{equation}
with eigenspace spanned by states $\Phi^{\eps}_1\in L^2_{\bK,\tau}$
\item $ E_D^\eps$ is also a simple $L^2_{\bK,\overline\tau}-$ eigenvalue of $\mathcal{L}^{(\eps)}$
with corresponding eigenfunction $\Phi^{\eps}_2=\left(\mathcal{C}\circ\mathcal{P}\right) \Phi_1^\eps\in L^2_{\bK,\bar{\tau}}$.
\item $\widetilde{ E}^\eps\ne E_D^\eps$ is a simple $L^2_{\bK, 1}-$ eigenvalue of $\mathcal{L}^{(\eps)}$:
\begin{equation}
\widetilde{ E}^{\eps}=|\bK|^2
+\eps\ \left(\bK^T A^{(1)}_{0,0}\ \bK+2\bK^T\ A^{(1)}_{0,-1}\ R\bK\  \right)+\mathcal{O}(\eps^2),
\label{Etilde-eps}\end{equation}
with corresponding eigenspace spanned by $\tilde{\Phi}^{(\eps)}\in L^2_{\bK,1}$.
\item
\begin{align}
{\upsilon^\bK_{_F}}(\eps)=\left| \frac12 \overline{\inner{\Phi^{\eps}_1,\mathscr{A^{(\eps)}}\Phi^{\eps}_2}} \cdot \begin{pmatrix}1\\i\end{pmatrix} \right|= \frac{4\pi}{3}+\mathcal{O}(\eps).
\label{vF-eps}
\end{align}
Therefore, $(\bK,E^\eps_D)$ is a Dirac point in the sense of Definition \ref{dirac_pt_defn}.
\item \subitem If $\eps\bK^T\ A^{(1)}_{0,-1}\ R\bK>0$, the Dirac points occurs at the intersection of the $1^{st}$ and $2^{nd}$ dispersion surfaces at the vertices of $\B_h$. \subitem If $\eps\bK^T\ A^{(1)}_{0,-1}\ R\bK<0$, Dirac points occur at the intersection of the $2^{nd}$ and $3^{rd}$ dispersion surfaces at the vertices of $\B_h$.
\end{enumerate}
\end{theorem}

\begin{proof}
Statements (1)-(4) of the theorem imply that for $\eps\in(-\eps_0,\eps_0)$, assumptions $(A_1)$-$(A_4)$ of Theorem \ref{prop_conical} hold. These statements therefore imply, by Proposition \ref{prop_conical}, the existence of Dirac points $(\bK_\star,E^{\eps}_D))$ at the vertices of $\B_h$.

Consider the eigenvalue problem in $L^2_{\bK,\sigma}$ for $\sigma=1,\tau, \bar{\tau}$:
\begin{equation}\label{Eq_EigenL}
(-\Delta+\eps \mathcal{L}^{(1)})\Phi_\sigma^{(\eps)}= E_\sigma^{\eps}\Phi_\sigma^{(\eps)}, \quad \Phi_\sigma^{(\eps)} \in L^2_{\bK,\sigma},
\end{equation}
where $\mathcal{L}^{(1)}=-\nabla\cdot A^{(1)}\nabla$.
We seek $ E_\sigma^{\eps}$ and $\Phi_\sigma^{(\eps)}$ in form of expansions:
\begin{equation}\label{Eq_Expan1}
   E_\sigma^{\eps}= E^{(0)} + \eps E_\sigma^{(1,\eps)},\quad \Phi_\sigma^\eps=\Phi_\sigma^{(0)} + \eps \Phi_\sigma^{(1)},\quad
   \left\langle \Phi_\sigma^{(0)},\Phi_\sigma^{(1)}\right\rangle=0,
\end{equation}
where $ E^{(0)}=|\bK|^2$ (Proposition \ref{eps_0_prop}) and $\Phi_\sigma^{(0)}\in L^2_{\bK,\sigma}$, given by the expression in
\eqref{Phi0sig},
satisfies $-\Delta\Phi_\sigma^{(0)}= E^{(0)}\Phi_\sigma^{(0)}$. Then, $\Phi_\sigma^{(1)}$ satisfies the non-homogeneous equation:
\begin{equation}
\label{F_eqn}
\left(-\Delta- E^{(0)}\right)\Phi_\sigma^{(1)}=\left( -\mathcal{L}^{(1)}+ E^{(1)}\right)\left(\Phi_\sigma^{(0)}+\eps\Phi_\sigma^{(1)}\right) .
\end{equation}

 Introduce the orthogonal projections: $P_\parallel$ onto span$\{\Phi_\sigma^{(0)}\}$, the $L^2_{\bK,\sigma}-$ nullspace of $(-\Delta- E^{(0)})$,  and $P_{\perp}=I-P_\parallel$.
Equation \eqref{F_eqn} may then be rewritten as the following equivalent system for $\Phi_\sigma^{(1)}$ and $ E^{(1)}$:
\begin{eqnarray}
\left(-\Delta- E^{(0)}\right)\Phi_\sigma^{(1)}=P_\perp \left( -\mathcal{L}^{(1)}+ E^{(1)}\right)\left(\Phi_\sigma^{(0)}+\eps\Phi_\sigma^{(1)}\right),\label{Eq_Phi1}\\
0=P_\parallel \left( -\mathcal{L}^{(1)}+ E^{(1)}\right)\left(\Phi_\sigma^{(0)}+\eps\Phi_\sigma^{(1)}\right)\label{Solvability}.
\end{eqnarray}
In analogy with the reduction strategy for the proof of Proposition \ref{prop_conical}, we first solve \eqref{Eq_Phi1} to obtain a solution for $\Phi_\sigma^{(1)}$ as a smooth functional of $\eps$ and $E^{(1)}$, and then substitute the solution, $\Phi_\sigma^{(1)}=\Phi_\sigma^{(1)}[\eps,E^{(1)}]$,  into \eqref{Solvability} to obtain a closed equation for $E^{(1)}$ as a function of $\eps$.

Because $\left(-\Delta- E^{(0)}\right)^{-1}$ is a bounded operator from $P_\perp L_{\bK, \sigma}^2$ to $P_\perp H_{\bK,\sigma}^2$, equation \eqref{Eq_Phi1} may be rewritten as
\begin{equation}\label{Eq_Phi12}
  \left(I-\eps\mathcal{W}(E^{(1)})\right)\Phi^{(1)}_\sigma
  =  -\left(-\Delta- E^{(0)}\right)^{-1}P_\perp\mathcal{L}^{(1)} \Phi^{(0)}_\sigma,
\end{equation}
where the operator $\mathcal{W}(E^{(1)})$, defined by:
\begin{equation*}
 f\mapsto \mathcal{W}(E^{(1)}) f\equiv  \left(-\Delta- E^{(0)}\right)^{-1} P_\perp\left(-\mathcal{L}^{(1)} + E^{(1)}\right),
\end{equation*}
is a bounded on $H^s_{\bK,\sigma}$, for any $s\geq0$. 
Furthermore, for $|\eps|+ |E^{(1)}|$ sufficiently small, the operator norm of $\eps\mathcal{W}$ is less than one. Hence,  $(1-\eps\mathcal{W}(E^{(1)}))^{-1}$ exists, is bounded on $H^s_{\bK,\sigma}$,  and equation \eqref{Eq_Phi12} is uniquely solvable in $H^2_{\bK}$:
\begin{equation}\label{Eq_Phi123}
\Phi^{(1)}_\sigma[\eps,E^{(1)}]= - \left(I-\eps\mathcal{W}(E^{(1)})\right)^{-1} \left(-\Delta- E^{(0)}\right)^{-1} P_\perp \mathcal{L}^{(1)}\Phi_\sigma^{(0)}\ .
\end{equation}

Substituting \eqref{Eq_Phi123} into the solvability equation \eqref{Solvability}, we obtain a closed equation for $E^{(1)}$:
\begin{equation}
\label{G_sigma_def}
G_\sigma(\eps, E^{(1)}) \equiv \inner{\Phi^{(0)}_\sigma, \left(-\mathcal{L}^{(1)}+ E^{(1)} \right) \left(\Phi^{(0)}_\sigma + \eps\Phi^{(1)}_\sigma[\eps,E^{(1)}] \right)}_{L^2_{\bK}} = 0.
\end{equation}
For $\eps=0$, equation \eqref{G_sigma_def} reduces to
$G_\sigma(0, E^{(1)}) = -\inner{\Phi^{(0)}_\sigma, \mathcal{L}^{(1)} \Phi^{(0)}_\sigma} + E^{(1)} = 0$. Thus,
$G_\sigma(0, E_\sigma^{(1,0)})=0$, where
\begin{equation}
E_{\sigma}^{(1,0)}\equiv \inner{\Phi^{(0)}_\sigma, \mathcal{L}^{(1)} \Phi^{(0)}_\sigma}.
\label{Esigma01}\end{equation}
  Moreover,
$\D_{E^{(1)}} G_\sigma(0,E_{\sigma}^{(1,0)})=1\ne0$. Hence, by the implicit function theorem, there is a unique mapping $\eps\mapsto E_\sigma^{(1,\eps)}$, defined and analytic in a neighborhood of $\eps=0$, which satisfies $E_\sigma^{(1,\eps=0)}=E^{(1,0)}_{\sigma}$ and  $G_\sigma(\eps, E_\sigma^{(1,\eps)})=0$, for all  $\eps$ in a complex neighborhood of $\eps=0$.

Assertions 1-3 of  Theorem \ref{low-dp} are now a direct consequence of the evaluation of $E^{(1,0)}_{\sigma}$, given in the following:
\begin{proposition}\label{prop_per2}
Recall that $\mathcal{L}^{(\eps)}= -\Delta + \eps\mathcal{L}^{(1)}=-\Delta-\eps\nabla\cdot A^{(1)}\nabla$, and
for $\bfm=(m_1,m_2)\in\Z^2$, denote by $A^{(1)}_\bfm$ the Fourier coefficients of $A^{(1)}(\bx)$:
\begin{equation*}
  A^{(1)}_\bfm \equiv \frac{1}{|\Omega_h|} \int_{\Omega_h} e^{-i\bfm\vec\bk\cdot \,\by} A^{(1)}(\by) d\by.
\end{equation*}
We have,
\begin{equation}
E^\eps_\sigma\ =\ |\bK|^2\ +\ \eps\ E_\sigma^{(1,\eps)}\ =\ |\bK|^2\ +\ \eps\ E_\sigma^{(1,0)}+\mathcal{O}(\eps^2),
\label{Eeps-sigma}
\end{equation}
where the following assertions concerning $E^{(1,0)}_{\sigma}$ hold:
\begin{enumerate}
\item If $[\pc,\mathcal{L}^{(1)}]=0$ and $[\mathcal{R},\mathcal{L}^{(1)}]=0$ (hence $A^{(1)}(R^*\bx)=R^*A^{(1)}(\bx)R$), then
\begin{equation}
\label{phi_sigma_identity}
E^{(1,0)}_{\sigma}\ =\ \inner{\Phi^{(0)}_\sigma, \mathcal{L}^{(1)} \Phi^{(0)}_\sigma} =\bK^T\ A^{(1)}_{0,0}\ \bK\ +\
 (\sigma+\overline{\sigma})\ \bK^T A^{(1)}_{0,-1}\ R\bK
.
\end{equation}
\item Assume further that $A^{(1)}(R^*\bx)=A^{(1)}(\bx)$. (Therefore, by Theorem \ref{invariance},  $A^{(1)}(\bx)=a^{(1)}(\bx)\ I_{_{2\times2}}\ +\ b^{(1)}(\bx)\sigma_2$ with $a^{(1)}(-\bx)=a^{(1)}(\bx)$ and $b^{(1)}(-\bx)=-b^{(1)}(\bx)$, and then $A^{(1)}_{0,-1}=a_{0,-1}^{(1)}(\bx)\ I_{_{2\times2}}\ +\ b_{0,-1}^{(1)}(\bx)\sigma_2$ ) \\
Then,
\begin{equation}
E^{(1,0)}_{\sigma}\ =\ \inner{\Phi^{(0)}_\sigma, \mathcal{L}^{(1)} \Phi^{(0)}_\sigma}=|\bK|^2\left(a^{(1)}_{0,0}-\left(\frac{1}{2}a^{(1)}_{0,-1}-\frac{\sqrt{3}}{2}ib^{(1)}_{0,-1}\right)(\sigma+\overline{\sigma})\right).
\label{ab11} \end{equation}
 by $\bK^T \ R \bK =-\frac{1}{2}|\bK|^2$ and $\bK^T \ \sigma_2 R\bK = \frac{\sqrt{3}}{2} i$.
 \end{enumerate}

\end{proposition}

We prove Proposition \ref{prop_per2} below. We first conclude the proof of Theorem \ref{low-dp} by verifying the assertion in part 4.
  By \eqref{Eq_Expan1} and \eqref{Eq_Phi123},
\begin{equation}
\label{phi_eps}
 \Phi_\sigma^{(\eps)} = \Phi^{(0)}_\sigma + \mathcal{O}(\eps),
\end{equation}
where $\Phi^{(0)}_\sigma$ is given in \eqref{Phi0sig}.
Note that $\Phi_\tau^{(\eps)}=\Phi^{\eps}_1$ and $\Phi_{\overline{\tau}}^{(\eps)}=\Phi^{\eps}_2$.
Substituting \eqref{phi_eps} into the definition of $\upsilon_F^{\bK}$ (\eqref{lamsharp_def}) and recalling that $A(\bx)=I_{_{2\times2}}+\eps A^{(1)}(\bx)$ (\eqref{normalized_honeycomb}) yields:
\begin{align*}
\upsilon_{_F}^{\bK}(\eps) &=
 \frac12\left| \overline{\inner{\Phi^{\eps}_1,\mathscr{A}^\eps\Phi^{\eps}_2}_{L^2_\bK}}\cdot\begin{pmatrix}1\\i\end{pmatrix} \right|\\
&= \left|\overline{\inner{\Phi_\tau^{(0)},-i\nabla\Phi_{\overline{\tau}}^{(0)}}}_{L^2_\bK} \cdot \begin{pmatrix}1\\i\end{pmatrix}\right| + \mathcal{O}(\eps) \\
&=\left| \frac{1}{3} \left(1 + \tau R + \bar\tau R^2\right)\bK \cdot \begin{pmatrix}1\\i\end{pmatrix}\right| + \mathcal{O}(\eps)
\\
&= \left|\frac{1}{3} \left(\tau\bk_2 - \bar\tau\bk_1 \right) \cdot \begin{pmatrix}1\\i\end{pmatrix} \right|+ \mathcal{O}(\eps) \\
&= \left|\frac{2\pi}{3} \begin{pmatrix}i\\1\end{pmatrix} \cdot \begin{pmatrix}1\\i\end{pmatrix} \right|+ \mathcal{O}(\eps)\ =\ \frac{4\pi}{3} + \mathcal{O}(\eps) .
\end{align*}
It follows that  for $\eps$ sufficiently small, $\upsilon_{_F}^{\bK}(\eps)\ne0$.
This completes the proof of part 4, and therewith Theorem \ref{low-dp}.
\end{proof}

It remains to prove Proposition \ref{prop_per2}.

\begin{proof}[Proof of Proposition \ref{prop_per2}]
Since $[\mathcal{R},\mathcal{L}^{(1)}]=0$,
\begin{equation}
 \alpha \equiv \left\langle e^{i\bK\cdot \bx}, \mathcal{L}^{(1)}e^{i\bK\cdot \bx}\right\rangle =
 \left\langle e^{iR\bK\cdot \bx}, \mathcal{L}^{(1)}e^{iR\bK\cdot \bx}\right\rangle = \left\langle e^{iR^2\bK\cdot \bx}, \mathcal{L}^{(1)}e^{iR^2\bK\cdot \bx}\right\rangle .
\end{equation}
Note that
$
\alpha=\left\langle e^{i\bK\cdot \bx}, \mathcal{L}^{(1)}e^{i\bK\cdot \bx}\right\rangle =\int_{\Omega_h}\bK\cdot A^{(1)}(\bx)\bK d\bx=|\Omega_h|\ \bK^T\ A^{(1)}_{0,0}\ \bK.
$
Again since  $[\mathcal{R},\mathcal{L}^{(1)}]=0$,
\begin{equation}
\beta\equiv \left\langle e^{i\bK\cdot \bx}, \mathcal{L}^{(1)}e^{iR\bK\cdot \bx}\right\rangle = \left\langle
e^{iR\bK\cdot \bx}, \mathcal{L}^{(1)}e^{iR^2\bK\cdot \bx}
\right\rangle = \left\langle  e^{iR^2\bK\cdot \bx}, \mathcal{L}^{(1)}e^{i\bK\cdot \bx}\right\rangle,
\end{equation}
and
\begin{equation}
 \gamma \equiv \left\langle e^{i\bK\cdot \bx}, \mathcal{L}^{(1)}e^{iR^2\bK\cdot \bx}\right\rangle = \left\langle e^{iR\bK\cdot \bx}, \mathcal{L}^{(1)}e^{i\bK\cdot \bx}\right\rangle = \left\langle e^{iR^2\bK\cdot \bx}, \mathcal{L}^{(1)}e^{iR\bK\cdot \bx}\right\rangle .
\end{equation}

Since $\mathcal{L}^{(1)}$ is self-adjoint,
 we have
 $
\gamma=\left\langle e^{i\bK\cdot \bx}, \mathcal{L}^{(1)}e^{iR\bK\cdot \bx}\right\rangle =\overline{\left\langle e^{iR\bK\cdot \bx}, \mathcal{L}^{(1)}e^{i\bK\cdot \bx}\right\rangle}=\overline{\beta}$.
Moreover,
\begin{equation}
\beta=\left\langle e^{i\bK\cdot \bx}, \mathcal{L}^{(1)}e^{iR\bK\cdot \bx}\right\rangle=\int_{\Omega_h}\bK\cdot A^{(1)}(\bx)R\bK e^{i\bk_2\cdot \bx}d\bx=|\Omega_h|\bK^T\ A^{(1)}_{0,-1}\ R\bK,
\end{equation}
where we have used $\mathbf{R}\bK-\bK=\bk_2$.
From the $\mathcal{PC}$ symmetry, we know that $A^{(1)}_{0,-1}$ is real by Corollary \ref{A-expand}. So $\gamma=\beta$ is real.

The preceding discussion yields
\begin{equation}
\begin{split}
\left\langle \Phi^{(0)}_\sigma, \mathcal{L}^{(1)}\Phi^{(0)}_\sigma \right\rangle
&=\frac{1}{3|\Omega_h|}\begin{pmatrix}1&\sigma&\overline{\sigma}\end{pmatrix}\begin{pmatrix}\alpha&\beta&\beta\\ \beta&\alpha&\beta\\ \beta&\beta&\alpha\end{pmatrix}\begin{pmatrix}1\\ \overline{\sigma}\\ \sigma\end{pmatrix}\\
&=\frac{1}{|\Omega_h|}(\alpha+\beta(\sigma+\overline{\sigma}))\\&=\bK\cdot A^{(1)}_{0,0}\bK+\bK\cdot A^{(1)}_{0,-1}R\bK(\sigma+\overline{\sigma}).
\end{split}
\end{equation}
For $\sigma=\tau$ or $\bar\tau$, $E_\sigma^{(1,0)}\equiv \left\langle \Phi^{(0)}_\sigma, \mathcal{L}^{(1)}\Phi^{(0)}_\sigma \right\rangle = \bK^T\ A^{(1)}_{0,0}\ \bK-\bK^T\  A^{(1)}_{0,-1}R\ \bK.$
\\ For $\sigma=1$,
$E_1^{(1,0)}\equiv\left\langle \Phi^{(0)}_1, \mathcal{L}^{(1)}\Phi^{(0)}_1 \right\rangle = \bK^T\ A_{0,0}^{(1)}\ \bK\ +\
 2\bK^T\ A_{0,-1}^{(1)}\ R\ \bK$.

\end{proof}

Note that if $\bK^T\ A_{0,-1}^{(1)}\ R\ \bK\ne0$, then
for all small non-zero $\eps$, in a neighborhood of $E^{(0)}=|\bK|^2$ are two distinct $L^2_{\bK}-$ eigenvalues separated by $\mathcal{O}(\eps)$: a simple $L^2_{\bK,1}-$ eigenvalue, $\tilde{E}^\eps_1=E_1^\eps$ and doubly-degenerate
 $L^2_{\bK,\tau}\oplus L^2_{\bK,\bar\tau}-$ eigenvalue,
  $E^\eps_D=E^\eps_\tau=E^\eps_{\bar\tau}$.

 If $\eps\bK^T\ A_{0,-1}^{(1)}\ R\ \bK>0$, then $E^\eps_D<\tilde{E}^\eps$; Dirac points occur at the intersection of the first and second dispersion surfaces.
 If $\eps\bK^T\ A_{0,-1}^{(1)}\ R\ \bK<0$, then $E^\eps_D>\tilde{E}^\eps$; Dirac points occur at the intersection of the second and third dispersion surfaces. This verifies the assertion in part 5 of Theorem \ref{eps_0_prop} and the proof of Theorem \ref{eps_0_prop} is now complete.

\subsection{Dirac points for arbitrary contrast honeycomb structures}\label{generic_eps}

Theorem \ref{low-dp} studies Dirac point for low contrast honeycomb structures by studying
 $\mathcal{L}^{(\eps)}=-\nabla\cdot\left(\ I\ +\ \eps\ A^{(1)}(\bx)\ \right)\nabla$ for $\eps\ne0$ and small.
In this section we discuss an extension of these results on Dirac points to arbitrary contrast structures.  \medskip

\nit {\bf Assumption:} For all $\eps>0$, $A^{(\eps)}(\bx)\equiv I\ +\ \eps\ A^{(1)}(\bx) $ defines a honeycomb structured medium and therefore  $\mathcal{L}^{(\eps)}\equiv -\nabla\cdot\ A^{(\eps)}(\bx)\ \nabla$ is  self-adjoint, positive definite and uniformly elliptic on $\R^2$.
\medskip

We state the result and then briefly discuss the strategy of proof, implemented in full detail
in the context of Schr\"odinger operators by C. L. Fefferman and one of the authors; see \cites{FW:12} and
Appendix D of \cites{FLW-MAMS:17}.  We do not present the detailed implementation of this strategy in this work.
\medskip

\nit {\bf Claim:}\ Let $\eps^0$ be as in Theorem \ref{low-dp}. There exists a discrete set
$\tilde{\mathcal{C}}\subset \R\setminus(0,\eps^0)$, such that if $\eps\notin\tilde{\mathcal{C}}$, then the conditions of
 Proposition \ref{prop_conical} hold for some quasi-momentum energy pair $(\bK_\star,E_D)$, where $\bK_\star$ is any vertex of the Brillouin zone. It follows that $(\bK_\star,E_D)$ is a Dirac point.  \medskip
 \medskip

\nit {\bf N.B.}\ For small $0<|\eps|<\eps^0$, Theorem
 \ref{low-dp} ensures that these Dirac points occur at the intersections between the first and second, or second and third dispersion surfaces.
  For general $\eps\notin\tilde{\mathcal{C}}$, possibly large, we make no assertions on which dispersion surfaces intersect at Dirac points.
\medskip

The strategy is based on an analytical characterization of the $L^2_{\bK_\star,\sigma}-$ eigenvalues of $\mathcal{L}^{(\eps)}$. By assumptions on $A^{(\eps)}$ (strong ellipticity and symmetries), $T(\eps) \equiv (\mathcal{L}^{(\eps)})^{-1}$
is defined as a bounded operator from $L^2_{\bK_\star,\sigma}$ to $H^2_{\bK_\star,\sigma}$. Hence, the $L^2_{\bK_\star,\sigma}-$ eigenvalue problem \eqref{eps_evp} may formulated equivalently as:
\begin{equation}
(I-E\ T(\eps))\Phi=0, \quad \Phi\in L^2_{\bK_\star, \sigma}.
\label{LipSch}
\end{equation}
We require a global analytical criteria for a complex number, $E$, to be an eigenvalue \eqref{LipSch}.
Although  $T(\eps)$ is compact, it is not trace class, and therefore its determinant is not defined. However, $T(\eps)$ is a Hilbert-Schmidt operator and we proceed by working with its modified determinant $\det_2(I-ET(\eps))$ \cites{Newton:72,Simon:05,GLZ:08}.

\begin{theorem}\label{arbitrary-dp}
Let $\sigma$ take on the values $1, \;\tau$ or $\bar{\tau}$.
\begin{enumerate}
\item $\eps \mapsto T(\eps)$ is an analytic mapping from $\mathbb C$ to the space of Hilbert-Schmidt operators on $L^2_{\bK,\sigma}$.
\item For $T(\eps)$, considered as a mapping on $L^2_{\bK_\star,\sigma}$, define
\begin{equation}
\mathcal E_\sigma(E,\eps)={\det}_2(I-ET(\eps)).
\label{det2}\end{equation}
The mapping $(E,\eps)\mapsto \mathcal E_\sigma(E,\eps)$, is analytic.
\item For $\eps$ real, $E$ is an $L^2_{\bK_\star,\sigma}-$ eigenvalue of geometric multiplicity $m$ if and only if $E$ is a root of $\mathcal E_\sigma(E,\eps)=0$ of multiplicity $m$.
\end{enumerate}
\end{theorem}

The complex function theory strategy of \cites{FW:12} (see also Appendix D of \cites{FLW-MAMS:17}) can be used to establish that for all $\eps$ which fall outside of a discrete subset of $\tilde{\mathcal{C}}\subset\R\setminus(0,\eps^0)$ that there exists $E_D^\eps\in\R$ such that
(a)\ $E_D^\eps\in\R$ is a simple zero of  $\mathcal{E}_\tau(E_D^\eps,\eps)$ and $\mathcal{E}_{\overline\tau}(E_D^\eps,\eps)$,\ (b)\ $\mathcal{E}_1(E_D^\eps,\eps)\ne0$, and (c) $\upsilon_{_F}(\eps)\ne0$. Here, $\mathcal{E}_\sigma(E,\eps)$,
 for $\sigma=1,\tau,\overline\tau$, is defined in \eqref{det2}. Therefore, by Proposition \ref{prop_conical}, for all such $\eps$ there exist of Dirac points at the vertices of the Brillouin zone, $\mathcal{B}$.

\section{Dirac points under perturbation - instability and persistence}\label{dirac_persistence}

In Section \ref{dirac-pts} we studied the existence of Dirac points of honeycomb structures $\LA=-\nabla\cdot A\nabla$ in the setting where  $\LA$ commutes with $\pc$ and with $\mathcal{R}$. Theorem \ref{prop_conical} gives criteria for the exists of Dirac points at vertices of the Brillouin zone $\B_h$. Theorem \ref{low-dp} in Section \ref{sec:dp-low}, proved using  Theorem \ref{prop_conical}, studies of Dirac points for low-contrast media. In Section \ref{generic_eps} discuss Dirac points for generic honeycomb structures, without any assumptions on contrast, appealing to the continuation argument developed in \cites{FW:12} for the case of Schr\"odinger operators with honeycomb lattice potential, $-\Delta+V(\bx)$, where $V(\bx)$ is $\Lambda_h-$ periodic, and $\pc-$ and $\mathcal{R}-$ invariant.

 In this section we discuss how the locally conical structure of dispersion surfaces near Dirac points deform when $\pc-$ invariance in broken. For results on the case of Schr\"odinger operators with a honeycomb potentials, see \cite{FW:12,berkolaiko-comech:15}.

 Starting with $\LA$, we introduce perturbed operators of the form
\begin{equation}
\label{L_delta}
\mathcal{L}^{(\delta)}=-\nabla\cdot\left(A + \delta B \right)\nabla=-\nabla\cdot A\nabla-\delta \nabla\cdot B\nabla \equiv \mathcal{L}^A + \delta \mathcal{L}^B.
\end{equation}
Here,  $B=B(\bx)$ is a smooth $2\times2$ Hermitian matrix function, which is $\Lambda_h-$periodic.  The parameter, $\delta$, is the strength of the perturbation, and is taken to be a real number and sufficiently small.

The perturbed operator, $\mathcal{L}^{(\delta)}$, breaks $\pc-$ invariance if $[\pc,\mathcal{L}^{(\delta)}]\ne0$, or equivalently if
$[\pc,\LB]\ne0$. We shall, in particular, assume
 that $\LB$ is $\pc-$ anti-symmetric:
 \begin{equation}
 \pc\LB\ =\ -\LB\pc. \label{pc-anti}
 \end{equation}

In Appendix \ref{maxwell},
 we show that this class of $\pc-$ invariance breaking perturbation encompasses examples of magneto-optic and bi-anisotropic media introduced discussed in the  introduction.

The perturbed eigenvalue problem in $L^2(\R^2/\Lambda_h)$ is given by
\begin{equation}\label{FB-delta}
\mathcal{L}^{(\delta)}(\bk)\phi^{\delta}(\bx;\bk)= E^{\delta}(\bk)\phi^{\delta}(\bx;\bk), \quad \phi^{\delta}(\bx+\bv;\bk)=\phi^{\delta}(\bx;\bk)\ \  \forall\ \bv\in\Lambda_h,
\end{equation}
where
\begin{equation}
\label{L_delta_k}
\mathcal{L}^{(\delta)}(\bk)= \mathcal{L}^A(\bk) + \delta \mathcal{L}^B(\bk) \equiv -(\nabla+i\bk)\cdot A(\nabla+i\bk)-\delta(\nabla+i\bk)\cdot B(\nabla+i\bk).
\end{equation}
To study the deformation of dispersion surfaces near Dirac points under perturbation we study the perturbed Floquet-Bloch eigenvalue problem \eqref{FB-delta} for $\delta$ small and  $\bkappa\equiv \bk-\bK_\star$ small,
 where $\bK_\star$ is a vertex of $\B_h$.

By Theorem \ref{low-dp}, the operator $\LA(\bK_\star)$ has a doubly degenerate $L^2(\R^2/\Lambda_h)$- eigenvalue, which we denote $E_{D}$ (independent of the particular vertex $\bK_\star$) with corresponding two dimensional eigenspace: $\textrm{span}
\{\phi^{\bK_\star}_1(\bx),\phi^{\bK_\star}_2(\bx)\}$.
Expansion of $\mathcal{L}^{(\delta)}(\bk)$ about $\bk=\bK_\star$ gives:
\begin{align*}
\mathcal{L}^{(\delta)}(\bK_\star+\bkappa)=\mathcal{L}^A(\bK_\star)+\delta \mathcal{L}^B(\bK_\star)+\bkappa\cdot \mathscr{A}(\bK_\star)+\mathscr{R}_2(\bkappa,\delta).
\end{align*}
Here $\mathscr{A}(\bK_\star)$ is defined in \eqref{Ak_def} and $\mathscr{R}_2(\bkappa,\delta)$ includes all terms of order $\mathcal{O}(|\bkappa|^2+\delta|\bkappa|)$:
\begin{equation*}
 \mathscr{R}_2(\bkappa,\delta) \equiv \bkappa^T A\bkappa-\delta\left(i\bkappa\cdot B(\nabla+i\bK_\star)+i(\nabla+i\bK_\star)\cdot(B\bkappa)-\bkappa^TB\bkappa\right).
\end{equation*}
Let
\begin{equation*}
\phi^{\delta}(\bx;\bk)=\phi_{\bK_\star}^{(0)}(\bx)+\phi_{\bK_\star}^{(1)}(\bx;\bkappa),  \quad  E^{\delta}(\bk)= E_D+ E_{\bK_\star}^{(1)}(\bkappa),
\end{equation*}
where $\phi_{\bK_\star}^{(0)}\in \textrm{span}
\{\phi^{\bK_\star}_1,\phi^{\bK_\star}_2\}$
 and $\left\langle\phi^{\bK_\star}_j,\phi_{\bK_\star}^{(1)}(\cdot;\bkappa)\right\rangle=0,\ j=1,2$.
We follow a Lyapunov-Schmidt reduction strategy, analogous to that used in the proof of Theorem \ref{prop_conical}. Since $\phi_{\bK_\star}^{(0)}(\bx)$ is in the nullspace of $\LA(\bK_\star)-E_D$, we have
\begin{equation}
\phi_{\bK_\star}^{(0)}(\bx)=\alpha_1^{\bK_\star}\phi_1^{\bK_\star}(\bx)+\alpha_2^{\bK_\star}\phi_2^{\bK_\star},
\end{equation}
where $\alpha_1^{\bK_\star},\ \alpha_2^{\bK_\star}$ are complex constants to be determined. Recall also that
 $\phi_j^{\bK_\star}(\bx)=e^{-i\bK_\star\cdot\bx}\Phi_j^{\bK_\star}(\bx)$, where $\Phi_1^{\bK_\star}\in L^2_{\bK_\star,\tau}$ and $\Phi_2^{\bK_\star}\in L^2_{\bK_\star,\bar\tau}$; see Theorem \ref{prop_conical}.

  Calculations which are analogous to those  in the proof of Theorem \ref{prop_conical}, lead to a system of homogeneous linear equations for $\alpha_1^{\bK_\star}$ and $\alpha_2^{\bK_\star}$:

\begin{equation}
\label{PT_K_eqn}
\left( E_{\bK_\star}^{(1)}I-\mathcal M^{\bK_\star}_{\mathscr{A}}(\bkappa)-\delta\mathcal M^{\bK_\star}_{\LB}-\mathcal M^{\bK_\star}_{\mathscr{R}_2}( E^{(1)},\bkappa,\delta) \right)\begin{pmatrix}\alpha_1^{\bK_\star}\\ \alpha_2^{\bK_\star}\end{pmatrix}=0,
\end{equation}
where, by Proposition \ref{prop_g} (all inner products over $L^2_{\bK_\star}$),
\begin{equation}
\label{J_defn}
\mathcal M^{\bK_\star}_{\mathscr{A}}(\bkappa) =
\upsilon_{_F}\ \begin{pmatrix}0 & \kappa^{(1)}+i\kappa^{(2)}\\
\kappa^{(1)}-i\kappa^{(2)}&0\end{pmatrix},
\end{equation}
\begin{equation}
\label{A_defn}
\mathcal M^{\bK_\star}_{\mathcal{L}^B} =
\begin{pmatrix}
\inner{\Phi^{\bK_\star}_1, \mathcal{L}^B \Phi^{\bK_\star}_1} & \inner{\Phi^{\bK_\star}_1, \mathcal{L}^B \Phi^{\bK_\star}_2}\\
\inner{\Phi^{\bK_\star}_2, \mathcal{L}^B \Phi^{\bK_\star}_1}&\inner{\Phi^{\bK_\star}_2, \mathcal{L}^B \Phi^{\bK_\star}_2}
\end{pmatrix} ;
\end{equation}
and $|\mathcal{M}^{\bK_\star}_{\mathcal{R}}( E^{(1)},\bkappa,\delta)|
\lesssim|\bkappa|^2+|\delta|\ |\bkappa|$.

Up to this point we have not used particular properties of the perturbation, $\LB$. In the following subsections, we discuss conditions on $B(\bx)$ under which Dirac points are unstable (conical behavior perturbs to locally smooth and ``gapped'' dispersion surfaces) and conditions under which the local conical structure of  Dirac points persists. The key, as we  see below, is to determine the character of the perturbation matrix
 $\mathcal M^{\bK_\star}_{\mathcal{L}^B}$. We note that $\mathcal M^{\bK_\star}_{\mathcal{L}^B}$ is Hermitian since $\LB$ is self-adjoint.

\subsection{Instability of Dirac points for a class of $\pc-$ breaking perturbations}\label{unstable-dirac}

\begin{proposition}\label{MLB}
  If $\mathcal{L}^B$ is $\pc-$ anti-symmetric, {\it i.e.} $\mathcal{PC} \mathcal{L}^B= -\mathcal{L}^B \mathcal{PC}$, then
\begin{equation}
\mathcal M^{\bK_\star}_{\mathcal{L}^B}=\begin{pmatrix}\thetasharp^{\bK_\star} &0\\0&-\thetasharp^{\bK_\star} \end{pmatrix},
\label{A_defn1}\end{equation}
where $\thetasharp^{\bK_\star}$ is given by
\begin{equation}
\thetasharp^{\bK_\star}\equiv \inner{\Phi^{\bK_\star}_1,\mathcal{L}^B \Phi^{\bK_\star}_1},
\label{thta_shp}
\end{equation}
and is real, by self-adjointness of $\LB$.
\end{proposition}

{\bf Assumption:} Throughout we shall assume that
\begin{equation}
\thetasharp^{\bK_\star}\ne0.
\label{thta_ne0}
\end{equation}
\bigskip

\begin{proof}[Proof of Proposition \ref{MLB}]
Suppose $\mathcal{PC}\mathcal{L}^B=-\mathcal{L}^B\mathcal{PC}$. In the calculations below, we frequently use the following equality: for any $f,g\in L^2_{\bK_\star}$,
\begin{equation*}
\inner{\mathcal{PC}f,\mathcal{PC}g}=\inner{\mathcal{C}f,\mathcal{C}g}=\overline{\inner{f,g}}=\inner{g,f}.
\end{equation*}

Consider first  the off-diagonal elements of
 $\mathcal M^{\bK_\star}_{\mathcal{L}^B}$, we have
 \begin{align*}
\inner{\Phi^{\bK_\star}_1,\mathcal{L}^B \Phi^{\bK_\star}_2}&=\inner{\mathcal{PC}\Phi^{\bK_\star}_2, \mathcal{L}^B \mathcal{PC}\Phi^{\bK_\star}_1} =-\inner{\mathcal{PC}\Phi^{\bK_\star}_2,\mathcal{PC} \mathcal{L}^B\Phi^{\bK_\star}_1}\nn\\
&=
-\inner{\mathcal{L}^B\Phi^{\bK_\star}_1,\Phi^{\bK_\star}_2}=-\inner{\Phi^{\bK_\star}_1, \mathcal{L}^B\Phi^{\bK_\star}_2},
\end{align*}
where we have used the self-adjointness of $\mathcal{L}^B$ at the last equality. Then it follows that $\inner{\Phi^{\bK_\star}_1, \mathcal{L}^B \Phi^{\bK_\star}_2}=0$.

We now turn to  the diagonal entries of $\mathcal M^{\bK_\star}_{\mathcal{L}^B}$.  We have
\begin{align}
\label{theta_defn1}
\inner{\Phi^{\bK_\star}_1, \mathcal{L}^B\Phi^{\bK_\star}_1}&=\inner{\mathcal{PC} \Phi^{\bK_\star}_2,\mathcal{L}^B \mathcal{PC} \Phi^{\bK_\star}_2}\nn\\
&=-\inner{\mathcal{PC} \Phi^{\bK_\star}_2,  \mathcal{PC} \mathcal{L}^B \Phi^{\bK_\star}_2}=-\inner{\mathcal{L}^B\Phi^{\bK_\star}_2, \Phi^{\bK_\star}_2}=\thetasharp^{\bK_\star}.
\end{align}
This completes the proof of Proposition \ref{MLB}.
\end{proof}

It follows from Proposition \ref{MLB} that if $\pc\LB=-\LB \pc$, then

\begin{equation}
\label{local_K_matrix}
\mathcal M^{\bK_\star}_{\mathscr{A}}(\bkappa)+\delta\mathcal M^{\bK_\star}_{\mathcal{L}^B}=
\begin{pmatrix}\delta\thetasharp^{\bK_\star} &\upsilon_{_F}(\kappa^{(1)}+i\kappa^{(2)})\\
\upsilon_F(\kappa^{(1)}-i\kappa^{(2)})&-\delta\thetasharp^{\bK_\star} \end{pmatrix}.
\end{equation}
 By \eqref{PT_K_eqn} and \eqref{local_K_matrix}, the energy $E=E_D+E_{\bK_\star}^{(1)}$ is a $L^2_{\bK_\star+\bkappa}-$ eigenvalue of the perturbed Floquet-Bloch eigenvalue problem for $\mathcal{L}^{(\delta)}$ if and only if
 \begin{equation}
 \label{det_rel}
 \det ( E_{\bK_\star}^{(1)}I-\mathcal M^{\bK_\star}_{\mathscr{A}}(\bkappa)-\delta\mathcal M^{\bK_\star}_{\LB}-\mathcal M^{\bK_\star}_{\mathscr{R}_2}( E^{(1)},\bkappa,\delta))  = 0.
 \end{equation}
 The eigenvalue condition \eqref{det_rel} is of the form
\begin{equation}
\label{local_K_energy}
 ( E_{\bK_\star}^{(1)})^2-\delta^2(\thetasharp^{\bK_\star})^2 -\upsilon_{_F}^2|\bkappa|^2
 + g_{120} + g_{111} + g_{012} + g_{021} + g_{003} = 0,
\end{equation}
where $g_{rsl}=g_{rsl}( E_{\bK_\star}^{(1)},\delta,\bkappa )$ are smooth and satisfy the bound
$g_{rsl}( E_{\bK_\star}^{(1)},\delta,\bkappa ) \leq C |E_{\bK_\star}^{(1)}|^r|\delta|^s|\bkappa|^l$.
After some manipulations, we may apply the implicit function theorem to obtain from \eqref{local_K_energy} that
\begin{equation}
\label{PT_antisym_local_energy}
 E_{\bK_\star}^{(1)}=\pm\sqrt{\upsilon_{_F}^2|\bkappa|^2+\delta^2(\thetasharp^{\bK_\star})^2}(1+e^{\bK_\star}_\pm(\bkappa,\delta)),
\end{equation}
where
$e^{\bK_\star}_\pm(\bkappa,\delta)=\mathcal{O}(\delta+\bkappa)$.

Therefore, provided $\thetasharp^{\bK_\star}\neq0$, the Dirac point $(\bK_\star, E_D)$ does not persist in the presence of $\mathcal{PC}$- anti-symmetric perturbations. That is, if $\pc\LB=-\LB\pc$ and $\thetasharp^{\bK_\star}\neq0$, then for all $\delta$ sufficiently small, the dispersion surfaces are locally smooth and a \underline{local} spectral gap up for quasi-momenta near $\bK_\star$.  We summarize the above discussion in the following

\begin{theorem}\label{PT_thm}
Consider the operator
$\mathcal{L}^{(\delta)}=\LA+\delta\LB=-\nabla\cdot(A+\delta B)\nabla$ defined in \eqref{L_delta}. For the unperturbed operator we assume that $A$ is smooth, Hermitian, $\Lambda_h-$ periodic, and that $[\pc,\LA]=0$ and $[ \mathcal{R},\LA]=0$.
For the perturbed operator we assume that
$B(\bx)$ is smooth, Hermitian, $\Lambda_h-$ periodic, and that
$\pc\LB=-\LB\pc$ ($\pc-$ anti-symmetry).

 Let $\bK_\star$ be a point of $\bK$ or $\bK'$ type, and let $(\bK_\star,E_D)$ denote a Dirac point of $\LA$ (Theorem \ref{low-dp}).
Assume $\thetasharp^{\bK_\star} \equiv \inner{\Phi^{\bK_\star}_1,\LB \Phi^{\bK_\star}_1}\neq0$. Then,
there exists a $\delta_0>0$ such that for $0<\delta<\delta_0$,
 the Dirac point $(\bK_\star,E_D)$ does not persist, the perturbed dispersion surfaces are locally smooth and a local spectral gap opens up in a neighborhood of $(\bk,E)=(\bK_\star,E_D)$.
\end{theorem}

\begin{example}\label{ex-pcanti} Here we give two typical examples of a family of matrices, $B(\bx)$, satisfying the hypotheses of Theorem \ref{PT_thm}. Let $\mu_1(\bx)$ be real-valued, $\mu_1(\bx)=\mu_1(-\bx)$, and $\mu_1(R^*\bx)=\mu_1(\bx)$. Define
\begin{equation}
B(\bx)\equiv\sigma_2\mu_1(\bx).\ \
\label{B-example}\end{equation}
Then, $[\mathcal{R},\LB]=0$ and $\pc\LB=-\LB\pc$, since
$
 \mathcal{C}[B(\bx)\cdot]=-\sigma_2\mu(\bx)\mathcal{C}\cdot=-B(\bx)\mathcal{C}\cdot$,
 while $
 \mathcal{P}[B(\bx)\cdot]=\sigma_2\mu(-\bx)\mathcal{P}\cdot
 = \sigma_2\mu(\bx)\mathcal{P}\cdot=B(\bx)\mathcal{P}\cdot$.

 The other example is
 \begin{equation*}
B(\bx)\equiv\mu_2(\bx) I_{2\times 2}.\ \
\label{B-example2}\end{equation*}
 where $\mu_2(\bx)$ is a real-valued function satisfying $\mu_2(-\bx)=-\mu_2(\bx)$, and $\mu_2(R^*\bx)=\mu_2(\bx)$.
Then, in this case $[\mathcal{R},\LB]=0$ and $\pc\LB=-\LB\pc$, since
$
 \mathcal{C}[B(\bx)\cdot]=B(\bx)\mathcal{C}\cdot$,
 while $
 \mathcal{P}[B(\bx)\cdot]=\mu_2(-\bx)\mathcal{P}\cdot
 = -B(\bx)\mathcal{P}\cdot$.

\end{example}

\begin{remark}\label{thetasharp_bK_bKp}
In our discussion of edge states in Section \ref{edge_states}, we shall consider two important sub-cases, where $\mathcal{L}^B$ is $\mathcal{PC}$ anti-symmetric. Case(a): $[\mathcal{P},\LB]=0$ and $\mathcal{C}\LB=-\LB\mathcal{C}$, as in Example \ref{ex-pcanti}, and Case (b) $\mathcal{P}\LB=-\LB\mathcal{P}$, but $[\mathcal{C},\LB]=0$. We shall see that
for case (a): $\thetasharp^\bKp=+\thetasharp^\bK$ and for case (b): $\thetasharp^\bKp=-\thetasharp^\bK$.  Implications  for the directionality of edge states in each of cases are explored in Section \ref{edge_states}.
\end{remark}

\subsection{Remarks on the persistence of Dirac points under $\mathcal{P}\circ\mathcal{C}$ symmetry preserving perturbations}
\label{persistence_dirac}

\begin{proposition}\label{persist-prop}
 Assume $B$ is $\Lambda_h-$ periodic
 and $\mathcal{L}^B$ is $\pc-$-symmetric, {\it i.e.} $[\mathcal{PC}, \mathcal{L}^B]=0$.
  Then,
  \begin{equation}
\mathcal M^{\bK_\star}_{\mathcal{L}^B}=
\begin{pmatrix}\thetasharp^{\bK_\star} & \varrho_\sharp^{\bK_\star}\\ \overline{\varrho_\sharp^{\bK_\star}} &\thetasharp^{\bK_\star} \end{pmatrix},
\end{equation}
 where $\thetasharp^{\bK_\star}\in\R$ is given in \eqref{thta_shp} and $\varrho_\sharp^{\bK_\star}\equiv \inner{\Phi^{\bK_\star}_1, \LB\Phi^{\bK_\star}_2}$.
\end{proposition}

\begin{proof}[Proof of of Proposition \ref{persist-prop}]
If $\mathcal{PC}\mathcal{L}^B=\mathcal{L}^B\mathcal{PC}$, then
\begin{align*}
\thetasharp^{\bK_\star} &\equiv \inner{\Phi^{\bK_\star}_1,\mathcal{L}^B \Phi^{\bK_\star}_1}=\inner{\mathcal{PC} \Phi^{\bK_\star}_2, \mathcal{L}^B \mathcal{PC} \Phi^{\bK_\star}_2} \\
&=\inner{\mathcal{PC} \Phi^{\bK_\star}_2,\mathcal{PC} \mathcal{L}^B \Phi^{\bK_\star}_2} = {\inner{\Phi^{\bK_\star}_2, \mathcal{L}^B \Phi^{\bK_\star}_2}},
\end{align*}
and
\begin{equation}
\varrho_\sharp^{\bK_\star} \equiv \inner{\Phi^{\bK_\star}_1, \LB\Phi^{\bK_\star}_2}=\inner{\LB\Phi^{\bK_\star}_1, \Phi^{\bK_\star}_2}=\overline{\inner{\Phi^{\bK_\star}_2, \LB\Phi^{\bK_\star}_1}}.
\end{equation}
This completes the proof of Proposition \ref{persist-prop}.
\end{proof}

Therefore, if $\pc\LB=\LB \pc$, we have
\begin{equation}
\label{pc_M_matrices}
\mathcal M^{\bK_\star}_{\mathscr{A}}(\bkappa)+\delta\mathcal M^{\bK_\star}_{\mathcal{L}^B}=
\begin{pmatrix}\delta\thetasharp^{\bK_\star} &\upsilon_F(\kappa^{(1)}+i\kappa^{(2)}) + \delta\varrho_\sharp^{\bK_\star}\\
\upsilon_F(\kappa^{(1)}-i\kappa^{(2)})+ \delta\overline{\varrho_\sharp^{\bK_\star}}&\delta\thetasharp^{\bK_\star} \end{pmatrix}.
\end{equation}
By \eqref{PT_K_eqn} and \eqref{pc_M_matrices}, the energy $E=E_D+E_{\bK_\star}^{(1)}$ is a $L^2_{\bK_\star+\bkappa}-$ eigenvalue of the perturbed Floquet-Bloch eigenvalue problem for $\mathcal{L}^{(\delta)}$ if and only if
 \begin{equation*}
 \det ( E_{\bK_\star}^{(1)}I-\mathcal M^{\bK_\star}_{\mathscr{A}}(\bkappa)-\delta\mathcal M^{\bK_\star}_{\LB}-\mathcal M^{\bK_\star}_{\mathscr{R}_2}( E^{(1)},\bkappa,\delta))  = 0.
 \end{equation*}
 Following the same procedure as in Section \ref{unstable-dirac} but for \eqref{pc_M_matrices}, we find that
 that $E_{\bK_\star}^{(1)}$ is given by
\begin{equation}
\label{PT_sym_local_energy}
 E_{\bK_\star}^{(1)}= \left( \delta\thetasharp^{\bK_\star} \pm \abs{\upsilon_F(\kappa^{(1)}+i\kappa^{(2)})+\delta\varrho_\sharp^{\bK_\star}} \right) ( 1 +e^{\bK_\star}_\pm(\bkappa,\delta) ) ,
\end{equation}
where $e^{\bK_\star}_\pm(\bkappa,\delta)=\mathcal{O}(\delta+|\bkappa|)$.

The local energy expansion \eqref{PT_sym_local_energy} demonstrates, to finite order in $\delta$, that the Dirac point $(\bK_\star,E_D)$ is protected by $\pc-$ symmetry. That is, if $[\pc,\LA]=0$, the Dirac point degeneracy shifts to $(\bK_\star+\delta\bb b,E_D+\delta\thetasharp^{\bK_\star})$, where $\bb b=(b^{(1)},b^{(2)})$ satisfies $\upsilon_F(b^{(1)}+ib^{(2)})+\varrho_\sharp^{\bK_\star}=0$
A rigorous proof can be implemented along the lines of Theorem 9.1  in \cites{FW:12}.

\section{Edge States}\label{edge_states}

Edge states are time-harmonic solutions of the wave equation (obtained via reduction of the 2D Maxwell Equations -- see Remark \ref{wavepkts} and Appendix \ref{maxwell}), which are bounded and oscillatory in the direction parallel to an extended line-defect, and spatially localized transverse to it.

In Sections \ref{rational_edges} we review the mathematical framework of \cites{FLW-AnnalsPDE:16,FLW-2DM:15} for describing ``rational'' edges in a honeycomb structure. We then introduce our model of a honeycomb structure perturbed by an edge in Section \ref{honeycomb_model}. We formally derive edge states in Section \ref{multiple_scales}, and sketch the rigorous justifications, along the lines of \cites{FLW-AnnalsPDE:16}, of our formal results in Section \ref{rigor}. Sections \ref{edge_state_kpar} and \ref{direction} investigate the direction of wavepacket propagation along the edge. {\it For simplicity, we restrict our attention to the case of real-valued  honeycomb structured media, $A(\bx)$; see Remark \ref{dp-nonreal}.}

\subsection{Rational edges}\label{rational_edges}

Recall, from Section \ref{triangle}, that the lattice and dual lattice are given by $\Lambda_h=\Z\bv_1\oplus\Z\bv_2$ and $\Lambda_h^*=\Z\bk_1\oplus\Z\bk_2$, respectively. An {\it edge} is a line  of the form $\R (a_1\bv_1+a_2\bv_2)$, where $(a_1,b_1)=1$, {\it i.e.} $a_1$ and $b_1$ are relatively prime integers. In Section \ref{honeycomb_model}, we shall introduce a family of operators, which interpolate between two distinct $\Lambda_h-$ periodic structures at ``$+\infty$'' and ``$-\infty$''.

We fix an edge by choosing a vector $\vtilde_1=a_1\bv_1+b_1\bv_2$, where  $(a_1, b_1)=1$. Since  $a_1, b_1$ are relatively prime,  there exists a relatively prime pair of integers: $a_2,b_2$  such that $a_1b_2-a_2b_1=1$.  Set $\vtilde_2 = a_2 \bv_1 + b_2 \bv_2$.
It follows that $\Z\vtilde_1\oplus\Z\vtilde_2=\Z\bv_1\oplus\Z\bv_2=\Lambda_h$.
Since $a_1b_2-a_2b_1=1$, we have dual lattice vectors $\ktilde_1, \ktilde_2\in\Lambda_h^*$, given by
\begin{equation}
\label{ktilde_1_2}
\ktilde_1=b_2\bk_1-a_2\bk_2,\ \  \ktilde_2=-b_1\bk_1+a_1\bk_2,
\end{equation}
  which satisfy
\begin{equation*} \ktilde_\ell \cdot \vtilde_{\ell'} = 2\pi \delta_{\ell, \ell'},\ \ 1\leq \ell, \ell' \leq 2.\label{ktilde-orthog}
\end{equation*}
Note that $\Z\ktilde_1\oplus\Z\ktilde_2=\Z\bk_1\oplus\Z\bk_2=\Lambda^*_h$.

Two typical examples of the edges are:
\begin{enumerate}
\item The zigzag edge: $\vtilde_1=\bv_1,~\vtilde_2=\bv_2$ and $\ktilde_1=\bk_1,~\ktilde_2=\bk_2$.
\item The armchair edge: $\vtilde_1=\bv_1+\bv_2,~\vtilde_2=\bv_2$ and $\ktilde_1=\bk_1,~\ktilde_2= \bk_2-\bk_1$.
\end{enumerate}

\subsection{Model of a honeycomb structure with a rational edge}\label{honeycomb_model}

The edge through a honeycomb material is modeled through a perturbation of the bulk operator $\LA$.
We have the following setup.
\begin{enumerate}
\item \emph{Unperturbed operator, $\LA$}: Let $A(\bx)$ be smooth, $\Lambda_h-$ periodic and such that $[\pc,\LA]=0$ and $[\mathcal{R},\LA]=0$, {\it i.e.}, a honeycomb structured medium, as described in Section \ref{honey-media}.
\item \emph{Dirac points}: Let $\bK_\star$ denote a quasi-momentum at the vertex of the Brillouin zone, $\B_h$ (a point of $\bK$ or $\bK'$ type), and assume that $(\bK_\star,E_D)$ is a Dirac point of the operator $\mathcal{L}^A=-\nabla\cdot A\nabla$; see Definition \ref{dirac_pt_defn}, {\it e.g.} the setting assumed in the hypotheses of Theorem \ref{prop_conical}.
\item \emph{Perturbed operator asymptotics far from the edge}: Let $B(\bx)$ be a $\Lambda_h-$periodic, $2\times 2$ Hermitian matrix such that $\mathcal{L}^B \equiv -\nabla\cdot B\nabla$ is anti-$\mathcal{PC}-$ symmetric: $\mathcal{PC}\LB=-\LB\mathcal{PC}$; {\it i.e.}, part (1) of Theorem \ref{PT_thm} applies to the operator $-\nabla\cdot(A+\delta B)\nabla$. Furthermore, assume that
\begin{equation}
\thetasharp^{\bK_\star} \equiv \inner{\Phi^{\bK_\star}_1,\LB \Phi^{\bK_\star}_1}\neq0\ .
\label{thta_shp1}
\end{equation}
\end{enumerate}

Suppose we fix an edge, given by the line $\R\vtilde_1$ or, alternatively, $\bx\in\R^2$ such that $\ktilde_2\cdot\bx=0$.
Our operator is a smooth and slow interpolation, transverse to the edge,  between the operators
\begin{equation}
\label{asymptotic_ops}
\mathcal{L}_{\pm}^{(\delta)} \equiv -\nabla \cdot \left[ A(\bx) \pm \delta \eta_\infty B(\bx)\right] \nabla,
\end{equation}
associated with periodic structures $A(\bx)-\delta \eta_\infty B(\bx)$ for $\ktilde_2\cdot\bx\to-\infty$ and $A(\bx)+\delta \eta_\infty B(\bx)$
 for $\ktilde_2\cdot\bx\to+\infty$. Here, $\eta_\infty$ is a positive constant.
This interpolation is  effected by a \emph{domain wall function}:

\begin{definition} \label{domain_wall_defn}
We call $\eta(\zeta)\in C^{\infty}(\R)$ a  domain wall function if $\eta(\zeta)$ tends to $\pm\eta_\infty$ as $\zeta\to\pm\infty$. We take $\eta(0)=0$, and
without loss of generality, we assume $\eta_\infty>0$.
\end{definition}

\nit Our model of a honeycomb structure with an edge is the domain-wall modulated operator:
\begin{equation}
 \label{dw_ham}
 \mathcal{L}_{\rm dw}^{(\delta)} \equiv -\nabla\cdot\left[A(\bx) +\delta\eta( \delta\ktilde_2 \cdot \bx)B(\bx)\right]  \nabla.
\end{equation}
The operator $\mathcal{L}_{\rm dw}^{(\delta)}$ breaks translation invariance with respect to arbitrary elements of the lattice, $\Lambda_h$, but  is invariant with respect to translation by $\vtilde_1$, parallel to the edge (because $\ktilde_2 \cdot\vtilde_1=0$ in \eqref{dw_ham}). Associated with this translation invariance is a parallel quasi-momentum, which we denote by $\kpar$.

Edge states are solutions of the eigenvalue problem
\begin{align}
&\mathcal{L}_{_{\rm dw}}^{(\delta)}\Psi(\bx;\kpar) = E(\kpar)\Psi(\bx;\kpar), \label{dw_evp}\\
&\Psi(\bx+\vtilde_1;\kpar)=e^{i\kpar}\Psi(\bx;\kpar),\qquad \textrm{(propagation parallel to the edge,  $\R\vtilde_1$)},\label{pseudo-per}\\
&\Psi(\bx;\kpar) \to 0\ \ {\rm as}\ \ |\bx\cdot\ktilde_2|\to\infty. \qquad  \textrm{(localization tranverse to the edge,  $\R\vtilde_1$)}\label{localized} .
\end{align}
We refer to a solution pair $(E(\kpar),\Psi(\bx;\kpar))$ of \eqref{dw_evp}-\eqref{localized} as an edge state or edge mode.
We shall construct edge modes for $\kpar$ near $\bK\cdot\vtilde_1$ and near $\bK'\cdot\vtilde_1=-\bK\cdot\vtilde_1$.

\subsection{Multiple scales construction of the edge state}\label{multiple_scales}

Recall the notations
\begin{equation}
 \label{op_notation}
 \mathcal{L}^A \equiv -\nabla_\bx\cdot A(\bx)\nabla_\bx, \qquad \mathcal{L}^B = -\nabla_\bx\cdot B(\bx)\nabla_\bx, \qquad
 \mathscr{A} \equiv \frac{1}{i}A(\bx) \nabla_\bx +\frac{1}{i}\nabla_\bx\cdot A(\bx) ;
\end{equation}
see \eqref{L_def}, \eqref{L_delta} and \eqref{A_def}, respectively.

We first consider the eigenvalue problem \eqref{dw_evp}-\eqref{localized} for the choice of parallel quasi-momentum $\kpar=\bK_\star\cdot\vtilde_1$.
Edge states for $|\kpar-\bK_\star\cdot\vtilde_1|$ small can be constructed  perturbatively; see the discussion in Section \ref{edge_state_kpar}.

For $\delta\ll1$, we formally seek solutions of the eigenvalue problem \eqref{dw_evp}-\eqref{localized}, which depend on  fast ($\bx$) and slow/transverse ($\zeta=\delta\ktilde_2\cdot\bx$) spatial scales:
\begin{align}
E^{\delta}&=E^{(0)}+\delta E^{(1)}+\cdots, \label{formal-E} \\
\Psi^{\delta}&=\Psi^{(0)}(\bx,\zeta)+\delta \Psi^{(1)}(\bx,\zeta)+\cdots,\ \ \zeta=\delta\ktilde_2\cdot\bx . \label{formal-psi}
\end{align}
The pseudo-periodicity condition \eqref{pseudo-per}, with
 $\kpar=\bK_\star\cdot\vtilde_1$, and decaying \eqref{localized} boundary conditions are encoded by requiring, for $j\ge0$:
\begin{align*}
&\Psi^{(j)}(\bx+\vtilde_1,\cdot)=e^{i{\bK_\star}\cdot\vtilde_1}\Psi^{(j)}(\bx,\cdot) \ \ \ \forall\ \vtilde\in\Lambda_h , \quad \text{and} \\
&\zeta\to\Psi^{(j)}(\bx,\zeta)\in L^2(\R_\zeta).
\end{align*}

We substitute the expansions \eqref{formal-E}-\eqref{formal-psi} into \eqref{dw_evp} and equate terms of equal order in $\delta^j$, $j\ge0$.
At order $\delta^0$ we have that $(E^{(0)},\Psi^{(0)})$ satisfies
\begin{equation}\label{Eq_delta0}
\begin{split}
(\mathcal{L}^A-E^{(0)})\Psi^{(0)}&=0,\\
\Psi^{(0)}(\bx+\vtilde_1,\zeta)&=e^{i{\bK_\star}\cdot\vtilde_1}\Psi^{(0)}(\bx,\zeta).
\end{split}
\end{equation}

We are interested in constructing solutions which are spectrally localized near a Dirac point $({\bK_\star},E_D)$. Therefore, we solve \eqref{Eq_delta0} by taking
\begin{equation}
\label{delta_zero}
E^{(0)}= E_D, \quad \Psi^{(0)}=\alpha^{\bK_\star}_1(\zeta)\Phi^{\bK_\star}_1(\bx)+\alpha^{\bK_\star}_2(\zeta)\Phi^{\bK_\star}_2(\bx),
\end{equation}
where the amplitudes, $\alpha^{\bK_\star}_1(\zeta)$ and
$\alpha^{\bK_\star}_2(\zeta)$ are to be determined.

Proceeding to order $\delta^1$, we find that $(E^{(1)},\Psi^{(1)})$ satisfies
\begin{equation}\label{Eq_delta1}
\begin{array}{l}(\mathcal{L}^A-E_D)\Psi^{(1)} =G_1^{(1)}(\bx,\zeta;\Psi^{(0)})+G_2^{(1)}(\bx,\zeta;\Psi^{(0)})+E^{(1)}\Psi^{(0)},\\[1.2 ex]
\Psi^{(1)}(\bx+\vtilde_1,\zeta)=e^{i{\bK_\star}\cdot\vtilde_1}\Psi^{(1)}(\bx,\zeta).
\end{array}
\end{equation}
where
\begin{equation}\label{G1}
G_1^{(1)}(\bx,\zeta;\Psi^{(0)})=i(\partial_\zeta\alpha^{\bK_\star}_1 {\ktilde_2}\cdot \mathscr{A} \Phi^{\bK_\star}_1+\partial_\zeta\alpha^{\bK_\star}_2 \ktilde_2 \cdot \mathscr{A} \Phi^{\bK_\star}_2),
\end{equation}
and
\begin{equation}\label{G2}
G_2^{(1)}(\bx,\zeta;\Psi^{(0)})=-\eta(\zeta)(\alpha^{\bK_\star}_1 \mathcal{L}^B \Phi^{\bK_\star}_1+\alpha^{\bK_\star}_2 \mathcal{L}^B \Phi^{\bK_\star}_2) .
\end{equation}

The pseudo-periodic boundary value problem \eqref{Eq_delta1}  is solvable if and only if its right hand side is $L^2_{\bK_\star}-$ orthogonal to the nullspace of $\LA-E_D$, which is spanned by $\Phi^{\bK_\star}_j,\ j=1,2$. This yields the two solvability  relations:
\begin{equation}
\label{Solvability_d}
-E^{(1)}\alpha^{\bK_\star}_j=\inner{\Phi^{\bK_\star}_j,G_1^{(1)}(\bx,\zeta;\Psi^{(0)})+G_2^{(1)}(\bx,\zeta;\Psi^{(0)})}, \qquad j=1,2.
\end{equation}

Let $\balpha^{\bK_\star}(\zeta)=(\alpha^{\bK_\star}_1(\zeta),\alpha^{\bK_\star}_2(\zeta))^T$.
Substituting \eqref{G1} and \eqref{G2} into \eqref{Solvability_d}, gives
\begin{equation}
\left( \mathcal{D}^{\bK_\star} - E^{(1)} \right) \balpha^{\bK_\star}(\zeta) = 0 , \quad \balpha^{\bK_\star}\in L^2(\R)  , \label{dirac_eqn}
\end{equation}
where $\mathcal{D}^{\bK_\star}$ denotes the 1D Dirac operator:
\begin{align}
 \label{dirac_op}
 \mathcal{D}^{\bK_\star}
 &= -i\mathcal M^{\bK_\star}_{\mathscr{A}}(\ktilde_2) \partial_\zeta + \eta(\zeta)\thetasharp^{\bK_\star} \sigma_3 .
\end{align}
In simplifying \eqref{dirac_op}, we have used the assumption that $\mathcal{L}^B$ is anti-$\mathcal{PC}-$ symmetric and recalled the matrix definitions \eqref{J_defn} and \eqref{A_defn1}:
\begin{align}\label{matr-k2}
\mathcal M^{\bK_\star}_{\mathscr{A}}(\ktilde_2) &=
\upsilon_{_F} \begin{pmatrix}0 & \ktilde_2^{(1)}+i\ktilde_2^{(2)}\\
\ktilde_2^{(1)}-i\ktilde_2^{(2)}&0\end{pmatrix},\quad \text{and} \quad
\mathcal M^{\bK_\star}_{\mathcal{L}^B}
= \thetasharp^{\bK_\star} \sigma_3 .
\end{align}

\begin{proposition}
\label{0E_prop}
Let $\eta(\zeta)$ be a domain wall function (Definition \ref{domain_wall_defn}), $\bK_\star$ a point of $\bK$ or $\bK'$ type, and assume that $\thetasharp^{\bK_\star}\neq0$. Then:
\begin{enumerate}
\item The Dirac operator $\mathcal{D}^{\bK_\star}$ \eqref{dirac_op} possesses a zero-energy eigenvalue, $E^{(1)}=0$,  with exponentially localized eigenfunction given by:
\begin{align}
\label{0mode_theta}
\balpha^{\bK_\star}_\star(\zeta) &=
 \left\{ \begin{array}{ll}
\gamma\ e^{-\frac{\abs{\thetasharp^{\bK_\star}}}{{\upsilon_{_F}} | \ktilde_2|}\ \int_0^\zeta \eta(s) ds}\ \chi_-(\ktilde_2^\perp)
& \quad \text{if} \quad \thetasharp^{\bK_\star}>0; \\
\gamma\ e^{-\frac{\abs{\thetasharp^{\bK_\star}}}{{\upsilon_{_F}} | \ktilde_2|}\ \int_0^\zeta \eta(s) ds}\ \chi_+(\ktilde_2^\perp)
& \quad \text{if} \quad \thetasharp^{\bK_\star}<0, \\
\end{array} \right. \quad \text{where} \\
& \chi_{\pm}(\ktilde_2^\perp)\ =\ \frac{1}{\sqrt2}\ \begin{pmatrix}\widehat{\mathfrak{z}}(\ktilde_2^\perp)\\ \pm1\end{pmatrix}, \nn
\end{align}
and $\widehat{\mathfrak{z}}(\ktilde)= (\smallktilde^{(1)}+i\smallktilde^{(2)})/|\ktilde|$.
The normalization constant,
\begin{equation}\label{gamma-def}
\gamma=\left(
\int_{-\infty}^\infty
\exp\left(
-2\ \frac{\abs{\thetasharp^{\bK_\star}}}{\upsilon_{F} | \ktilde_2|}\ \int_0^\zeta \eta(s) ds\right) d\zeta \right)^{-{\frac{1}{2}}} ,
\end{equation}
 is a real and positive constant and chosen so that $\norm{\balpha^{\bK_\star}_\star}_{L^2(\R)}=1$.
\item The zero-energy mode \eqref{0mode_theta} is ``topologically protected'' in the sense that this mode persists against sufficiently spatially localized (even large) perturbations of  $\eta(\zeta)$,
 which preserve the asymptotic behavior, $\pm\eta_\infty$,  as $\zeta\to\pm\infty$.
\end{enumerate}
\end{proposition}

Proposition \ref{0E_prop} is proved below. This proposition and the formal expansion preceding it imply the following result on the bifurcation of edge states, which is topologically protected against arbitrary localized perturbations of the domain-wall, defined by $\eta$:

\begin{theorem}[Formal expansion of edge states]
\label{edgemode_thm}
Assume:
\begin{enumerate}
\item $A(\bx)$ is smooth, $\Lambda_h-$ periodic and Hermitian, and such that $[\pc,\LA]=0$ and $[\mathcal{R},\LA]=0$.
 \item  $(\bK_\star,E_D)$ is a Dirac point of $\LA=-\nabla\cdot A\nabla$ (see Definition \ref{dirac_pt_defn} and Theorems \ref{low-dp} and \ref{arbitrary-dp}), where
  $\bK_\star$ is a vertex of the hexagonal Brillouin zone, $\B_h$.
\item $\mathcal{L}^B \equiv -\nabla\cdot B\nabla$ is anti-$\mathcal{PC}-$ symmetric, and such that $\thetasharp^{\bK_\star}\neq0$; see \eqref{thta_shp1}.
\end{enumerate}
Let $\mathcal{L}_{\rm dw}^{(\delta)} \equiv -\nabla\cdot\left[A(\bx) +\delta\eta( \delta\ktilde_2 \cdot \bx)B(\bx)\right]  \nabla$,
denote the domain-wall operator, associated with the line-defect across the rational edge, $\R\vtilde_1$.

Then, the edge-state eigenvalue problem for $\mathcal{L}_{\rm dw}^{(\delta)}$, \eqref{dw_evp}-\eqref{localized}, for fixed parallel quasi-momentum $\kpar=\bK_\star\cdot\vtilde_1$, has the formal, \underline{topologically protected} eigenpair solution $(E^\delta_{\bK_\star},\Psi^{\delta}_{\bK_\star})$,
corresponding to a state which propagates in the $\vtilde_1-$ direction with parallel quasi-momentum $\kpar = {\bK_\star} \cdot \vtilde_1$, and is exponentially decaying in the transverse direction, as $\ktilde_2\cdot \bx \to \pm \infty$.

Furthermore, the eigenpair $(E^\delta_{\bK_\star},\Psi^{\delta}_{\bK_\star})$ can be expanded to any finite order in $\delta$ in powers of $\delta$. To leading order $E^\delta_{\bK_\star}=E_D+\mathcal{O}(\delta^2)$ and
\begin{equation}
 \label{ldord}
 \Psi^{\delta}_{\bK_\star}(\bx) =
 \left\{ \begin{array}{ll}
 \delta^{1/2}\ \gamma\ \Big[ \chi_-(\ktilde_2^\perp)\cdot\left(\Phi_1^{\bK_\star}(\bx),\Phi_2^{\bK_\star}(\bx)\right) \Big] \
 &e^{-\frac{\abs{\thetasharp^{\bK_\star}}}{\upsilon_F|\ktilde_2|}\int_0^{\delta \ktilde_2\cdot \bx}\eta(s)d s} + \mathcal{O}(\delta), \\
 & \quad \text{if} \quad \thetasharp^{\bK_\star}>0; \\\\
 \delta^{1/2}\ \gamma\ \Big[ \chi_+(\ktilde_2^\perp)\cdot\left(\Phi_1^{\bK_\star}(\bx),\Phi_2^{\bK_\star}(\bx)\right) \Big] \
 &e^{-\frac{\abs{\thetasharp^{\bK_\star}}}{\upsilon_F|\ktilde_2|}\int_0^{\delta \ktilde_2\cdot \bx}\eta(s)d s} + \mathcal{O}(\delta), \\
 & \quad \text{if} \quad \thetasharp^{\bK_\star}<0. \\\\
\end{array} \right.
\end{equation}
The constant $\gamma$ is displayed in \eqref{gamma-def}, and the factor of $\delta^{1/2}$ ensures that the edge state is normalized to unity.
\end{theorem}

\begin{proof}[Proof of Proposition \ref{0E_prop}]
Setting $E^{(1)}=0$, equation \eqref{dirac_eqn} may be rewritten as:
\[\partial_\zeta\balpha^{\bK_\star}(\zeta)\ =\ -i\ \eta(\zeta)\ \mathcal{M}^{\bK_\star}_{\mathscr{A}}(\ktilde_2)^{-1}\ \mathcal{M}^{\bK_\star}_{\LB}\ \balpha^{\bK_\star}(\zeta).\]
Using the expressions displayed in \eqref{matr-k2} we obtain:
\begin{equation}
\label{dirac_0E}
\partial_\zeta \balpha^{\bK_\star}(\zeta)\ =\
\frac{\thetasharp^{\bK_\star}}{\upsilon_F^3|\ktilde_2|^2}\ \eta(\zeta)\ \mathcal{M}^{\bK_\star}_{\mathscr{A}}(\ktilde_2^{\perp})\ \balpha^{\bK_\star}(\zeta),
\end{equation}
where $\ktilde_2=(\ktilde_2^{(1)},\ktilde_2^{(2)})$ and $\ktilde_2^\perp=(-\ktilde_2^{(2)},\ktilde_2^{(1)})$.

The system \eqref{dirac_0E} may be diagonalized using the eigenpairs of $\mathcal{M}^{\bK_\star}_{\mathscr{A}}(\ktilde_2^{\perp})$.
 Note that $\mathcal{M}^{\bK_\star}_{\mathscr{A}}(\ktilde_2^{\perp})$ has two eigenpairs:
\begin{align}\label{chip+}
\mu_+ &=+\upsilon_F|\ \ktilde_2^{\perp}|,\qquad \chi_+(\ktilde_2^{\perp})=\frac{1}{\sqrt{2}}
\begin{pmatrix}
\widehat{\mathfrak{z}}(\ktilde_2^{\perp})\\ 1
\end{pmatrix},\  \text{and}\\
\mu_-&=-\upsilon_F|\ \ktilde_2^{\perp}|,\qquad \chi_-(\ktilde_2^{\perp})=\frac{1}{\sqrt{2}}
\begin{pmatrix}
\widehat{\mathfrak{z}}(\ktilde_2^{\perp})\\ -1
\end{pmatrix},
\label{chip-}
\end{align}
where $\widehat{\mathfrak{z}}(\ktilde)= (\smallktilde^{(1)}+i\smallktilde^{(2)})/|\ktilde|$.

Therefore, \eqref{dirac_0E} has the general solution
\begin{equation}\label{zeromode_generalsolution}
\balpha^{\bK_\star}_\star(\zeta)=c_1\ e^{\frac{\thetasharp^{\bK_\star}}{\upsilon|\ktilde_2|}\ \int^{\zeta}_0\eta(s)ds}\chi_+(\ktilde_2^{\perp})+c_2\ e^{-\frac{\thetasharp^{\bK_\star}}{\upsilon|\ktilde_2|}\ \int^{\zeta}_0\eta(s)ds}\chi_-(\ktilde_2^{\perp})
\end{equation}
If $\thetasharp^{\bK_\star}>0$, we obtain an exponentially localized solution of \eqref{dirac_0E} by setting $c_1=0$ in \eqref{zeromode_generalsolution}, and if $\thetasharp^{\bK_\star}<0$, we obtain the exponentially localized solution by setting $c_2=0$ in \eqref{zeromode_generalsolution}. Imposing the normalization $\norm{\balpha^{\bK_\star}_\star}_{L^2(\R)}=1$ concludes the proof of Proposition \ref{0E_prop}.
\end{proof}

\subsection{Rigorous formulation of Theorem \ref{edgemode_thm}, the spectral no-fold condition,
and outline of the proof}\label{rigor}

A rigorous reformulation of Theorem \ref{edgemode_thm} and a detailed proof follows the strategy implemented in full detail for the case of Schr\"odinger operators in \cites{FLW-AnnalsPDE:16}. These arguments can be adapted to the current setting of divergence form operators. In this section we provide a detailed sketch of the construction, motivate the key {\it spectral no-fold condition}, and conclude the section with the reformulation in Theorem \ref{rig-edge}.

We begin by formulating the edge state eigenvalue problem  in an appropriate Hilbert space.
For a fixed edge, $\R\vtilde_1$, introduce the cylinder $\Sigma\equiv \R^2/ \Z\vtilde_1$. A function on $\Sigma$ is invariant under the shift  $\bx\to\bx+\vtilde_1$. For $s\ge0$, denote by  $H^s(\Sigma)$,  the  Sobolev space of order $s$  of  functions on $\Sigma$. Note $L^2(\Sigma)=H^0(\Sigma)$.
The pseudo-periodicity and decay conditions \eqref{pseudo-per}-\eqref{localized} are encoded by requiring $ \Psi \in H^s_\kpar(\Sigma)$ for some $s\ge0$, where
\begin{equation}
\label{H_kpar_def}
H^s_\kpar(\Sigma)\equiv \ \left\{f : f(\bx)e^{-i\frac{\kpar}{2\pi}\ktilde_1\cdot\bx}\in H^s(\Sigma) \right\}.
\end{equation}

The eigenvalue problem \eqref{dw_evp}-\eqref{localized} can then be reformulated as:
\begin{equation}
 \label{EVP}
 \mathcal{L}_{\rm dw}^{(\delta)} \Psi = E \Psi , \quad \Psi\in H_{\kpar}^2(\Sigma) .
\end{equation}
Edge states for $\kpar$ in a neighborhood of $\kpar=\bK_\star\cdot\vtilde_1$ can be constructed perturbatively.
\medskip

In analogy to Theorem 4.2 of \cites{FLW-AnnalsPDE:16}, any function in $L^2_{\kpar=\bK_\star\cdot\vtilde_1}(\Sigma)$ can be decomposed into a superposition of Floquet-Bloch modes with the pseudo-periodicity condition $f(\bx+\vtilde_1)=e^{i\bK\cdot\vtilde_1} f(\bx)$. The set of all such modes is given by:
$\{\Phi_b(\bx;\bK_\star+\lambda\ktilde_2)\}_{b\ge1}$ and where $|\lambda|\le1/2$.

\begin{theorem}\label{fourier-edge}
Let $A(\bx)$ be periodic with respect to the triangular lattice, $\Lambda_h$, and such that $\LA=-\nabla_\bx\cdot A(\bx)\ \nabla_\bx$ is strictly elliptic on $\R^2$. In particular, we may take $A(\bx)$ to be a honeycomb structured medium.
Let $f\in L_\kparvstar^2(\Sigma)= L_\kparvstar^2(\R^2/\Z\vtilde_1)$. Then,
\begin{enumerate}
\item  $f$ can be represented as a superposition of Floquet-Bloch modes of $\LA$ with quasimomenta in $\mathcal{B}_h$ located on the segment $\lambda\mapsto \bK_\star+\lambda\ktilde_2$ for $|\lambda|\le\frac{1}{2}$.
\begin{align}
f(\bx)
&=\sum_{b\ge1}\int_{-\frac12}^{\frac12} \widetilde{f}_b(\lambda) \Phi_b(\bx;\bK_\star+\lambda \ktilde_2) d\lambda  \nn \\ 
&= e^{i\bK_\star\cdot\bx} \sum_{b\ge1}\int_{-\frac12}^{\frac12} e^{i\lambda\ktilde_2\cdot\bx}\widetilde{f}_b(\lambda) p_b(\bx;\bK+\lambda \ktilde_2) d\lambda,\qquad {\rm where} \label{fb-Sigma} \\
\widetilde{f}_b(\lambda) \ &=\ \left\langle \Phi_b(\cdot,\bK_\star+\lambda\ktilde_2),f(\cdot)\right\rangle_{L_\kparv^2} . \nn
\end{align}
Here, the sum representing $e^{-i\bK_\star\cdot\bx}f(\bx)$, in \eqref{fb-Sigma} converges in the $L^2(\Sigma)$ norm.
\item In the special case where $A(\bx)=I_{2\times2}$\ ($\LA=-\Delta$):
 \begin{align*}
 f(\bx) = \sum_{{\bfm}\in\Z^2}e^{ i(\bK_\star+\bfm\vec\ktilde)\cdot\bx } \int_{-\frac12}^{\frac12}
 \widehat{f}_\bfm(\lambda)e^{i\lambda\ktilde_2\cdot\bx} d\lambda\ .
  \end{align*}
 Here, for $\bfm=(m_1,m_2)\in\Z^2$, we define $\bfm\vec\ktilde=m_1\ktilde_1+m_2\ktilde_2$.
 \end{enumerate}
\end{theorem}

We now seek $(\Psi,E)$, a solution of the eigenvalue problem \eqref{EVP}, in the form of a truncated multiscale expansion plus a corrector to be constructed:
\begin{align}
\Psi^\delta &\equiv \Psi^{(0)}(\bx,\delta\ktilde_2\cdot\bx)+
\delta\Psi^{(1)}(\bx,\delta\ktilde_2\cdot\bx)+\delta U(\bx) , \label{U-def}\\
E^\delta &\equiv E_D + \delta^2\mu \label{mu-def},
\end{align}
where  $ \Psi^{(0)}(\bx,\delta\ktilde_2\cdot\bx) = \alpha^{\bK_\star}_1(\delta\ktilde_2\cdot\bx)\Phi^{\bK_\star}_1(\bx)+\alpha^{\bK_\star}_2(\delta\ktilde_2\cdot\bx)\Phi^{\bK_\star}_2(\bx)
$. The pair  $\balpha^{\bK_\star}(\zeta)=(\alpha^{\bK_\star}_1(\zeta),\alpha^{\bK_\star}_2(\zeta))^T$ is a zero energy eigenstate of a Dirac operator: $\mathcal{D}^{\bK_\star}\balpha^{\bK_\star}(\zeta) = 0$, $\balpha^{\bK_\star}\in L^2(\R)$,
where $\mathcal{D}^{\bK_\star}$ is displayed in \eqref{dirac_op}; see \eqref{0mode_theta} of Proposition \ref{0E_prop}. $\Psi^{(1)}(\bx,\delta\ktilde_2\cdot\bx)$ is the solution of \eqref{Eq_delta1} with $E^{(1)}=0$. Both  $\Psi^{(0)}(\bx,\zeta)$ and $\Psi^{(1)}(\bx,\zeta)$ are constructed
to be $\bK_\star-$ pseudo-periodic as functions of $\bx$ and decaying as functions of $\zeta$. Finally, we seek the $U\in L^2_{\kpar=\vtilde_1\cdot\bK}$.

Substitution of \eqref{U-def} and \eqref{mu-def} into \eqref{EVP} yields an equation for the corrector $(\mu,U(\bx))$.
 By Theorem \ref{fourier-edge}, we have
 \begin{equation}
 U(\bx)=\sum_{b\ge1}\int_{-\frac12}^{\frac12} \widetilde{U}_b(\lambda) \Phi_b(\bx;\bK_\star+\lambda \ktilde_2) d\lambda\ .
 \label{U-exp}
 \end{equation}
The expression in \eqref{U-exp} is a superposition of modes corresponding to the {\it band structure slice at quasi-momentum $\bK_\star$, which is dual to the edge $\R\vtilde_1$}:
\begin{equation*}
\lambda\mapsto E_b(\bK_\star+\lambda\ktilde_2),\ \ |\lambda|\le 1/2,\ \ b\ge1.
\end{equation*}
 Figure \ref{E_k_slices} displays the first three dispersion curves of band structure slices for the zigzag and armchair edges, respectively.

The equation for $U(\bx)$ is equivalent to the coupled system for the amplitudes $\{\widetilde{U}_b(\lambda)\}_{b\ge1},\ |\lambda|\le1/2$, associated with the Floquet-Bloch modes
 $\{\Phi_b(\bx;\bK_\star+\lambda\ktilde_2)\}_{b\ge1},\ |\lambda|\le1/2$:

\begin{equation}
\label{amp_b}
 \begin{split}
& \left( E_b(\lambda)  - E_D \right) \widetilde{U}_b(\lambda) -
\delta \left\langle \Phi_b^{\bK_\star}(\cdot;\lambda) , \nabla_\bx\cdot \left[ \eta(\delta\ktilde_2\cdot\bx)  B(\bx) \nabla_\bx U(\bx) \right]  \right\rangle_{L^2_{\kparv}(\Sigma_\bx)} \\
& \qquad = \delta\widetilde{F}_b[\mu,\delta](\lambda)  + \delta^2\ \mu\ \widetilde{U}_b(\lambda) , \qquad b\ge1, \ |\lambda|\le1/2 .
\end{split}
\end{equation}

Here, $\widetilde{F}_b[\mu,\delta](\lambda)
$ is an expression involving  $\Psi^{(j)}(\bx,\delta\ktilde_2\cdot\bx)$, $j=0,1$, and their derivatives, projected onto
the Floquet-Bloch mode, $\Phi_b^{\bK_\star}(\bx;\bK_\star+\lambda\ktilde_2)$. A solution $\{\widetilde{U}_b(\lambda): b\ge1,\ |\lambda|\le1/2\}$ of the system \eqref{amp_b} is sought which satisfies $\sum_{b\ge1} \int_{-1/2}^{1/2}\ (1+|\lambda|^2)^s\ |\widetilde{U}_b(\lambda)|^2\ d\lambda<\infty$, for some $s\ge0$.

Suppose that the Dirac point $(\bK_\star,E_D)$ occurs at the intersection of the $b_*^{th}$ and $(b_*+1)^{st}$ spectral bands of $\LA$.

Following the strategy of \cite{FLW-AnnalsPDE:16}, we decompose the system \eqref{amp_b} into two subsystems consisting of near- and far- quasi-momentum Floquet-Bloch amplitudes:
\begin{itemize}
 \item $\widetilde{U}_{{\rm near}}(\lambda)=\{\widetilde{U}_{{\rm near},-}(\lambda) , \widetilde{U}_{{\rm near},+}(\lambda)\}$, for $|\lambda|\le \delta^\nu$: Floquet-Bloch amplitudes corresponding to quasi-momenta which are near the Dirac point, where
 $\{-,+\} \equiv \{b_\star, b_{\star}+1\}$.
 \item $\{ \widetilde{U}_{{\rm far},b}(\lambda)\}_{_{b\ge1}}$ for $1/2\ge|\lambda|\ge \left(\delta_{b,b_\star}+\delta_{b,b_\star+1}\right) \delta^\nu$: Floquet-Bloch amplitudes corresponding to quasi-momenta which are bounded away from the Dirac point.
\end{itemize}
Here, $ \nu>0$ is appropriately chosen and fixed.

To construct a solution for \eqref{amp_b}, we proceed via a Lyapunov-Schmidt reduction strategy.
The first step is to solve for the far-amplitudes in terms of the near-amplitudes, by constructing the mapping: $\widetilde{U}_{{\rm near}}[\mu,\delta]\mapsto \widetilde{U}_{{\rm far}}[\widetilde{U}_{{\rm near}},\mu,\delta]$.
The second step is to substitute this mapping into the near-amplitude equations to obtain a closed system of two
(nonlocal) equations for $\widetilde{U}_{{\rm near},-}(\lambda)$ and $\widetilde{U}_{{\rm near},+}(\lambda)$, for $|\lambda|\le \delta^\nu$. An appropriate rescaling of this closed system  can be solved for $\delta$ sufficiently small. In this way, the corrector $(U^\delta,\mu^\delta)$ can be constructed.

The construction  of the mapping $\widetilde{U}_{{\rm near}}[\mu,\delta]\mapsto \widetilde{U}_{{\rm far}}[\widetilde{U}_{{\rm near}},\mu,\delta]$ and the necessary estimates on this mapping requires the \emph{spectral no-fold condition},  introduced in \cites{FLW-AnnalsPDE:16}. Specifically, we require that  for $\lambda$ bounded away from $\lambda=0$ or $b\notin\{b_*,b_*+1\}$, the energy difference $E_{b}(\lambda)  - E_D$ is uniformly bounded away from zero.
This condition ensures that, for $\delta\ne0$, the $L^2_{\kpar=\bK_\star\cdot\vtilde_1}-$ spectra of the asymptotic operators $\mathcal{L}_-^{(\delta)}$ and  $\mathcal{L}_+^{(\delta)}$ (see \eqref{asymptotic_ops}) have a full spectral gap containing the energy $E=E_D$. In one-dimensional periodic problems with Dirac points (linear band crossings), the spectral no-fold condition holds automatically; see \cites{FLW-PNAS:14,FLW-MAMS:17}.

The left panel of Figure \ref{E_k_slices} displays the first three bands of the zigzag slice of the band structure for a honeycomb structure, $A(\bx)$.
By Theorem \ref{PT_thm}, the corresponding slice of the band structure of $\mathcal{L}^{(\delta)}_\pm=-\nabla_\bx\cdot\left(A(\bx)\pm\delta\eta_\infty B(\bx)\right)\nabla_\bx$ will have small local ($\lambda$ near 0) gap with energy $E_D$ in its interior.
Since the energies ($E$) attained by the dispersion curves of $\LA$, outside a neighborhood of $\lambda=0$, are bounded away from $E_D$, it follows that for $\delta$ sufficiently small  $\mathcal{L}^{(\delta)}$ has a full $L^2_{\kpar=\bK_\star\cdot\vtilde_1}-$ spectral gap, an open interval about $E_D$. That is, the spectral no-fold condition holds in this case.

The right panel of Figure \ref{E_k_slices} displays the first three bands of the armchair slice for the same choice of $A(\bx)$.
Again, by Theorem \ref{PT_thm}, the corresponding slices of the band structures of $\mathcal{L}^{(\delta)}_\pm$ have a local spectral gap near $\bk=\bK$ ($\lambda=0$). However, the intersections of the horizontal dashed line away from $\lambda=0$ indicate that for $\delta$ small there is no $L^2_{\kpar=\bK\cdot\vtilde_1}-$ spectral gap; the spectral no-fold condition fails for the armchair edge.

\begin{remark}\label{high_contrast_limit}
A honeycomb Schr\"odinger operator is an operator of the form $H^\lambda=-\Delta +\lambda^2 V(\bx)$,
where $V$ has the symmetries of a hexagonal tiling of the plane \cites{FW:12}. An important special case is where
$V(\bx)$ is a superposition of translates of  an `atomic' potential well. For any fixed rational edge,  the spectral no-fold condition as stated in  \cites{FLW-AnnalsPDE:16}, holds for all $\lambda$  sufficiently large.
This follows from the scaled-convergence of the low-lying dispersion surfaces to those of the tight-binding limiting model, proved in \cites{FLW-CPAM:17}.  In the setting of Maxwell's equations, numerical studies
indicate that for high contrast the spectral no-fold condition holds for the zig-zag edge. However, to date we have not found this condition to hold for other edges. It is an interesting question to determine for which edges the no-fold condition holds. This has implications for energy localization and is currently under investigation.  \end{remark}

%

The above strategy is implemented along the lines of the Schr\"odinger operator setting \cites{FLW-AnnalsPDE:16}.
An essential difference between the Schr\"odinger setting and the present Maxwell setting
 is that the perturbation, $\delta \nabla_\bx\cdot \eta(\delta\ktilde_2\cdot\bx)  B(\bx)\nabla_\bx U(\bx)$, which is projected  in \eqref{amp_b}, is a second order differential operator acting on the unknown $U$; compare with  Equation (7.12) of \cites{FLW-AnnalsPDE:16}.
This difference results in simple changes in the analysis of the closed system for the near-amplitudes, $\widetilde{U}_{{\rm near},-}(\lambda)$ and $U_{{\rm near},+}(\lambda)$ (step two), and to a different definition of $\thetasharp^{\bK_\star}$; see \eqref{thta_shp}). For the  Schr\"odinger case,  see Proposition 6.2 of \cites{FLW-AnnalsPDE:16}.
To obtain the  far-amplitudes as a functional of the near-amplitudes (step one), we construct the mapping $\widetilde{U}_{{\rm near}}[\mu,\delta]\mapsto \widetilde{U}_{{\rm far}}[\widetilde{U}_{{\rm near}},\mu,\delta]$ by solving the fixed point equation:
 $\mathcal{Q}_\delta \phi = \phi$, where
 {\footnotesize{
  \begin{align}
&\widetilde{\left[ \mathcal{Q}_\delta \phi \right]}_b(\lambda)\nn\\
&\quad \equiv -\delta \frac{ \chi(\abs{\lambda}\geq(\delta_{b,b_{\star}}+\delta_{b,b_{\star}+1})\delta^{\exponent})}{E_b(\lambda)-E_{\star}}
\inner{\Phi_{b}(\bx,\lambda),\nabla_{\bx}\cdot \eta(\delta \ktilde_2 \bx)B(\bx) \nabla_\bx \phi (\bx)}_{L_\kparv^2} \nn \\
&\quad + \delta^2\mu \frac{\widetilde{\phi}_b(\lambda)}{{E_b(\lambda)-E_{\star}}} ;  \label{tQ-def}
\end{align}
}}
For  $s\geq0$, introduce the Sobolev space of order $s$ consisting of  $H^s_{\kparvstar}(\Sigma)$ functions with quasi-momenta $\lesssim\delta^\nu$:
$H^s_{\rm near, \delta^\nu} \equiv \{f \in H^s_{\kparvstar}(\Sigma) : \widetilde{f}_b(\lambda) = \chi(|\lambda| \leq \delta^\nu) \widetilde{f}_b(\lambda)\}$.
One can verify, using that $(\LA+I)^{-1}\ \nabla_{\bx}\cdot \eta(\delta \ktilde_2 \bx)B(\bx) \nabla_\bx$  is a bounded operator on $H^2_{\kparvstar}$, that  $\mathcal{Q}_\delta$ maps $H^2_{\rm near, \delta^\nu}$ to $H^2_{\rm near, \delta^\nu}$ with norm bounded by ${\rm constant}\times\mathfrak{e}(\delta)$,%
where
 \begin{align}
&\mathfrak{e}(\delta)\equiv \sup_{b=\pm}\ \ \sup_{|\delta|^\exponent\le|\lambda|\le\frac12}\
 \frac{|\delta|}{|E_b(\lambda)-E_{\star}|}\nn\\
 &\quad  +\ \sup_{ b\ge1,\ b\ne\pm}\ (\ 1+|b|\ ) \sup_{0\le|\lambda|\le\frac12}
 \frac{|\delta|}{|E_b(\lambda)-E_{\star}|} .
\label{frak-e-def}
\end{align}
The spectral no-fold condition ensures that $\mathfrak{e}(\delta)=o(1)$ as $\delta\rightarrow0$.
Hence, by the contraction mapping principle the fixed point equation $\mathcal{Q}_\delta \phi = \phi$ may be solved, for $\delta>0$ and sufficiently small, on an appropriately sized ball in $H^2_{\rm near, \delta^\nu}$.

With these observations, the above strategy of proof yields the existence of edge states for the Maxwell system:

\begin{theorem}\label{rig-edge}
 Consider the $\kpar=\bK\cdot\vtilde_1-$ edge state eigenvalue problem
\eqref{EVP}.
There exist positive constants $\delta_0, c_0$ and a branch of solutions of \eqref{EVP},
   \[|\delta|\in(0,\delta_0)\longmapsto (E^\delta,\Psi^\delta)\in (E_D-c_0\ \delta_0\ ,\ E_D+c_0\ \delta_0)\times H^2_{\kparvstar}(\Sigma) , \]
   such that the following holds:
\begin{enumerate}
\item $\Psi^\delta$ is well-approximated by a slow modulation of the
degenerate Floquet-Bloch modes $\Phi_1^{\bK_\star}$ and $\Phi_2^{\bK_\star}$,
which decays to zero transverse to the edge, $\Z\vtilde_1$:
\begin{align}
&\left\|\ \Psi^\delta(\cdot)\ -\
\left[\alpha^{\bK_\star}_{\star,1}(\delta\ktilde_2 \cdot)\Phi_1^{\bK_\star}(\cdot)+\alpha^{\bK_\star}_{\star,2}(\delta\ktilde_2 \cdot )\Phi_2^{\bK_\star}(\cdot)\right]\
\right\|_{H^2_{\kparvstar}}\
\lesssim\ \delta^{\frac12} , \label{TM_Psi-error}\\
& E^\delta = E_D + \mathcal{O}(\delta^2) . \label{TM_E-error}
\end{align}
\item The amplitude vector,  $\balpha^{\bK_\star}_\star(\zeta)=\left(\alpha^{\bK_\star}_{\star,1}(\zeta),\alpha^{\bK_\star}_{\star,2}(\zeta)\right)$, is an  $L^2(\R_\zeta)- $
normalized, topologically protected zero-energy eigenstate of the Dirac operator, $\mathcal{D}^{\bK_\star}$; see \eqref{dirac_op}.
\end{enumerate}
\end{theorem}

\begin{figure}
\centering
\includegraphics[width=4in]{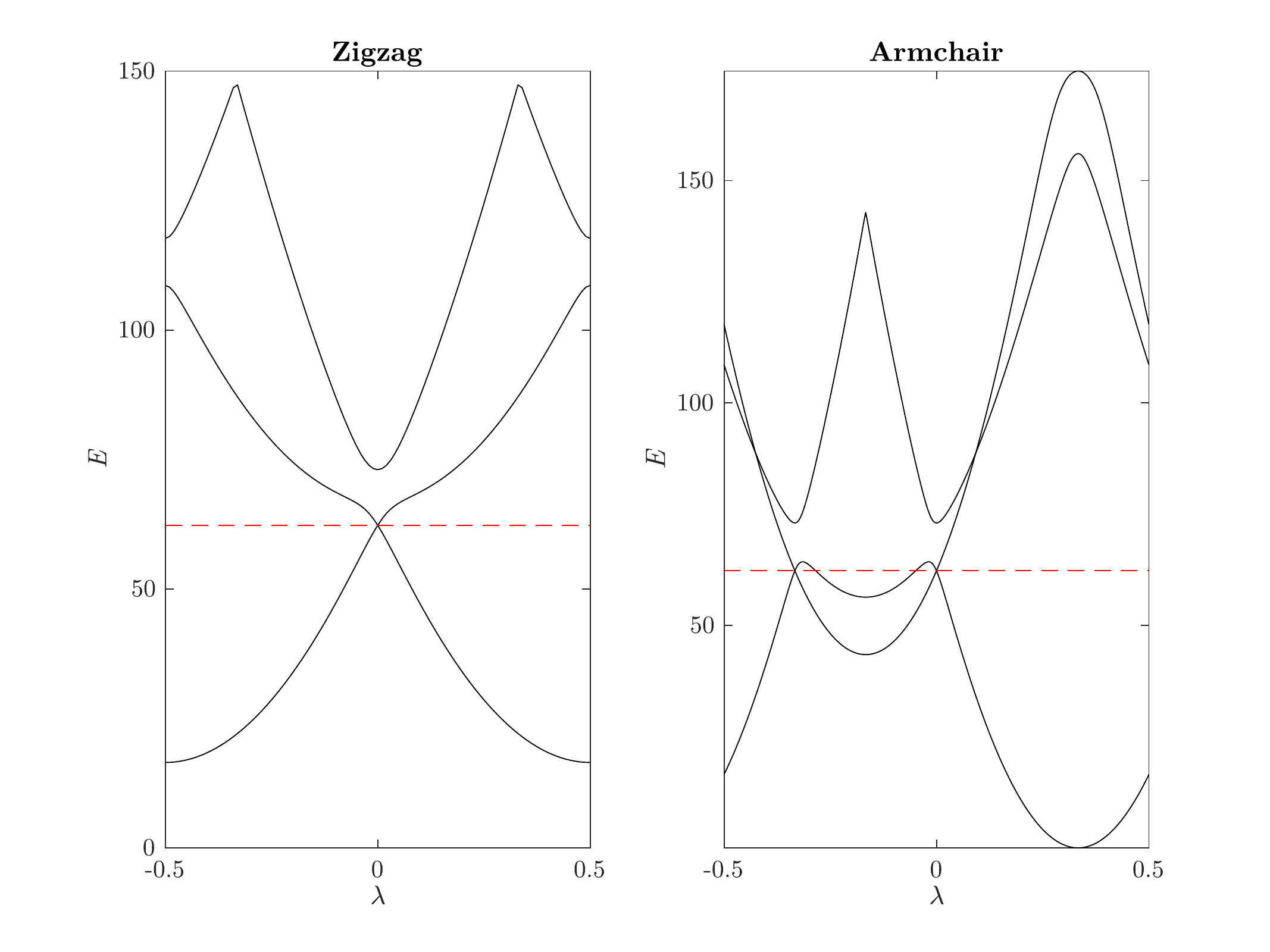}
\caption{\footnotesize
Band dispersion slices of $\LA$ along the quasi-momentum segments:  $\bK+\lambda\ktilde_2,\ |\lambda|\le1/2$, for $\ktilde_2=\bk_2$ (zigzag) and $\ktilde_2=-\bk_1+\bk_2$ (armchair).
Here, $A(\bx)=\e(x)I$ where $\e(x)$ is a particular honeycomb lattice function.
Dirac point energy levels $E=E_D$ are indicated with dotted red lines.
The spectral no-fold condition is satisfied for the zigzag slice (left panel) but not the armchair slice (right panel).}
\label{E_k_slices}
\end{figure}

\section{Edge state dispersion curves}\label{edge_state_kpar}

Under the stated hypotheses, Theorem \ref{edgemode_thm} guarantees the existence of protected edge states for $\delta>0$ sufficiently small and $\kpar=\bK_\star\cdot\vtilde_1$ fixed. Building off this result, perturbation theory can be used to construct edge state eigenpairs, $(\Psi(\cdot;\kpar),E(\kpar))$, of the eigenvalue problem \eqref{dw_evp}-\eqref{localized} for all $\kpar$ in a neighborhood of $\kpar=\bK_\star\cdot\vtilde_1$:

\begin{theorem}
 \label{near-kpar1}
 Assume hypotheses (1)-(3) of Theorem \ref{edgemode_thm}. Let $\bK_\star$ be of the $\bK$ or $\bK'$ type. Fix $\delta$ sufficiently small.
 Then there exists $c_0\ll\delta$ such that for all $\kpar$ satisfying $|\kpar - \bK_\star\cdot\vtilde_1| < c_0$, there exists a
 $\kpar\mapsto \left(E_{\bK_\star}(\kpar),\Psi_{\bK_\star}(\bx;\kpar)\right)$ edge state solution of the eigenvalue problem \eqref{dw_evp}-\eqref{localized} (equivalently \eqref{EVP}). Here, the energy
 $E^\delta_{\bK_\star}(\kpar)= E^\delta_{\bK_\star} + \mu^{\delta} (\kpar-\bK_\star\cdot\vtilde_1) +\mathcal{O}(|\kpar-\bK_\star\cdot\vtilde_1|^2)$, where $E^{\delta}_{\bK_\star}=E_D+\mathcal{O}(\delta^2)$ is the eigenvalue of the edge state given in \eqref{edgemode_thm}, and $\mu^\delta$ is a constant, which is independent of $\kpar$.
\end{theorem}

By taking a continuous superposition of the edge states given by Theorem \ref{near-kpar1}, we obtain dispersive wavepackets which are concentrated along the edge. These wavepackets have parallel-quasi-momentum components concentrated near $\bK_\star\cdot\vtilde_1$, and their group velocity is determined by $\frac{dE^{\delta}(\kpar)}{d\kpar}$, for $\kpar$ near $\bK_\star\cdot\vtilde_1$.

\begin{theorem}\label{gp-vel}
Assume hypotheses (1)-(3) of Theorem \ref{edgemode_thm}. Let $\bK_\star$ be of the $\bK$ or $\bK'$ type.
Then,
\begin{equation}
\label{GV2_thm}
 \frac{d E^{\delta}}{d\kpar}\Big|_{_{\kpar=\bK_\star\cdot\vtilde_1}} =
 \left\{
 \begin{array}{ll}
 - \frac{1}{2\pi}\frac{\upsilon_F}{|\ktilde_2|}(\ktilde_2\wedge\ktilde_1) + \mathcal{O}(\delta^{1/2})
 & \quad \text{if}\quad \thetasharp^{\bK_\star}>0 ; \\
 \\
 +\frac{1}{2\pi}\frac{\upsilon_F}{|\ktilde_2|}(\ktilde_2\wedge\ktilde_1) + \mathcal{O}(\delta^{1/2})
  & \quad \text{if}\quad \thetasharp^{\bK_\star}<0 .
  \end{array}\right.
\end{equation}
\end{theorem}

\begin{proof}[Proof of Theorem \ref{gp-vel}]
The edge state eigenvalue problem with $\kpar-$ pseudo-periodicity is given \eqref{dw_evp}-\eqref{localized}.
It is convenient to factor the Bloch phase with parallel quasi-momentum, $\kpar$, in order to replace \eqref{pseudo-per} by an equivalent periodic boundary condition. Thus, we set
$\Psi(\bx; \kpar) = e^{i\frac{\kpar}{2\pi} \ktilde_1 \cdot \bx} \psi(\bx; \kpar)$, where $\psi(\bx+\vtilde_1; \kpar)=\psi(\bx; \kpar)$. Further, we set
$
\frac{\kpar}{2\pi} \ktilde_1 \equiv \bK_\star + \frac{\xi}{2\pi} \ktilde_1$, with  $ |\xi| \ll 1$ .
Since $\ktilde_1\cdot\vtilde_1=2\pi$, it follows that
\[\kpar=\bK_\star\cdot\vtilde_1+\xi.\]
Substituting
\[\Psi(\bx; \kpar) = e^{i\frac{\kpar}{2\pi} \ktilde_1 \cdot \bx} \psi(\bx; \kpar) = e^{i (\bK_\star + \frac{\xi}{2\pi} \ktilde_1) \cdot \bx} \psi(\bx; \xi)\]
into \eqref{dw_evp}-\eqref{localized}, yields an equation for
 $\psi(\bx;\xi)$:
\begin{align}
&\left[\mathcal{L}_{_{\rm dw}}^{(\delta)}(\bK_\star)+\frac{\xi}{2\pi} \ktilde_1\cdot\left(
\mathscr{A}(\bK_\star)+\delta \mathscr{Q}^\delta(\bK_\star) \right) \right. \label{kpar_evp_psi1}  \\
&\qquad \left. + \left(\frac{\xi}{2\pi}\right)^2 \ktilde_1 \cdot (A + \delta \eta^\delta B) \ktilde_1 \right] \psi^\delta (\bx; \xi)  \nn
= E^\delta (\xi) \psi^\delta (\bx; \xi) , \\
&\psi(\bx+\vtilde_1; \xi) = \psi(\bx;\xi) , \qquad \psi(\bx; \xi) \to 0\ \ {\rm as}\ \ |\bx\cdot\ktilde_2|\to\infty . \label{kpar_evp_psi3}
\end{align}

Here, $\mathscr{A}$ and $\mathscr{Q}^\delta$ are first order operators given by:
\begin{align*}
\mathscr{A}&=\frac{1}{i}\nabla\cdot A + A\frac{1}{i}\nabla,\qquad
\mathscr{A}(\bb k)=e^{-i\bk\cdot \bx}\ \mathcal{A}\ e^{-i\bk\cdot \bx} ; \quad \text{see also \eqref{Ak_def}}; \nn\\
\mathscr{Q}^\delta&=\frac{1}{i}\nabla\eta^\delta B+\frac{1}{i}\eta^\delta B\nabla,\qquad \mathscr{Q}^\delta(\bb k)=e^{-i\bk\cdot \bx}\ \mathscr{Q}^\delta\ e^{-i\bk\cdot \bx}.
\end{align*}
Recall that $E^\delta(\xi=0)=E_{K_\star}^\delta$ is an edge state eigenvalue with corresponding $L^2(\Sigma)-$ (normalized) eigenstate $\psi^\delta (\bx; 0)=e^{-i\bK_\star\cdot\bx} \Psi^\delta(\bx;0)=e^{-i\bK_\star\cdot\bx} \Psi^\delta_{\bK_\star}$.
Therefore, differentiating \eqref{kpar_evp_psi1} with respect to $\xi$ and setting $\xi=0$ yields
\begin{align}
\label{ham_xi_diff1}
&\left( \mathcal{L}_{_{\rm dw}}^{(\delta)}(\bK_\star) - E_{\bK_\star}^\delta \right) \dot\psi^\delta (\bx; 0)\nn\\
&\quad = \dot E^\delta(0) \psi^\delta (\bx; 0)-\frac{1}{2\pi} \ktilde_1\cdot
\left( \mathscr{A}(\bK_\star)+\delta \mathscr{Q}^\delta(\bK_\star) \right) \psi^\delta (\bx; 0) ,
\end{align}
where we use the notation $\dot f(0) \equiv \D_\xi f(\xi)|_{\xi=0}$.

Taking the $L^2(\Sigma)-$ inner product  of \eqref{ham_xi_diff1}
 with $\psi^\delta (\bx; 0)$, we obtain the following expression for the group velocity at $\xi=0$ ($\kpar=\bK_\star\cdot\vtilde_1$):
 \begin{align}
 \frac{d E^\delta}{d\xi}(0)
 & = \frac{1}{2\pi} \inner{\psi^\delta (\bx; 0) , \ktilde_1\cdot
 \left( \mathscr{A}(\bK_\star)+\delta \mathscr{Q}^\delta(\bK_\star) \right) \psi^\delta (\bx; 0)}_{_{L^2(\Sigma)}} \nn\\
 & = \frac{1}{2\pi} \inner{ \Psi^\delta_{\bK_\star} , \ktilde_1\cdot \mathscr{A}\ \Psi^\delta_{\bK_\star} }_{_{L^2(\Sigma)}} + \mathcal{O}(\delta) .
 \label{E_xi_0}
\end{align}

Since $\delta$ is small, we may further simplify the expression for the group velocity in \eqref{E_xi_0} by substituting the expansion for  $\Psi_{\bK_\star}^{\delta}(\bx)$, displayed in \eqref{ldord} of Theorem \ref{edgemode_thm}.
N.B. In \eqref{ldord}, the positive sign corresponds to the case  $\vartheta_\sharp<0$ and the negative sign to the case $\vartheta_\sharp>0$.

We substitute \eqref{ldord} into \eqref{E_xi_0}. This inner product is a sum of terms of the form:
\begin{equation}
 \inner{ \Gamma_1(\bx,\delta\bk_2\cdot\bx), \mathscr{A}(\bx) \Gamma_2(\bx,\delta\bk_2\cdot\bx)}_{_{L^2(\Sigma)}},
\end{equation}
where $\bx\mapsto\Gamma_j(\bx,\zeta)$ is $\Lambda_h-$ periodic and $H^2(\Omega_h)$  with values in $L^2(\R_\zeta)$, and $\mathscr{A}\equiv\mathscr{A}(\bx)=\frac{1}{i}\nabla_{\bx}\cdot A(\bx) + A(\bx)\frac{1}{i}\nabla_{\bx}$.
Therefore, we may exploit the multiple scale character to evaluate the inner product; see Lemma 7.9 of \cites{FLW-AnnalsPDE:16}. We obtain:

\begin{align}
 \frac{d E^\delta}{d\xi}(0)\ =\
 &\frac{1}{2\pi}\ \ktilde_1 \cdot
 \inner{ \chi_\pm(\ktilde_2^\perp)\cdot\left(\Phi_1^{\bK_\star},\Phi_2^{\bK_\star}\right) ,
 \mathscr{A}\  \chi_\pm(\ktilde_2^\perp)\cdot\left(\Phi_1^{\bK_\star},\Phi_2^{\bK_\star}\right) }_{L^2(\Omega_{h})}\nn\\
 &\qquad + \mathcal{O}(\delta^{1/2}),
 \label{GV1}
 \end{align}

where we have also substituted the value of $\gamma$, displayed in \eqref{gamma-def}.
Again we emphasize that in \eqref{GV1}
the positive sign corresponds to the case
$\vartheta_\sharp<0$ and the negative sign to the case
 $\vartheta_\sharp>0$.
Finally, observe that
\begin{align}
& \ktilde_1 \cdot \inner{ \chi_\pm(\ktilde_2^\perp)\cdot\left(\Phi_1^{\bK_\star},\Phi_2^{\bK_\star}\right) ,
  \mathscr{A}\ \chi_\pm(\ktilde_2^\perp)\cdot\left(\Phi_1^{\bK_\star},\Phi_2^{\bK_\star}\right) }_{L^2(\Omega_{h})} \nn\\
& \quad=\chi_\pm^\dagger(\ktilde_2^{\perp})\mathcal{M}_\mathscr{A}(\ktilde_1)\ \chi_\pm(\ktilde_2^{\perp}) \nn \\
&= \pm\frac{\upsilon_F}{|\ktilde_2|}
 \Re\left(\overline{\mathfrak{z}(\ktilde_2^\perp)} \mathfrak{z}(\ktilde_2)\right) \nn \\
 &= \pm \frac{\upsilon_F}{|\ktilde_2|}(\ktilde_1^{(2)}\ktilde_2^{(1)}-\ktilde_1^{(1)}\ktilde_2^{(2)})
 = \left\{ \begin{array}{ll}
 +\frac{\upsilon_F}{|\ktilde_2|}(\ktilde_2\wedge\ktilde_1) & \quad \text{if} \quad \vartheta_\sharp<0; \\
 -\frac{\upsilon_F}{|\ktilde_2|}(\ktilde_2\wedge\ktilde_1) & \quad \text{if} \quad \vartheta_\sharp>0.
 \end{array} \right.
 \label{phipminner}
\end{align}
Here, $\ktilde_2\wedge\ktilde_1\equiv (\ktilde_1^{(2)}\ktilde_2^{(1)}-\ktilde_1^{(1)}\ktilde_2^{(2)})$.

Substituting \eqref{phipminner} into \eqref{GV1}, we obtain
\begin{equation}
\label{GV2}
 \frac{d E^\delta}{d\xi}(0) =
 \left\{
 \begin{array}{ll}
 - \frac{1}{2\pi}\frac{\upsilon_F}{|\ktilde_2|}(\ktilde_2\wedge\ktilde_1) + \mathcal{O}(\delta^{1/2})
 & \quad \text{if} \quad \thetasharp^{\bK_\star}>0 ;\\
 \\
 + \frac{1}{2\pi}\frac{\upsilon_F}{|\ktilde_2|}(\ktilde_2\wedge\ktilde_1) + \mathcal{O}(\delta^{1/2})
  & \quad \text{if}\quad \thetasharp^{\bK_\star}<0 .
  \end{array} \right.
\end{equation}
Equation \eqref{GV2_thm} and Theorem \ref{gp-vel} follow from \eqref{GV2} using the change of variables: $\xi=\kpar-\bK_\star\cdot\vtilde_1$.
\end{proof}

\subsection{The direction of energy propagation along the edge}\label{direction}

In this section we apply Theorem \ref{gp-vel} to two particular situations where $\LB=-\nabla\cdot B\nabla$ is $\pc$ anti-symmetric: $\pc\LB=-\LB\pc$.
\begin{enumerate}
\item  Preserving $\mathcal{C}-$ invariance and breaking $\mathcal{P}-$ invariance;
$[\mathcal{C},\LB]=0$ and $\mathcal {P}\LB=-\LB\mathcal{P}$.
\item  Preserving $\mathcal{P}-$ invariance and breaking $\mathcal{C}-$ invariance;
$[\mathcal{P},\LB]=0$ and $\mathcal {C}\LB=-\LB \mathcal{C}$.
\end{enumerate}

First note that because $\overline{L^2_{\bK,\bar\tau}}=L^2_{\bK',\tau}$, we have that $\Phi_2^{\bK}(\bx)=\mathcal{C}[\Phi_1^{\bK'}](\bx)$. Recall also that $\pc\Phi_2^{\bK}=\Phi_1^{\bK}$.
%

\emph{Case (1): $[\mathcal{C},\LB]=0$ and  $\mathcal {P}\LB=-\LB\mathcal{P}$; bi-directional propagation along the edge.}
We claim that $\thetasharp^{\bK'}=-\thetasharp^{\bK}$.  Indeed, we have
%
%
\begin{align*}
\thetasharp^{\bK'}
& = \inner{\Phi_1^{\bK'},\LB\Phi_1^{\bK'}}_{L^2(\Omega)}
 = \inner{\mathcal{C}\Phi_2^{\bK},\LB\mathcal{C}\Phi_2^{\bK}}_{L^2(\Omega)}
  = \inner{\mathcal{P}\Phi_1^{\bK},\LB\mathcal{P}\Phi_1^{\bK}}_{L^2(\Omega)}\\
& = -\inner{\mathcal{P}\Phi_1^{\bK},\mathcal{P}\LB\Phi_1^{\bK}}_{L^2(\Omega)}
 = -\inner{\Phi_1^{\bK},\LB\Phi_1^{\bK}}_{L^2(\Omega)}
 = -\thetasharp^{\bK} .
\end{align*}

In particular, $\thetasharp^{\bK'}$ and
$\thetasharp^{\bK}$ have opposite sign. Moreover, by \eqref{GV2_thm}:

\begin{equation}
\lim_{\delta\to0} \D_\kpar E^\delta (\bK\cdot\vtilde_1) = - \lim_{\delta\to0} \D_\kpar E^\delta(\bK'\cdot\vtilde_1) .
\label{E_xi_diff}
\end{equation}

The quantities $\D_\kpar E^\delta(\bK\cdot\vtilde_1)$ and $\D_\kpar E^\delta(\bK'\cdot\vtilde_1)$ are, respectively, the group velocities of wavepackets, constructed as a superposition of edge states,
 $(E^\delta(\kpar),\Psi^\delta(\bx;\kpar))$ for quasi-momentum $\kpar$ near $\bK\cdot\vtilde_1$ and $\bK'\cdot\vtilde_1$, respectively. We conclude that such packets with quasi-momentum centered at $\kparv$ and $\kparpv$, travel in opposite directions and at the same  speed; see Figure \ref{E_kpar_zz} (top panel) and the discussion in the Introduction (Section \ref{intro}).

\emph{Case 2: $[\mathcal{P},\LB]=0$ and  $\mathcal {C}\LB=-\LB\mathcal{C}$; uni-directional propagation along the edge.}
In contrast to Case 1, here we claim that $\thetasharp^{\bK'}=\thetasharp^{\bK}$.  Indeed, we have
\begin{align*}
\thetasharp^{\bK'}
& = \inner{\Phi_1^{\bK'},\LB\Phi_1^{\bK'}}_{L^2(\Omega)}
 = \inner{\mathcal{C}\Phi_2^{\bK},\LB\mathcal{C}\Phi_2^{\bK}}_{L^2(\Omega)}
  = \inner{\mathcal{P}\Phi_1^{\bK},\LB\mathcal{P}\Phi_1^{\bK}}_{L^2(\Omega)}\\
& = \inner{\mathcal{P}\Phi_1^{\bK},\mathcal{P}\LB\Phi_1^{\bK}}_{L^2(\Omega)}
 = \inner{\Phi_1^{\bK},\LB\Phi_1^{\bK}}_{L^2(\Omega)}
 = \thetasharp^{\bK} .
\end{align*}

Arguing analogously to Case 1 we have  find for Case 2:
\begin{equation}
\lim_{\delta\to0} \D_\kpar E^\delta (\bK\cdot\vtilde_1) = \lim_{\delta\to0} \D_\kpar E^\delta(\bK'\cdot\vtilde_1) .
\label{E_xi_diff2}
\end{equation}

Hence, the group velocities of wavepackets, constructed as a superposition of edge states,
 $(E^\delta(\kpar),\Psi^\delta(\bx;\kpar))$ for quasi-momentum $\kpar$ near $\bK\cdot\vtilde_1$, respectively, near $\bK'\cdot\vtilde_1$ are equal in magnitude and sign.  We conclude, in this case where $\mathcal{C}-$ invariance is broken,  that such packets with quasi-momentum centered at $\kparv$ and $\kparpv$, travel in the same direction and at the same  speed. See Figure \ref{E_kpar_zz} (bottom panel) and the discussion in the Introduction (Section \ref{intro}). Such unidirectional propagation of topologically protected edge states is well-known to be the hallmark of topological insulators \cites{hasan2010colloquium,ando2013topological,bernevig2013topological,ortmann2015topological}.

\appendix

\section{ Maxwell eqns and the operators $\LA$ and $\mathcal{L}_{\rm dw}^{(\delta)}$}\label{maxwell}

In this Appendix we apply the main results of this paper to Maxwell's equation in media with magneto-optic and bi-anisotropic constitutive  laws.  We begin with a self-contained outline of Maxwell's equations in such media
and their reduction to the setting of our main analytical results.  We also discuss the sense in which these media are time-reversal invariant.

 In appropriate units \cites{joannopoulos2011photonic}, Maxwell's system for the electromagnetic field in a medium with no free charges or currents  is given by
\begin{equation}
\label{Maxwell_eq}
\begin{split}
&\nabla\times \bb E=-\frac{\partial \bb B}{\partial t},\quad\quad \nabla\times\bb H=\frac{\partial \bb D}{\partial t}\\
&\nabla\cdot \bb D=0,\quad\quad\nabla\cdot\bb B=0.
\end{split}
\end{equation}

Here, $\bb E$ and $\bb H$ are the macroscopic electric and magnetic fields, $\bb D$ and $\bb B$ are the displacement and magnetic induction fields, respectively. Each field $\bb F=(F_1, F_2, F_3)^T$ is a vector-valued function of the time $t\in \mathbb{R}_+$ and the spatial coordinate $\bx=(x_1,x_2,x_3)\in\R^3$. The standard vector notations $\nabla\times$, $\nabla\cdot$ and $\nabla$, are used for the curl, divergence, and gradient.

The displacement and magnetic induction fields $(\bb D, \bb B)$ are related to the electromagnetic fields $(\bb E, \bb H)$ by a material-dependent constitutive relation.
The constitutive relation of linear loss-free material has the general form
\begin{equation}
\label{constitutive_relation_Rmatrix}
\begin{pmatrix}\bb D\\ \bb B\end{pmatrix}=\widehat{\mathbf R}\begin{pmatrix}\bb E\\ \bb H\end{pmatrix}\equiv\begin{pmatrix}\hat{\epsilon}&\hat{\xi}\\ \hat{\xi}^{\dagger}&\hat{\mu}\end{pmatrix}\begin{pmatrix}\bb E\\ \bb H\end{pmatrix},
\end{equation}
where the constitutive matrix $\widehat{\mathbf R}=\widehat{\mathbf R}(\bx)$ is a $6\times6$ positive-definite Hermitian matrix.
The permittivity tensor $\hat \epsilon$ and the permeability tensor $\hat \mu$ are $3\times 3$ are positive-definite Hermitian matrices. The bianisotropy tensor, $\hat{\xi}$, and  its conjugate-transpose $\hat{\xi}^{\dagger}$  couple the magnetic and electric fields.

In terms of $\bb E$ and $\bb H$, Maxwell's equations become
\begin{equation}\label{Maxwell_eq2}
\frac{\partial}{\partial t}\widehat{\bb R} \begin{pmatrix}
                               \bb E \\
                               \bb H \\
                             \end{pmatrix}=\begin{pmatrix}
                               \nabla \times\bb H \\
                              - \nabla\times \bb E \\
                             \end{pmatrix}.
\end{equation}
For the purpose of discussing time-reversibility, it is convenient to express Maxwell's equations \eqref{Maxwell_eq2} as a   system of Schr\"odinger-type:
\begin{equation}\label{Maxwell_eq3}
i\partial_t {\bf \Psi}= \mathscr{M}_{\widehat{\bb R}} {\bf \Psi}
\end{equation}
for the electric and magnetic fields, ${\bf \Psi}=(\bb E,\bb H)^T$, with Maxwell operator  $ \mathscr{M}_{\widehat{\bb R}}$:
\begin{equation}\label{Maxwell_operator}
\mathscr{M}_{\widehat{\bb R}}\equiv {\widehat{\bb R}}^{-1}\begin{pmatrix}
  0&+i\nabla\times\\-i\nabla\times&0
\end{pmatrix}.
\end{equation}
Note that $\mathscr{M}_{\widehat{\bb R}}$ is self-adjoint with the respect the weighted inner product:
\[
\left\langle {\bf F},{\bf G}\right\rangle_{\widehat{\bb R}}\ =\ \left\langle {\bf F}\ ,\ \widehat{\bb R}\ {\bf G}\right\rangle_{L^2(\R^3;\C^6)}
\]

\subsection{Time-reversal symmetry and Maxwell's equations}\label{Maxwell_symmetry}
In this section we discuss the sense in which Maxwell's equations with the above class of constitutive relations is time-reversal invariant. Our discussion is motivated by the treatment in  \cites{De_Nittis-Lein:14,De_Nittis_Lein:17}. Introduce the transformations
\begin{align}
& \mathcal{T}_r=\sigma_3\otimes I: \ (\bb E, \bb H)~\mapsto (\bb E, -\bb H), \quad \text{and} \label{TR}\\
& \mathcal{T}_c= \left(\sigma_3\otimes I\right)\circ \mathcal{C}:\  (\bb E, \bb H)~\mapsto (\overline{\bb E}, -\overline{\bb H})\ .\label{TC}
\end{align}
used to express time-reversibility. Note that $\mathcal{T}_c$ is introduced for the case of complex-valued constitutive tensors, $\widehat{\bb R}$.

\begin{proposition}[Time reversibility] \label{Tinvariant}
Let $\widehat{\bb R}$ be the constitutive matrix defined in \eqref{constitutive_relation_Rmatrix},
$\mathscr{M}_{\widehat{\bb R}}$ be the associated Maxwell operator defined in \eqref{Maxwell_operator}, and  $\mathcal{T}_r$ and $\mathcal{T}_c$ be the transformation operators defined in \eqref{TR} and \eqref{TC}.
Then
\begin{enumerate}
\item $\mathcal{T}_r\ \widehat{\bb R}\ \mathcal{T}_r=\widehat{\bb R}$ if and only if $\hat{\xi}=0$.
 Under this condition we have
$\mathcal{T}_r\ e^{-it\mathscr{M}_{\small \widehat{\bb R}}}\ \mathcal{T}_r = e^{+it\mathscr{M}_{\widehat{\bb R}}}$.
Therefore, Maxwell's equations are invariant under the transformation: $t\to -t$, $(\bb E, \bb H)~\mapsto (\bb E, -\bb H)$.
\item $\mathcal{T}_c\ \widehat{\bb R}\ \mathcal{T}_c=\widehat{\bb R}$ if and only if $\overline{\hat{\epsilon}} =\hat{\epsilon}, \quad \overline{\hat{\mu}}=\hat{\mu},\quad \overline{\hat{\xi}}=-\hat{\xi}$.
Under these conditions we have $\mathcal{T}_c\ e^{-it\mathscr{M}_{{\small \widehat{\bb R}}}}\ \mathcal{T}_c = e^{+it\mathscr{M}_{\small \widehat{\bb R}}}$ and therefore Maxwell's equations are invariant under the transformation: $t\to -t$, $(\bb E, \bb H)~\mapsto (\overline{\bb E}, -\overline{\bb H})$.
\end{enumerate}
\end{proposition}

\begin{proof}
The following direct calculation is the key of the proof
\begin{equation*}
\begin{pmatrix}I&0\\0&-I\end{pmatrix}\begin{pmatrix}\bb R_1&\bb R_2\\ \bb R_3 & \bb R_4\end{pmatrix}\begin{pmatrix}I&0\\0&-I\end{pmatrix}=\begin{pmatrix}\bb R_1&-\bb R_2\\ -\bb R_3&\bb R_4\end{pmatrix}.
\end{equation*}
Therefore $\mathcal{T}_r \ \widehat{\bb R}\ \mathcal{T}_r=\widehat{\bb R}$ if and only if $\hat{\xi}=-\hat{\xi}=0$.
 Furthermore,
  \begin{equation*}
\mathcal{T}_r\ \begin{pmatrix}
  0&+i\nabla\times\\-i\nabla\times&0
\end{pmatrix}\ \mathcal{T}_r = -\ \begin{pmatrix}
  0&+i\nabla\times\\-i\nabla\times&0
\end{pmatrix}.
\end{equation*}
It follows that if $\mathcal{T}_r \widehat{\bb R} \mathcal{T}_r = \widehat{\bb R}$, then $\mathcal{T}_r \mathscr{M}_{_{\widehat{\bb R}}} \mathcal{T}_r = -\mathscr{M}_{\widehat{\bb R}}$ and furthermore
$\mathcal{T}_r e^{-it\mathscr{M}_{\small \widehat{\bb R}}}\mathcal{T}_r=e^{+it\mathscr{M}_{\small \widehat{\bb R}}}$.
The proof for part (2) is similar.
\end{proof}

\subsection{In-plane propagation in 2-dimensional media }

We next impose simplifying constraints on the material weight matrix, $\widehat{\bb R}$, following \cites{Shvets-PTI:13}.
Consider a material in which the constitutive matrix only varies in the transverse plane with coordinates $\bx_\perp=(x_1,x_2)$, and is invariant with respect to translations in the longitudinal direction, with coordinate $x_3$. Thus, $ \widehat{\bb R}=\widehat{\bb R}(\bx_\perp)$. Furthermore, we assume the coupling entries in $\widehat{\bb R}$ between transverse and longitudinal directions to be zero. Bianisotropy is assumed to exist only in the transverse directions. Therefore, the constitutive tensors are of the form
\begin{equation}\label{Constit_Relat}
  \hat \epsilon(\bx_\perp)=\begin{pmatrix}\hat{\epsilon}_\perp &0\\0&\epsilon_{3}
  \end{pmatrix},\;  \hat \mu(\bx_\perp) =\begin{pmatrix}\hat{\mu}_\perp &0\\0&\mu_{3}
  \end{pmatrix},\; \hat \xi(\bx_\perp) =\begin{pmatrix}\hat{\xi}_\perp&0\\0&0
  \end{pmatrix},
\end{equation}
where $\hat{\epsilon}_\perp$, $\hat{\mu}_\perp$, $\hat{\xi}_\perp$ are $2\times 2$ matrices depending on $\bx_\perp$. We re-express Maxwell's equations in terms of transverse and longitudinal components: ${\bf \Psi}_\perp=({\bb E}_\perp, ~{\bb H}_\perp)^T$ and ${\bf \Psi}_\parallel=(E_3, ~H_3)^T$ of the electromagnetic field. Thus,  we introduce the $4\times4$ transverse constitutive tensor $\widehat{\bb R}_\perp$ and the $2\times 2$ longitudinal constitutive tensor $\widehat{\bb R}_\parallel$
\begin{equation}
\widehat{\bb R}_\perp=\begin{pmatrix}\hat{\epsilon}_\perp&\hat{\xi}_\perp\\ \hat{\xi}^{\dagger}_\perp&\hat{\mu}_\perp\end{pmatrix},\quad \widehat{\bb R}_\parallel=\begin{pmatrix}\epsilon_3&0\\ 0&\mu_3\end{pmatrix}.
\label{tcon}\end{equation}
 $\widehat{\bb R}_\perp$ and $\widehat{\bb R}_\parallel$ are  Hermitian matrices and positive-definite, and hence invertible.

We consider the in-plane propagating electromagnetic waves, i.e., all fields are independent of $x_3$.  
Substituting the constitutive relations \eqref{tcon} into \eqref{Maxwell_eq3}, we obtain
\begin{equation}\label{tran}
\partial_t\begin{pmatrix}\bb E_\perp\\ \bb H_\perp \end{pmatrix}=\widehat{\bb R}_\perp^{-1}\begin{pmatrix}-J \nabla_\perp\  H_3\\ J \nabla_\perp\ E_3\end{pmatrix}.
\end{equation}
and
\begin{equation}
\label{EHlong}
\partial_t\begin{pmatrix}E_3\\ H_3 \end{pmatrix}=\widehat{\bb R}_\parallel^{-1}\begin{pmatrix} J\nabla_\perp \cdot \bb H_\perp\\ -J \nabla_\perp \cdot \bb E_\perp\end{pmatrix}.
\end{equation}
To obtain \eqref{tran}-\eqref{EHlong} we have used that if $\bb F=(F_1(\bb x_\perp), F_2(\bb x_\perp),F_3(\bb x_\perp))$, then
$
(\nabla \times \bb F)_\perp=-J\nabla_\perp F_3, \ \text{and} \ (\nabla \times \bb F)_\parallel=J\nabla_\perp\cdot \bb F_\perp,
$,
where $\nabla_\perp=(\partial_{x_1} , \partial_{x_2})^T$ and $J=\begin{pmatrix}0&-1\\1&0\end{pmatrix}$.

Using \eqref{tran} to eliminate $(\bb E_\perp, \bb H_\perp)$ in \eqref{EHlong}, we obtain a closed system of longitudinal fields $H_3$ and $E_3$. Specifically, we first denote the entries the inverse of $\widehat{\bb R}_\perp$ as follows
\begin{equation}
\label{Rinverse}
\widehat{\bb R}_\perp^{-1}(\bx_\perp) \equiv \begin{pmatrix}\hat {e}& \hat{\theta}\\ \hat{\theta}^{\dagger}&\hat{m}
\end{pmatrix}
\end{equation}
where $\hat e, \hat\theta, \hat m$ are $2\times 2$ matrices depending on $\bx_\perp$.%

The closed system for $(E_3,H_3)$ is
\begin{equation}\label{H3E32}
\begin{split}
&\mu_3\ \partial_{t}^2H_3~=~-\left(
\mathcal{L}^{A_e}H_3-
\mathcal{L}^{B}E_3\right),\\
&\epsilon_3\ \partial_{t}^2 E_3~=~-\left(\mathcal{L}^{A_m} E_3-
\mathcal{L}^{B^\dagger} H_3\right),
\end{split}
\end{equation}
where $\mathcal{L}^A=-\nabla\cdot A\nabla$ ( see \eqref{L_def} ) and
\begin{equation}
A_e=J^T \hat{e} J ,\quad A_m=J^T \hat{m} J ,\quad B=J^T \hat{\theta} J.
\label{AeAmB}
\end{equation}
Since $\hat{e}$ and $\hat{m}$ are positive definite Hermitian matrices, so are the matrices $A_e$ and $A_m$.
 The Hermitian matrix $B$, arises due to bianisotropy and couples $E_3$ and $H_3$ in \eqref{H3E32}.
 Note that the system \eqref{H3E32} is invariant under $t\mapsto t'=-t$, consistent with our earlier discussion of time-reversal symmetry;
  see Proposition \ref{Tinvariant}.


To construct the full electro-magnetic field for in-plane propagating electromagnetic waves in a general 2-dimensional bianisotropic medium we may first obtain $(E_3, H_3)$ from \eqref{H3E32}. Then, the transverse components ${\bf E}_\perp$ and  ${\bf H}_\perp$ can be obtained from \eqref{tran}.

Consider now time-harmonic solutions of Maxwell's equation:
\[ (\bb E(t,\bx), \bb H(t,\bx))\to e^{-i\omega t}({\bb E}(\bx), {\bb H}(\bx)).\]
 The system  \eqref{H3E32} now becomes
\begin{equation}\label{H3E33}
\begin{split}
&\mu_3\ \omega^2\ H_3~=~\mathcal{L}^{A_e} H_3-~\mathcal{L}^{B} E_3
\\ &\epsilon_3\ \omega^2\ E_3~=\mathcal{L}^{A_m} E_3-~\mathcal{L}^{B^\dagger} H_3.
\end{split}
\end{equation}
and by \eqref{tran} the corresponding transverse fields are
\begin{equation}\label{Trans2}
\begin{pmatrix}
\bb E_\perp\\
\bb H_\perp
\end{pmatrix}=\frac{i}{\omega}\widehat{\bb R}_\perp^{-1}\begin{pmatrix}-J\nabla_\perp H_3\\ J \nabla_\perp E_3\end{pmatrix} = \frac{i}{\omega}\ \begin{pmatrix}\hat {e}& \hat{\theta}\\ \hat{\theta}^{\dagger}&\hat{m}
\end{pmatrix}\ \begin{pmatrix}-J\nabla_\perp H_3\\ J \nabla_\perp E_3\end{pmatrix}
\end{equation}

Suppose the bianisotropy of the medium is negligible, {\it i.e.} $\hat{\xi}=0$. From \eqref{tcon}, we then have that  $\hat{\bf R}_\perp^{-1}$, defined in \eqref{Rinverse} is diagonal  with  $\hat e_\perp=\hat \epsilon_\perp^{-1}$, $\hat m_\perp=\hat \mu_\perp^{-1}$ and $\hat{\theta}=0$. Furthermore, equations \eqref{H3E32} for $H_3$ and $E_3$ is a  decoupled system wave equations.

 Assuming for simplicity that  $\mu_3$ and $\epsilon_3$ are constant, and rescaling so that $\mu_3=1$ and $\epsilon_3=1$, we find that \eqref{H3E32} reduces to a pair of decoupled wave equations of the form \eqref{LA-wave} discussed in Remark \ref{wavepkts}:
%
\begin{equation}\label{H3E322}
\begin{split}
&\D_t^2 H_3~=~\nabla_\perp \cdot A_e \nabla_\perp H_3=-\mathcal{L}^{A_e} H_3,\\
&\D_t^2 E_3~=~\nabla_\perp \cdot A_m \nabla_\perp E_3=-\mathcal{L}^{A_m} H_3.
\end{split}
\end{equation}

Further, we may decompose the electromagnetic vector $({\bf E},{\bf H})$  into its decoupled TE- and TM- parts,
\begin{align}
\textrm{TE:}\qquad {\bf E}(\bx_\perp)&=\begin{pmatrix} E_1(\bx_\perp)\\ E_2(\bx_\perp)\\ 0\end{pmatrix},\qquad
{\bf H}(\bx_\perp)=\begin{pmatrix} 0\\ 0\\ H_3(\bx_\perp)\end{pmatrix}\label{TE}\\
\textrm{TM:}\qquad {\bf E}(\bx_\perp)&=\begin{pmatrix} 0\\ 0\\ E_3(\bx_\perp)\end{pmatrix},\qquad
{\bf H}(\bx_\perp)=\begin{pmatrix} H_1(\bx_\perp)\\ H_2(\bx_\perp)\\ 0\end{pmatrix}\label{TM}
\end{align}
 starting with $H_3$, in the absence of bianisotropy ($\hat\xi=0$)  we construct the transverse electrical components ${\bf E}(\bx_\perp)$ in terms of $H_3$ alone from \eqref{Trans2}. Similarly, starting with $E_3$, we construct the transverse magnetic components ${\bf H}(\bx_\perp)$ in terms of $E_3$ alone from \eqref{Trans2}.
 \bigskip

 \subsection{ Frequencies $\pm\omega$ and eigenvalue the $E$}\label{Eomega}
 In the main body of this article we study the
eigenvalue problem $\LA\psi=E\psi$, where $\LA$ is self-adjoint and positive definite. The  eigenvalues, $E$, are real and positive.
Consider now the case of $TE-$ modes in the electromagnetic setting. Given an eigensolution $(\psi,E)$ of $\LA\psi=E\psi$, we set $E=\omega^2$ and $H_3(\bx)=\psi(\bx)$, and obtain  time-harmonic longitudinal magnetic components of frequency $-\omega$ and $\omega>0$, respectively:
 $e^{-i\omega t}H_3(\bx)$ and $e^{+i\omega t}H_3(\bx)$.

Therefore, via the relation \eqref{Trans2},  we have that     $(\bb E_\perp(\bx),H_3(\bx))=(-\frac{i}{\omega\epsilon_\perp} J\nabla_\perp H_3, H_3)$ is the TE mode corresponding to frequency $\omega$;
and  $(\bb E_\perp(\bx),H_3(\bx))=(+\frac{i}{\omega\epsilon_\perp} J\nabla_\perp H_3, H_3)$ is the TE mode corresponding to frequency $-\omega$.

\subsection{Edge states in magneto-optic and bianisotropic media}\label{HR-K}
In the following two subsections we show that the photonic edge states of
Haldane and Raghu for magneto-optic materials \cites{HR:07,RH:08}, and Khanikaev {\it et al} \cites{Shvets-PTI:13}
for bianisotropic materials are covered by the mathematical framework and analysis of this paper.

\subsubsection{Magneto-optic materials \cites{HR:07}}\label{HR}
In magneto-optic materials with the Faraday-rotation effect, the polarization of light is rotated in the transverse plane which is perpendicular to the external magnetic field. Here we neglect the bianisotropy and the spatial variation of the magnetic permeability. Thus, we have
\begin{equation}
\widehat{\bb R}_\perp=\begin{pmatrix}
  \epsilon I_{_{2\times2}}-\gamma \sigma_2&0\\ 0&\mu I_{_{2\times2}}
\end{pmatrix}\ =\
\begin{pmatrix}
  \epsilon I_{_{2\times2}}- i\gamma J_{_{2\times2}}&0\\ 0&\mu I_{_{2\times2}}
\end{pmatrix}
\label{magneto}
\end{equation}
and
\begin{equation}
\widehat{\bb R}_\perp^{-1}\ =\ \begin{pmatrix}\hat {e}& \hat{\theta}\\ \hat{\theta}^{\dagger}&\hat{m}
\end{pmatrix}
\ =\ \begin{pmatrix}
  \frac{\epsilon}{\epsilon^2-\gamma^2} I_{_{2\times2}}+\frac{\gamma}{\epsilon^2-\gamma^2} \sigma_2&0\\ 0&\mu^{-1} I_{_{2\times2}}
\end{pmatrix}.
\end{equation}
Here,  $\gamma$ is real-valued and denotes the strength of the Faraday-rotation.

Since bianisotropy has been neglected, {\it i.e.} $\hat\xi=0$, and the electric and magnetic fields
${\bf E}_\perp$ and  ${\bf H}_\perp$ are not coupled by the constitutive tensor $\hat{\bf R}_\perp$.  Furthermore, by Proposition \ref{Tinvariant}, Maxwell's equation with the  constitutive relation \eqref{magneto} is invariant under
the transformations: $t\to t'=-t$, $({\bf E}, {\bf H})^T\mapsto \mathcal{T}_r ({\bf E}, {\bf H})^T=({\bf E}, -{\bf H})^T$.

If the strength of the Faraday-rotation is weak, {\it i.e.} $\gamma\ll\epsilon$, then
$\epsilon/(\epsilon^2-\gamma^2)=\epsilon^{-1}+\mathcal{O}(\gamma^2/\epsilon^2)$,  and $\gamma/(\epsilon^2-\gamma^2)=\epsilon^{-2}\gamma+\mathcal{O}(\gamma^2/\epsilon^2)$.
Specializing to the TE mode, \eqref{TE},  taking $\mu$ to be constant and scaling so that
$\mu=1$,  we have by \eqref{H3E32}
\begin{equation}
\partial^2_t H_3-\left[\ \nabla_\perp\cdot \epsilon^{-1}\nabla_\perp\ +\ \nabla_\perp \cdot\gamma\epsilon^{-2}\sigma_2\nabla_\perp\ \right] H_3=0.
\end{equation}
Time-harmonic solutions then satisfy
\begin{equation}
\label{TE_mag-op}
-\nabla_\perp \cdot [\epsilon^{-1}I+\gamma\epsilon^{-2}\sigma_2]\nabla_\perp H_3=\omega^2 H_3.
\end{equation}
The operator on the left hand side of  \eqref{TE_mag-op} is of the class  \eqref{L_delta} studied in Section \ref{edge_states}. In this case, $A(\bx_\perp)=\epsilon^{-1}(\bx_\perp) I_{_{2\times2}}$ is taken to be a honeycomb structured media for which $[\pc,\LA]=0$; see Section \ref{honey-media}.  Furthermore, $B=(\epsilon^2)^{-1}\sigma_2=(\epsilon^2)^{-1}\ i\ J$ for which $[\mathcal{P},\LB]=0$ and $\mathcal{C}\LB=-\LB\mathcal{C}$.  Our analytical results imply:

\begin{enumerate}
\item If $\gamma=0$, then there exist Dirac points within the band structure of \eqref{TE_mag-op}. See Section \ref{dirac-pts}, and in particular Theorem \ref{prop_conical} and Theorem \ref{low-dp}.

\item If $\gamma$ is a small constant, then since $\mathcal{C}-$ symmetry is broken,  a \underline{local} spectral gap opens about the Dirac point of the unperturbed operator; see Section \ref{dirac_persistence}

\item If the bulk structure, $\epsilon(\bx_\perp)$, satisfies the spectral no-fold hypothesis and $\gamma=\delta\eta(\delta\ktilde_2\cdot\bx_\perp)$ is a domain wall function in the sense of Definition \ref{domain_wall_defn}, then there exist uni-directional TE edge state
curves $\kpar\mapsto e^{\pm i\omega(\kpar)t}(\bb E_\perp(\bx_\perp;\kpar),H_3(\bx_\perp;\kpar))$; see Section \ref{edge_states}. In particular, these curves are defined in an open neighborhood of the parallel quasi-momenta $\kpar=\bK\cdot \vtilde_1$ and $\kpar=\bK^\prime\cdot \vtilde_1$.  The group velocities computed at these parallel quasi-momenta are equal; see \eqref{E_xi_diff2}.  Hence, {\it edge wave-packets} which are constructed from a superposition of edge modes in this neighborhood propagate in the same direction; propagation is unidirectional.

\quad The two dark (red) curves in the Figure \ref{E_kpar_zz} (top panel), $\kpar\mapsto E^{\bK}(\kpar)$ and $\kpar\mapsto E^{\bK'}(\kpar)$,  correspond to two edge modes  propagating along a zigzag edge. Our analytical results  construct these curves in a neighborhood of $\kpar=\bK\cdot\bv_1=2\pi/3$ and $\kpar=\bK'\cdot\bv_1=-2\pi/3\equiv 4\pi/3\ (mod\ 2\pi)$,  respectively. These (unidirectional) edge state curves, defined for $\kpar\in[0,2\pi]$, are computed numerically. The edge-mode electromagnetic frequency curves are obtained from the relation $\pm\omega(\kpar)=\pm\sqrt{E(\kpar)}$.
\end{enumerate}

\subsubsection{2-D bianisotropic meta-materials \cites{Shvets-PTI:13}}
Here, we consider  constitutive tensors of the form
\begin{equation}
\widehat{\bb R}_\perp\ =\ \begin{pmatrix}\epsilon I_{_{2\times2}} & -\chi \sigma_2 \\ -\chi \sigma_2 & \mu I_{_{2\times2}}\end{pmatrix}
\ =\ \begin{pmatrix}\epsilon I_{_{2\times2}} & -i \chi J_{_{2\times2}} \\ -i \chi J_{_{2\times2}} & \mu I_{_{2\times2}}\end{pmatrix}
\label{bian}\end{equation}
and
\begin{equation}
\widehat{\bb R}_\perp^{-1}\ =\ \begin{pmatrix}\hat {e}& \hat{\theta}\\ \hat{\theta}^{\dagger}&\hat{m}
\end{pmatrix}
\ =\ \begin{pmatrix}\frac{\mu}{\mu\epsilon-\chi^2} I_{_{2\times2}} & \frac{\chi}{\mu\epsilon-\chi^2} \sigma_2 \\ \frac{\chi}{\mu\epsilon-\chi^2} \sigma_2 & \frac{\epsilon}{\mu\epsilon-\chi^2} I_{_{2\times2}}\end{pmatrix},
\end{equation}
where $\epsilon,\mu>0$ are the transverse principle permittivity and permeability, the real parameter $\chi$ represents the Pasteur parameter which induces a phase delay of the electric polarization from the magnetic field \cites{bianisotropy_book}.
From Proposition \ref{Tinvariant}, Maxwell's equation with the  constitutive relation \eqref{bian} is invariant under
the transformations: $t\to t'=-t$, $({\bf E}, {\bf H})^T\mapsto \mathcal{T}_c ({\bf E}, {\bf H})^T=(\overline{\bf E}, -\overline{\bf H})^T$.

In bianisotropic media one typically has $|\chi|\ll \sqrt{\mu\epsilon}$ giving the following expansions of the entries of
$\widehat{\bb R}_\perp^{-1}$:\
$ \mu/(\mu\epsilon-\chi^2)=\epsilon^{-1}+\mathcal O(\chi^2/(\mu\epsilon))$,\ $ \epsilon/(\mu\epsilon-\chi^2)=\mu^{-1}+\mathcal O(\chi^2/(\mu\epsilon))$ and $\chi/(\mu\epsilon-\chi^2)=\chi/(\mu\epsilon)+\mathcal O(\chi^2/(\mu\epsilon))$.
Retaining only the dominant order terms of the constitutive matrices,  the system  \eqref{H3E32} becomes
\begin{equation}
\label{bianiso1}
\begin{split}
&\mu_3\partial_t^2H_3-\nabla_\perp\cdot \epsilon^{-1}\nabla_\perp H_3 -\nabla_\perp\cdot \frac{\chi}{\mu\epsilon}\sigma_2\nabla_\perp E_3=0,\\
&\epsilon_3\partial_t^2E_3-\nabla_\perp\cdot \mu^{-1}\nabla_\perp E_3-\nabla_\perp\cdot \frac{\chi}{\mu\epsilon}\sigma_2\nabla_\perp H_3=0.
\end{split}
\end{equation}

In \cites{Shvets-PTI:13}, the situation where $\eps=\mu$ is studied.

In typical materials, the permittivity and permeability are different, i.e. $\epsilon\neq \mu$.  A meta-material design for which $\epsilon_\perp\approx\mu_\perp$ and $\epsilon_3\approx\mu_3\approx {\rm constant}$ is discussed in \cites{Shvets-PTI:13}.

We rescale so that  $\eps_3=\mu_3=1$. Introduce the spin-like variables
\begin{equation}
  \Psi_\pm=H_3\pm E_3.
\end{equation}
Then, $\Psi_+$ and $\Psi_-$ satisfy the pair of \underline{decoupled}  wave equations:
 \footnote{This is related to the following observation. Consider 1-dimensional Maxwell's equations: $\epsilon(x) E_t= H_x$, $\mu(x)H_t=E_x$. Introduce the electromagnetic wave speed  $c(x)=1/\sqrt{\epsilon(x)\mu(x)}$ and wave-impedance $Z=\sqrt{\mu(x)/\epsilon(x)}$ \cite{FGPS:07}. Assume $Z$ is constant. Then the wave components $E+ Z H$ and $E-Z H$ are decoupled and unidirectionally counter-propagating with speed $c(x)$ which, in general, is non-constant.}
\begin{equation}
\label{spin1}
\partial_t^2\Psi_\pm-\nabla_\perp\cdot [\epsilon^{-1}I\pm  \epsilon^{-2}\chi \sigma_2]\nabla_\perp \Psi_\pm=0.
\end{equation}
Time harmonic solutions are of the form: $\Psi_\pm(\bx)e^{-i\omega t}$ where
\begin{equation}\label{dw_bianiso}
\omega^2\Psi_\pm=\mathcal{L}_{\rm dw,b}^{\pm}\Psi_\pm\equiv-\nabla_\perp\cdot [\epsilon^{-1}I\pm\epsilon^{-2}\chi \sigma_2]\nabla_\perp \Psi_\pm.
\end{equation}

The operators $\mathcal{L}_{\rm dw,b}^{\pm}$ are of the class \eqref{L_delta} studied in Section \ref{edge_states}. Similar to last example, in this case,  $A=\epsilon^{-1}I$ is taken to be a honeycomb structured media for which $[\pc,\LA]=0$; see Section \ref{honey-media}.  Furthermore, $B=(\epsilon^2)^{-1}\sigma_2=(\epsilon^2)^{-1}\ i\ J$ for which $[\mathcal{P},\LB]=0$ and $\mathcal{C}\LB=-\LB\mathcal{C}$.  Our analytical results imply:
%
%
\begin{enumerate}
\item If $\chi=0$, then there exist Dirac points in the dispersion surfaces  of  \eqref{dw_bianiso} for both ``spin $+$''  and ``spin $-$'' operators $\mathcal{L}_{\rm dw,b}^{\pm}$ .
\item If $\chi$ is a small constant, then since $\mathcal{C}-$ symmetry is broken, a \underline{local} spectral gap opens about the Dirac point of the unperturbed operator for both ``spin $+$''  and ``spin $-$''  states; see Section \ref{dirac_persistence}.
\item  If the bulk structure, $\epsilon(\bx_\perp)$, satisfies the spectral no-fold hypothesis and
$\chi=\delta\eta(\delta\ktilde_2\cdot\bx_\perp)$, where $0<\delta\ll1$, is a domain wall function in the sense of Definition \ref{domain_wall_defn}, then there exist unidirectional ``spin $+$''  and uni-directional ``spin $-$'' edge states. The previous discussion of Subsection \ref{HR} applies separately to ``spin $\pm$'' states.

\quad Let $E^{(+)}(k_\parallel)$ denote either of the two  edge state dispersion curves  corresponding to $\mathcal{L}^+_{\rm dw,b}$; see the discussion in Section \ref{HR}.  Since  $\overline{\mathcal{L}^+_{\rm dw,b}}=\mathcal{L}^-_{\rm dw,b}$, it follows that if $(E_*, \Psi(\bx))$ is the edge state pair at $k_\parallel$ for the operator $\mathcal{L}^+_{\rm dw,b}$, then $(E_*, \overline{\Psi(\bx)})$ is the edge state pair at $-k_\parallel$ for $\mathcal{L}^-_{\rm dw,b}$. In other words, $E^{(+)}(k_\parallel)=E^{(-)}(-k_\parallel)=E^{(-)}(2\pi-k_\parallel)$ ($-\kpar\equiv2\pi-\kpar$\ $mod\ 2\pi$),  where
$E^{(-)}(k_\parallel)$ is a dispersion curve associated with $\mathcal{L}^-_{\rm dw,b}$.  Therefore, $\D_\kpar E^{(+)}(k_\parallel)=-\D_s E^{(-)}(s)\Big|_{s=-k_\parallel}=-\D_s E^{(-)}(s)\Big|_{s=2\pi-k_\parallel}$. Since the group velocities of edge states with like spin are the same, it follows that edge states of opposite spin components are counterpropagating and decoupled.
\end{enumerate}

\section{Numerical methods}\label{numerical_schemes}

Let a honeycomb structured medium be defined by $A(\bx)$ (Section \ref{honey-media}), where $\Omega$ is its period cell.
The study of Dirac points of the Floquet-Bloch eigenvalue problem of $\LA$ \eqref{L_def} in Sections \ref{dirac-pts}-\ref{dirac_persistence} can be formulated as a family of periodic eigenvalue problems for the operator $\LA(\bk)=-(\nabla+i\bk)\cdot A (\nabla+i\bk)$ acting on $L^2(\Omega)$, parametrized by $\bk\in\B_h$. %

The study of $\vtilde_1-$ edge states of $\mathcal{L}_{\rm dw}^{(\delta)}$ \eqref{dw_ham} in Section \ref{edge_states} can be formulated as an eigenvalue problem for the operator
 $\LA^{(\delta)}_{\rm dw}(\kpar)\equiv -(\nabla + \frac{\kpar}{2\pi}\ktilde_1)\cdot (A(\bx)+\delta \eta(\delta\ktilde_2\cdot\bx)B(\bx) (\nabla + \frac{\kpar}{2\pi}\ktilde_1)$ on the cylinder $\Sigma=\R^2/\Z\vtilde_1$
  where  periodicity is imposed in the $\vtilde_1$ direction and a decaying boundary condition is imposed in the transverse,
   $|\ktilde_2\cdot\bx|\to\infty$, direction.  For eigenvalue problems on the cylinder, $\Sigma$,  we truncate the unbounded direction of the cylinder and impose a Dirichlet boundary condition; see the discussion below.

We use three different numerical schemes, depending on the geometry ($\Omega$ vs. $\Sigma$) and the form of the matrices $A(\bx)$ and $B(\bx)$, to reduce these  PDE eigenvalue problems to algebraic eigenvalue problems.
In each case the resulting algebraic eigenvalue problem can be solved with a sparse eigenvalue solver.

\begin{enumerate}
\item \emph{Spectral methods}.

To study the family of periodic eigenvalue problems, we expand $A(\bx)$ in Fourier series and reduce the periodic eigenvalue problem: $\LA(\bk)\psi=E\psi$ to a family of algebraic systems of equations, parameterized by $\bk$, for the Fourier coefficients of $\psi(\bx;\bk)$,
\item \emph{Finite difference methods}.

For diagonal $A(\bx)$ and $B(\bx)$ matrices (see \eqref{dw_ham}), including the case of $\mathcal{P}-$ symmetry breaking, the eigenvalue problems: $\mathcal{L}_{\rm dw}^{(\delta)}(\kpar)\psi=E\psi$, $\psi\in L^2(\Sigma)$, can be reduced to a family of sparse algebraic eigenvalue problems, parametrized by $\kpar\in[0,2\pi]$.

Some care needs to taken when computing on the (non-rectangular) honeycomb lattice. We find it more convenient to work in the standard, rectangular basis in the plane with coordinates denoted $\by=(y_1,y_2)$, rather than in the triangular basis, $\{\bv_1,\bv_2\}$, with coordinates denoted $\bx=(x_1,x_2)$. We therefore transform coordinates $\bx \mapsto \by$: $\by=x_1\bv_1+x_2\bv_2$. Under this transformation $\frac{\D}{\D{y_j}} = \sum_{k=1}^2\frac{\D x_k}{\D y_j} \frac{\D}{\D{x_k}},\ j=1,2$.

\item \emph{Finite element methods}. For non-diagonal $A(\bx)$ or $B(\bx)$ matrices, including the case of $\mathcal{C}-$ symmetry breaking, it is difficult to preserve the self-adjointness using a finite difference approach to discretize the edge state eigenvalue problems, $\mathcal{L}_{\rm dw}^{(\delta)}\psi=E\psi$, $\psi\in L^2(\Sigma)$.
In these cases, we found it simpler to use a finite element method to discretize the problem. Our finite element method was implemented in FEniCS \cites{logg2012automated}.
\end{enumerate}

Finally, we comment on the spurious (gray) edge modes observed in Figure \ref{E_kpar_zz} (top panel).
These modes are spatially localized at the computational boundary, where the cylinder $\Sigma$ is truncated. They are heavily dependent on finiteness effects associated with the choice of numerical domain.
In Figure \ref{E_L_comparison}, we compare bifurcation diagrams for several different finite cylinder truncations, parametrized by the cylinder-length, $L$. The (red) topologically protected (domain-wall induced) edge modes stabilize and are essentially independent of the cylinder truncation, while the spurious (blue) edge modes, localized at the truncated cylinder boundary, are heavily dependent on the choice of numerical domain truncation.
A computational method for which  spurious, $L-$ dependent edge modes, would not occur would require the use of appropriate perfectly matched or  radiation conditions; see {\it e.g.} \cite{Fliss-Joly:16}.

\begin{figure}
\centering
\includegraphics[width=5in]{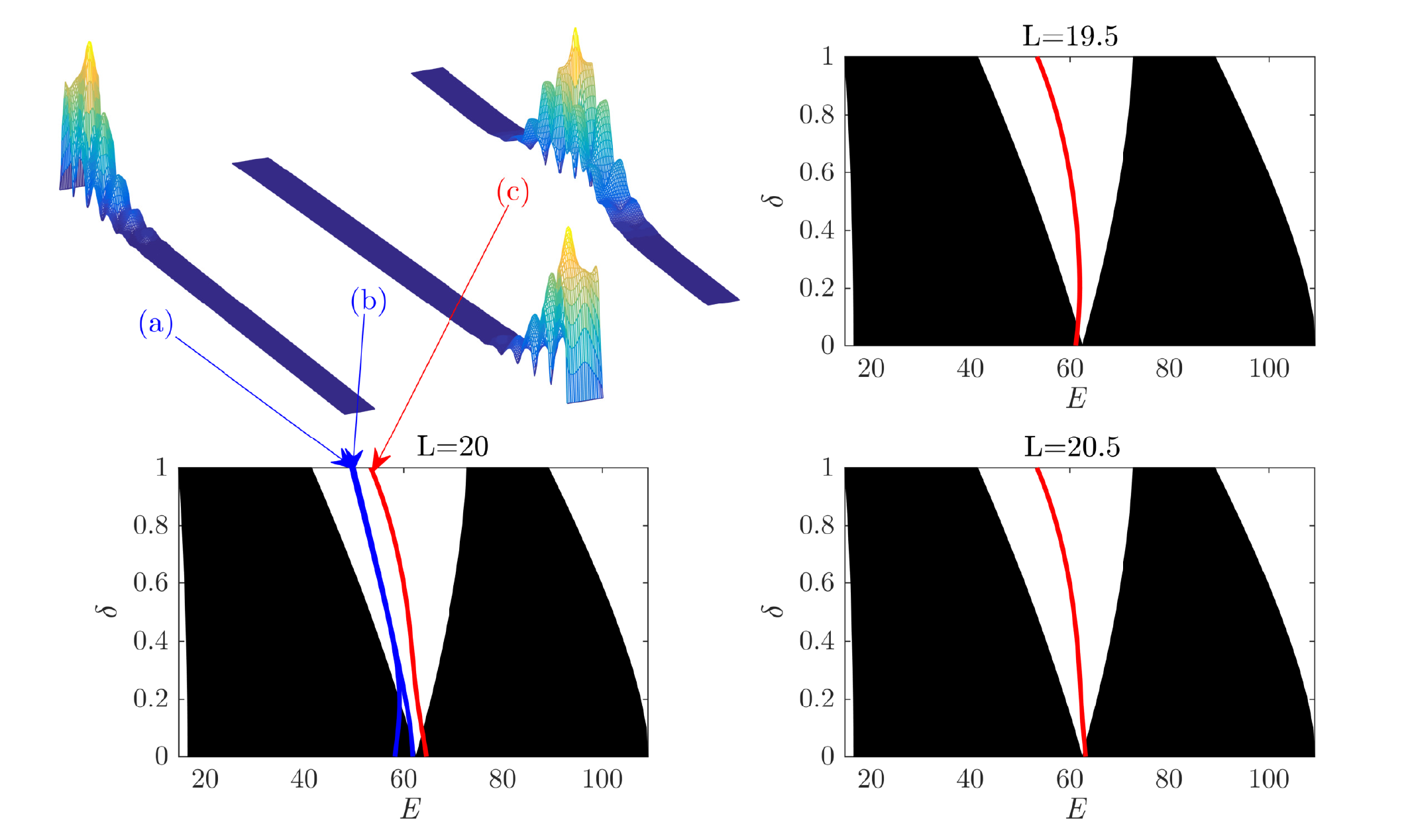}
\caption{\footnotesize
 Bifurcation curves for the zigzag edge, illustrated by the $L^2_{\kpar=\bK\cdot\bv_1=2\pi/3}(\Sigma)-$ energy spectrum of $\mathcal{L}_{\rm dw}^{(\delta)}$ vs. the perturbation parameter, $\delta$, for various cylinder-lengths, $L$.
 $\mathcal{L}_{\rm dw}^{(\delta)}$ is chosen as in the top panel of Figure \ref{E_kpar_zz}.
 Edge modes (a)-(c) correspond to the respective blue and red energies at $\delta=1$ for the case where $L=20$.
 Modes (a)-(b) are spurious hard edge modes. Mode (c) is a true edge mode.
 As we vary $L$, we observe that (red) true edge mode energy curves are independent of the numerical setup, whereas the (blue) spurious edge energy curves are heavily dependent on the size of the numerical domain.}
\label{E_L_comparison}
\end{figure}

\section{Figure potentials}\label{fig_potentials}

\begin{figure}
\centering
\includegraphics[width=5in]{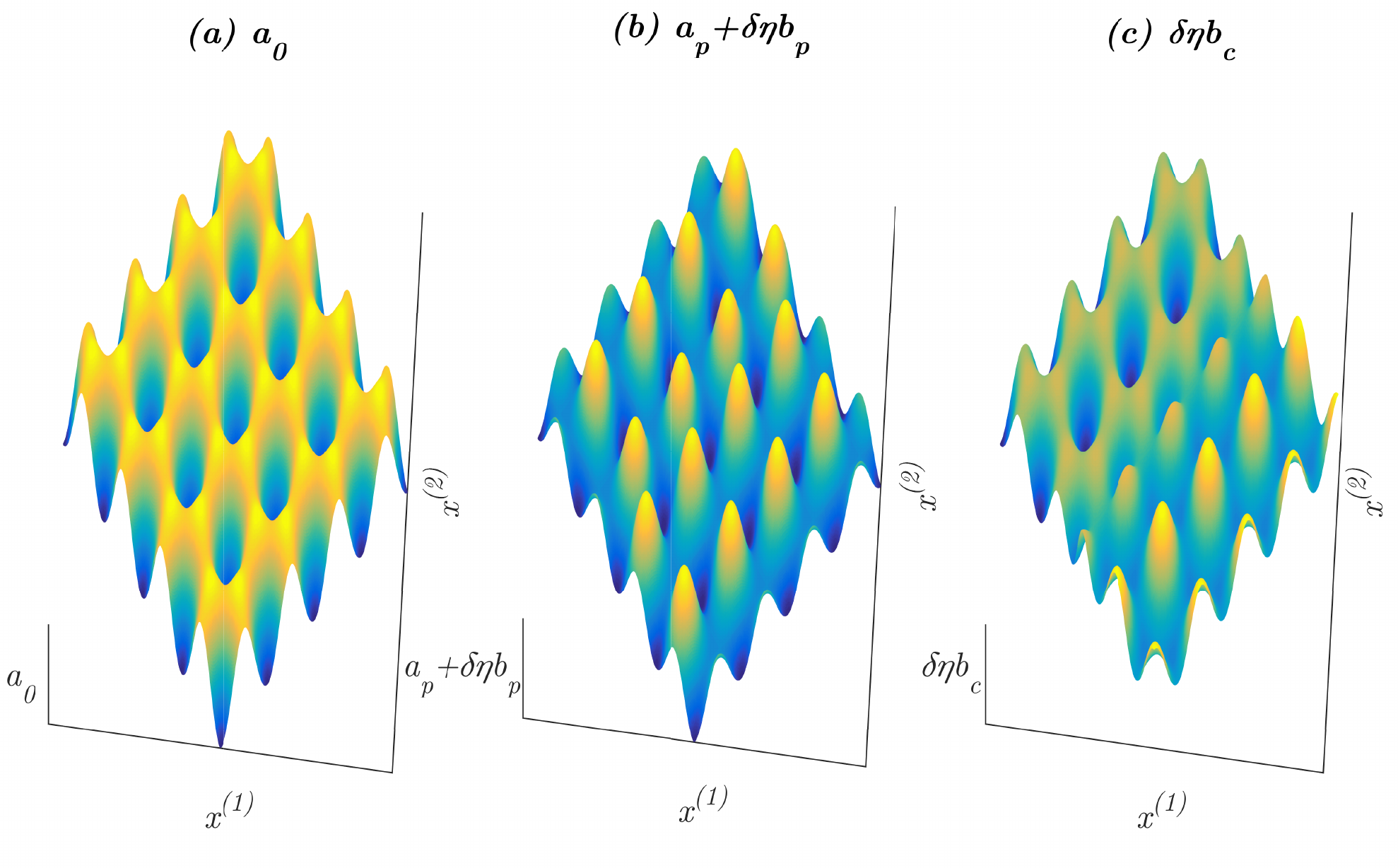}
\caption{\footnotesize
Plots of the potentials used in spectral plots throughout the text.
{\bf (a)}: Potential $\e_0(\bx)$ given in \eqref{unperturbed_potential}.
{\bf (b)}: Potential $\e_p(\bx)+\delta\eta(\delta\bk_2\cdot\bx)b_p(\bx)$ (for $\delta=1$), where $\eta(\zeta)=\tanh(\zeta)$, and $\e_p(\bx)$ and $b_c(\bx)$ are given in \eqref{P_breaking_potential_a}-\eqref{P_breaking_potential_b}.
{\bf (c)}: Potential $\delta\eta(\delta\bk_2\cdot\bx)b_c(\bx)$ (for $\delta=1$), where $\eta(\zeta)=\tanh(\zeta)$ and $b_c(\bx)$ is defined in \eqref{C_breaking_potential_b}.
}
\label{all_potentials}
\end{figure}

The following potentials were used in the various figures appearing throughout the text.
\begin{enumerate}
 \item [$\bullet$] Figures \ref{E_k_BZ} and \ref{E_k_slices} (unperturbed media):
 The honeycomb structured media is $A(\bx)=\e_0(\bx)I$, where
 \begin{equation}
 \label{unperturbed_potential}
  \e_0(\bx)= 4 - \eps \left( \cos(\bk_1\cdot\bx) + \cos(\bk_2\cdot\bx) + \cos((\bk_1+\bk_2)\cdot\bx) \right),
 \end{equation}
 with $\eps=1$.
 The potential \eqref{unperturbed_potential} is plotted in Figure \ref{all_potentials}(a).
 \item [$\bullet$] Bottom panels of Figures \ref{E_delta_zz} and \ref{E_kpar_zz} ($\mathcal{P}-$ symmetry breaking):
 The honeycomb structured media with $\mathcal{P}-$ symmetry broken for $\delta>0$ is $A(\bx)+\delta\eta(\delta\bk_2\cdot\bx)B(\bx)$, where $A(\bx)=\e_p(\bx)I$, $B(\bx)=b_p(\bx)I$ and $\eta(\zeta)=\tanh(\zeta)$.
 Here,
 \begin{align}
 \e_p(\bx) &= 4.5 - \left( \cos(\bk_1\cdot\bx) + \cos(\bk_2\cdot\bx) + \cos((\bk_1+\bk_2)\cdot\bx) \right),  \label{P_breaking_potential_a} \\
 b_p(\bx) &= \sin(\bk_1\cdot\bx) + \sin(\bk_2\cdot\bx) - \sin((\bk_1+\bk_2)\cdot\bx)  \label{P_breaking_potential_b} .
 \end{align}
 The potential $\e_p(\bx)+\delta\eta(\delta\bk_2\cdot\bx)b_p(\bx)$ (for $\delta=1$) is plotted in Figure \ref{all_potentials}(b).
 \item [$\bullet$] Top panels of Figures \ref{E_delta_zz} and \ref{E_kpar_zz}, and all panels of Figure \ref{E_L_comparison} ($\mathcal{C}-$ symmetry breaking):
 The honeycomb structured media with $\mathcal{C}-$ symmetry broken for $\delta>0$ is $A(\bx)+\delta\eta(\delta\bk_2\cdot\bx)B(\bx)$, where $A(\bx)=\e_0(\bx)I$, $B(\bx)=b_c(\bx)\sigma_2$ and $\eta(\zeta)=\tanh(\zeta)$.
 Here, $\e_0(\bx)$ is given in \eqref{unperturbed_potential} (with $\eps=1$) and
  \begin{equation}
 \label{C_breaking_potential_b}
  b_c(\bx)= \cos(\bk_1\cdot\bx) + \cos(\bk_2\cdot\bx) + \cos((\bk_1+\bk_2)\cdot\bx),
 \end{equation}
 The potentials $\e_0(\bx)$ and $\delta\eta(\delta\bk_2\cdot\bx)b_c(\bx)$ (for $\delta=1$) are shown in Figure \ref{all_potentials} panels (a) and (c), respectively.
\end{enumerate}

\section{Complex-valued honeycomb structured media}\label{nonreal}

A honeycomb structured media $A(\bx)$ is generically complex. According to Corollary \ref{A-expand}, the simplest nonconstant honeycomb media containing the lowest Fourier components is of the form
\begin{equation*}
\begin{split}
A(\bx)=&a_0I+C~ e^{i\bk_1\cdot \bx}+ R C R^* ~ e^{i\bk_2\cdot \bx} + R^* C R ~e^{i(-\bk_1-\bk_2)\cdot \bx} \\
&+ C^T ~e^{-i\bk_1\cdot \bx}+ R C^T R^* ~e^{-i\bk_2\cdot \bx} + R^* C^T R ~\frac{}{}e^{i(\bk_1+\bk_2)\cdot \bx}
\end{split}
\end{equation*}
where $C$ could be any real $2\times2$ matrix and $a_0$ is a positive constant ensuring that $A(\bx)$ is positive definite.

If $C$ is symmetric, then $A(\bx)$ is real and $E_D^{\bK}~=~E_D^{\bK'}$ according to Theorem \ref{prop_conical}. For a general non-symmetric $C$, $A(\bx)$ is complex. In Figure \ref{nonreal_Dirac}, we present the lowest three dispersion slices for a non-real $A(\bx)$ with $C=\begin{pmatrix}
  -1&-1\\-2&-2
\end{pmatrix}$ and $a_0=10$. It is seen that $E_D^{\bK}\neq E_D^{\bK'}$ in this case.

\begin{figure}
\centering
\includegraphics[width=4in]{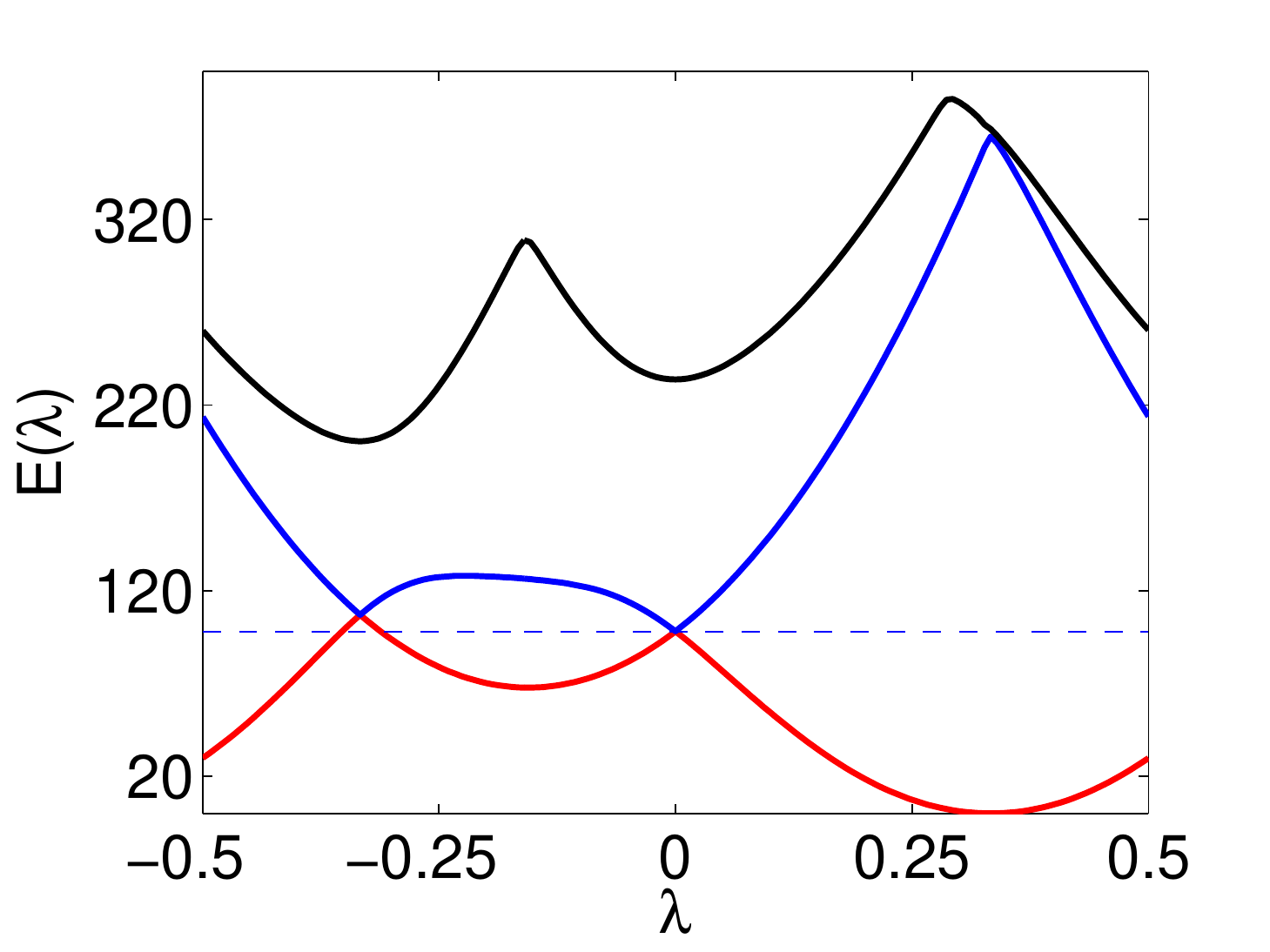}
\caption{\footnotesize Band dispersion slice along the $\bk_1-\bk_2$ direction of a complex honeycomb structured media $A(\bx)$. Two Dirac points at $\bK$ and $\bK'$ but $E_D^{\bK}\neq E_D^{\bK'}$.
}
\label{nonreal_Dirac}
\end{figure}

\section{Statement on Ethics}

\nit {\bf Funding:}\bigskip

\nit M. I.  Weinstein  was supported in part by U.S. National Science Foundation grants DMS-1412560, DMS-1620418 and  DGE-1069420 and Simons Foundation Math + X Investigator grant \#376319. \medskip

\nit J. P. Lee-Thorp was supported in part by U.S.  NSF grants DMS-1412560, DMR-1420073, and   Simons Foundation grant \#376319 (M.I. Weinstein).\medskip

\nit Y. Zhu by the Tsinghua University Initiative Scientific Research Program \# 20151080424 and NSFC grants \#11471185 and \#11871299.
\bigskip

\nit {\bf Conflict of Interest:}  The authors declare that they have no conflict of interest.

\bibliographystyle{}
\bibliography{Maxwell_HC}

\address{Courant Institute of Mathematical Sciences, \\ New York University, New York, NY, USA; leethorp@cims.nyu.edu \and
Department of Applied Physics and Applied Mathematics and Department of Mathematics, \\ Columbia University, New York, NY, USA; miw2103@columbia.edu \and
Zhou Pei-Yuan Center for Applied Mathematics, \\ Tsinghua University, Beijing, China; yizhu@tsinghua.edu.cn
}

\end{document}